\documentclass[preprint]{imsart}

\RequirePackage{amsthm,amsmath,amsfonts,amssymb}
\RequirePackage[authoryear]{natbib} 
\RequirePackage[colorlinks,citecolor=blue,urlcolor=blue]{hyperref}
\RequirePackage{graphicx}
\usepackage{xr}
\usepackage{xr-hyper}
\startlocaldefs

\usepackage{color} 
\usepackage{enumerate}
\usepackage[normalem]{ulem}
\usepackage{algorithm,algorithmicx}
\usepackage{algpseudocode}

\numberwithin{equation}{section}

\theoremstyle{plain}
\newtheorem{theorem}{Theorem}[section]
\newtheorem{corollary}[theorem]{Corollary}
\newtheorem{prop}[theorem]{Proposition}
\newtheorem{lemma}[theorem]{Lemma}
\newtheorem{claim}[theorem]{Claim}

\theoremstyle{definition}
\newtheorem{definition}[theorem]{Definition}
\newtheorem{defn}[theorem]{Definition}
\newtheorem{procedure}[theorem]{Procedure}

\theoremstyle{remark}
\newtheorem{remark}[theorem]{Remark}

\newtheorem{notation}[theorem]{Notation}

\newcommand{\argmin}{\operatornamewithlimits{argmin}}

\newcommand{\Sym}{{\rm Sym}}
\newcommand{\Symp}{ {\rm Sym}^+(p)  }
\newcommand{\Esr}{{E^{(\mathcal{SR})}}}
\newcommand{\Ensr}{{E_n^{(\mathcal{SR})}}}
\newcommand{\Epsr}{{E^{(\mathcal{PSR})}}}
\newcommand{\Enpsr}{{E_n^{(\mathcal{PSR})}}}
\newcommand{\lwr}{ {\rm lwr} }
\newcommand{\fpsr}{ f^{(\mathcal{PSR})} }
\newcommand{\fnpsr}{ f_n^{(\mathcal{PSR})} }
\newcommand{\fsr}{ f^{(\mathcal{SR})} }
\newcommand{\fnsr}{ f_n^{(\mathcal{SR})} }
\newcommand{\dpsr}{ d_{\mathcal{PSR}} }
\newcommand{\dsr}{ d_{\mathcal{SR}} }
\newcommand{\Sptop}{ S_p^{\rm top} }
\newcommand{\Splwr}{ S_p^{\rm lwr} }
\newcommand{\Real}{  \mathbb{R}}
\newcommand{\Diagp}{ {\rm Diag}^+(p)}
\newcommand{\J}{{\sf J}}

\newcommand{\Fc}{\mathcal{F}}
\newcommand{\Uc}{\mathcal{U}}
\newcommand{\Vc}{\mathcal{V}}

\newcommand{\bfr}{{\Real}}

\renewcommand{\Tilde}{\widetilde}
\newcommand{\be}{\begin{equation}}
\newcommand{\ee}{\end{equation}}
\newcommand{\bearray}{\begin{eqnarray}}
\newcommand{\eearray}{\end{eqnarray}}
\newcommand{\bestar}{\begin{eqnarray*}}
\newcommand{\eestar}{\end{eqnarray*}}
\newcommand{\ben}{\begin{displaymath}}
\newcommand{\een}{\end{displaymath}}
\newcommand{\lnbrk}{\linebreak}

\renewcommand{\d}{\delta}
\newcommand{\e}{\epsilon}

\renewcommand{\l}{\lambda}

\newcommand{\noi}{\noindent}
\renewcommand{\ss}{\vspace{.1in}}

\newcommand{\ssn}{\vspace{.1in}\noindent}
\newcommand{\indenum}{\mbox{\hspace{3ex}}}
\newcommand{\Union}{\bigcup}

\newcommand{\minus}{\backslash}

\newcommand{\diam}{{\rm diam}}

\newcommand{\D}{{\mathcal D}}
\newcommand{\F}{F}   
\newcommand{\G}{{\mathcal{G}}}
\renewcommand{\H}{{\mathcal{H}}}

\newcommand{\calc}{{\mathcal C}}
\renewcommand{\L}{\Lambda}

\newcommand{\cals}{{\mathcal{S}}}
\newcommand{\tS}{\tilde{S}}
\newcommand{\T}{{\mathcal{T}}}

\newcommand{\td}{{d_M}}
\newcommand{\tg}{{\tilde{g}}}
\newcommand{\tim}{{M}}

\newcommand{\sop}{{SO(p)}}
\newcommand{\dpp}{{{\rm Diag}^+(p)}}

\newcommand{\dso}{d_{SO}} 

\newcommand{\ddpp}{d_\D}

\newcommand{\sympp}{{\rm Sym}^+(p)}

\newcommand{\pset}{\{1,2,\dots, p\}}
\newcommand{\partpset}{{\rm Part}(\pset)}

\newcommand{\fpsrp}{\fpsr}

\newcommand{\fgl}{\mathfrak{gl}}
\newcommand{\glpr}{\fgl(p,\bfr)}
\newcommand{\gl}{{\mathfrak g}_\Lambda}
\newcommand{\glp}{\gl^\perp}
\newcommand{\ad}{{\rm ad}}
\newcommand{\lsop}{\mathfrak{so}(p)}
\newcommand{\plus}{\oplus}
\newcommand{\frob}{{\rm Frob}}

\newcommand{\fr}{{\rm Fr}}
\newcommand{\so}{\mathfrak{so}}
\newcommand{\fd}{\mathfrak{d}}
\newcommand{\tge}{{\tg_e}}

\newcommand{\dstrat}{\d_{\rm strat}}

\endlocaldefs

\begin{document}

\begin{frontmatter}
\title{Averaging symmetric positive-definite matrices on the space of eigen-decompositions}
\runtitle{Averaging {SPD} matrices via eigen-decomposition}

\begin{aug}

\author[A]{\fnms{Sungkyu} \snm{Jung}\ead[label=e1]{sungkyu@snu.ac.kr}},
\author[B]{\fnms{Brian} \snm{Rooks}\ead[label=e2]{btrooks88@gmail.com}},
\author[C]{\fnms{David} \snm{Groisser}\ead[label=e3]{groisser@ufl.edu}}
\and
\author[D]{\fnms{Armin} \snm{Schwartzman}\ead[label=e4]{armins@ucsd.edu}}
\runauthor{Jung et al.}
\address[A]{Department of Statistics, Seoul National University, \printead{e1}}

\address[B]{Statistics \& Data Corporation, \printead{e2}}
\address[C]{Department of Mathematics, University of Florida, \printead{e3}}
\address[D]{Division of Biostatistics and Hal{\i}c{\i}o\u{g}lu Data Science Institute, University of California, San Diego, \printead{e4}}
\end{aug}
 
\begin{abstract}
We study extensions of Fr\'{e}chet means for random objects in the space ${\rm Sym}^+(p)$ of $p \times p$ symmetric positive-definite  matrices using the scaling-rotation   geometric framework
introduced by Jung et al. [\textit{SIAM J. Matrix. Anal. Appl.} \textbf{36} (2015) 1180-1201].
The scaling-rotation framework is designed to enjoy a clearer interpretation of the changes in random ellipsoids in terms of scaling and rotation. In this work, we formally define the \emph{scaling-rotation (SR) mean set} to be the set of Fr\'{e}chet means in ${\rm Sym}^+(p)$  with respect to the scaling-rotation distance. Since computing such means requires a difficult optimization, we also define the \emph{partial scaling-rotation (PSR) mean set} lying on the space of eigen-decompositions as a proxy for the SR mean set. The PSR mean set is easier to compute and its projection to ${\rm Sym}^+(p)$ often coincides with SR mean set. Minimal conditions are required to ensure that the mean sets are non-empty.
Because eigen-decompositions are never unique, neither are PSR means, but we give sufficient conditions for the sample PSR mean to be unique up to the action of a certain finite group. We also establish strong consistency of the sample PSR means as estimators of the population PSR mean set, and a central limit theorem.  In an application to multivariate tensor-based morphometry, we demonstrate that a two-group test using the proposed PSR means can have greater power than the two-group test using the usual affine-invariant geometric framework for  symmetric positive-definite matrices.
\end{abstract}

\begin{keyword}[class=MSC2020]
\kwd[Primary ]{62R30}
\kwd[; secondary ]{62E20}
\end{keyword}

\begin{keyword}
\kwd{scaling-rotation distance}
\kwd{statistics on manifolds}
\kwd{strong consistency}
\kwd{central limit theorem}
\end{keyword}

\end{frontmatter}



\section{Introduction}
Recently, much work has been done to advance the statistical analysis of random symmetric positive-definite (SPD) matrices. Applications in which data arise as SPD matrices include analysis of diffusion tensor imaging (DTI) data \citep{Alexander2005,Batchelor2005}, multivariate tensor-based morphometry (TBM) \citep{Lepore2008,Paquette2017}, and tensor computing \citep{Pennec2006}. In this paper we consider the setting in which we have a random sample of SPD matrices and wish to estimate a population mean. 

Location estimation is an important first step in the development of many statistical techniques. For applications in which data are SPD matrices, these techniques include two sample hypothesis testing \citep{Schwartzman2010} for comparing average brain scans from two groups of interest, principal geodesic analysis \citep{Fletcher2004} for visualizing major modes of variation in a sample of SPD matrices, and weighted mean estimation, which has useful applications in diffusion tensor processing, including fiber tracking, smoothing, and interpolation \citep{Batchelor2005,Carmichael2013}.

One of the challenges of developing methods for analyzing SPD-valued data is that the positive-definiteness constraint precludes ${\rm Sym}^+(p)$, the space of $p \times p$ SPD matrices, from being a vector subspace of ${\rm Sym}(p)$, the space of all symmetric $p\times p$ matrices. This can be easily visualized for $p=2$; the free coordinates (two diagonal elements and upper off-diagonal element) of all $2\times 2$ SPD matrices in ${\rm Sym}(2) \cong \mathbb{R}^3$ constitutes an open convex cone. Hence, conventional estimation or inferential techniques developed for data that varies freely over Euclidean space may not be appropriate for the statistical analysis of SPD matrices. With this in mind, many location estimation frameworks for ${\rm Sym}^+(p)$ have been developed in recent years, including the log-Euclidean framework \citep{Arsigny2007}, affine-invariant framework \citep{Fletcher2004,Pennec2006}, log-Cholesky framework \citep{lin2019riemannian}, and Procrustes framework \citep{Dryden2009,masarotto2019procrustes}; see \cite{feragen2017geometries} for other examples. Given a sample of SPD matrices, most of these estimation methods amount to transforming the SPD-valued observations, averaging in the space of the transformed observations, and then mapping the mean of the transformed data into ${\rm Sym}^+(p)$. For example, the log-Euclidean method maps each observation into ${\rm Sym}(p)$ via the matrix logarithm, computes the sample mean of the transformed observations, and then maps that mean into ${\rm Sym}^+(p)$ via the matrix exponential function, while the Procrustes size-and-shape method begins with averaging the Cholesky square roots of observations, and then maps the average $\hat{L}$ to ${\rm Sym}^+(p)$ as $\hat{\Sigma} = \hat{L}\hat{L}^T$, where ${A}^T$ denotes the transpose of a matrix ${A}$.


While these geometric frameworks account for the positive-definiteness constraint of ${\rm Sym}^+(p)$, it is not clear which, if any, of the log-Euclidean, affine-invariant, or Procrustes size-and-shape frameworks is most ``natural'' for describing deformations of SPD matrices. Motivated by the analysis of DTI data, a setting in which observations are SPD matrices represented as ellipsoids in $\mathbb{R}^3$, \cite{Jung2015} developed a different framework, called the scaling-rotation (SR) framework for ${\rm Sym}^+(p)$. Under this framework, the distance between SPD matrices $X$ and $Y$ is defined as the minimal amount of rotation of axes and scaling of axis lengths necessary to deform the ellipsoid associated with $X$ into the ellipsoid associated with $Y$. For this, an SPD matrix $X$ is decomposed into eigenvectors and eigenvalues, which respectively stand for rotations and scalings.
The SR framework yields interpolation curves that have desirable properties, including constant rate of rotation and log-linear scaling of eigenvalues, and is the only geometric framework (compared to the aforementioned frameworks)
to produce both pure-scaling interpolation curves and pure-rotation curves when the endpoints differ by pure scaling or pure rotation.
While interpolation approaches similar to the SR framework can be found in \cite{wang2014tracking} and \cite{collard2012anisotropy}, only the SR framework addresses the non-uniqueness of eigen-decompositions \citep{Groisser2017,Groisser2017a}. See \cite{feragen2017geometries} and \cite{feragen2020statistics} for a comparison of the SR framework with other geometric frameworks for SPD matrices.

A major complication in developing statistical procedures using the SR framework is that eigen-decompositions are not unique. For example, an SPD matrix
$X = \mbox{diag}(8,3) = \big(\begin{smallmatrix}
  8 & 0\\
  0 & 3
\end{smallmatrix}\big)$
can be eigen-decomposed into either
$$X = U_1 D_1 U_1^T, \quad U_1 = \big(\begin{smallmatrix}
                                    1 & 0 \\
                                    0 & 1
                                  \end{smallmatrix}\big), \ D_1 = \big(\begin{smallmatrix}
                                    8 & 0 \\
                                    0 & 3
                                  \end{smallmatrix}\big),$$
                                  or
$$X = U_2 D_2 U_2^T, \quad U_2 =  \big(\begin{smallmatrix}
                                    0 & -1 \\
                                    1 & 0 \\
                                  \end{smallmatrix}\big), \ D_2 = \big(\begin{smallmatrix}
                                    3 & 0 \\
                                    0 & 8 \\
                                  \end{smallmatrix}\big).$$
(There are in fact 4 distinct eigen-decompositions for $\mbox{diag}(8,3)$, if the eigenvector matrices are required to be orthogonal matrices of positive determinant.)
Write $(U_X, D_X)$ for an eigen-decomposition (a pair of eigenvector and eigenvalue matrices) of an SPD matrix $X$, and let $\mathcal{F}$ be the eigen-composition map, e.g., $\mathcal{F}(U_X, D_X) = U_XD_XU_X^T = X$ (see Definition \ref{defn:geodesic_dist}). The SR framework defines the ``distance" between $X, Y \in \Symp$ to be
$\dsr(X,Y) := \inf d_M( (U_X, D_X), (U_Y, D_Y))$, where the infimum is taken over all possible eigen-decompositions of both $X$ and $Y$, and $d_M$ is the (geodesic) distance function on the space $M(p)$ of eigen-decompositions
(see Definition \ref{def:dsr}).
$\Symp$ is a stratified space; the stratum to which $X \in \Symp$ belongs is determined by the topological structure of the fiber $\mathcal{F}^{-1}(X)$ (the set of all eigen-decompositions corresponding to $X \in \Symp$).
The scaling-rotation distance $\dsr$ fails to be  a true metric on $\Symp$, and is difficult to compute because the set we minimize over in the definition of $\dsr (X,Y)$ is a pair of these fibers (whose topology varies with the strata of $X$ and $Y$).
With these complications in mind, the goal of this paper is to establish location-estimation methods using the SR framework as a foundation for future methods that will inherit the interpretability of the framework.

If one of the well-established  geometric frameworks, such as the affine-invariant or log-Cholesky frameworks, is used, then $\Symp$ is understood as a Riemannian manifold with a Riemannian metric tensor defined on the tangent bundle. The Riemannian metric gives rise to a distance function, say $d$, and $(\Symp, d)$ is a metric space. For these metric spaces, the Fr\'{e}chet mean  \citep{Frechet1945} is a natural candidate for  a location parameter, and conditions that guarantee uniqueness of  Fr\'{e}chet means, convergence of empirical  Fr\'{e}chet means to the population counterpart, and central-limit-theorem type results, are well-known \citep[\emph{cf}.][]{Afsari2011,Bhatt2003,Bhatt2005,Bhatt2017,Huckemann2011a,Huckemann2011b,eltzner2021stability,schotz2022strong}.

But in the SR framework, since $\dsr$ is not a true metric on $\Symp$, many of the theoretical properties of Fr\'{e}chet means (if they are defined) are no longer guaranteed. Moreover, on a practical side, computing a scaling-rotation (SR) mean, defined as a minimizer over the sum of squared SR distances to observations, requires discrete optimization in general and is thus challenging to implement. As a proxy for the SR mean, we define a \emph{partial scaling-rotation (PSR) mean} on the space of eigen-decompositions; for a finite sample $X_1,\ldots,X_n \in \Symp$, the PSR mean set is the set of minimizers
\begin{equation}\label{eq:psrmean_intro}
  \argmin_{(U,D)} \frac{1}{n}\sum_{i=1}^n \left\{ \inf_{(U_X,D_X) \in \mathcal{F}^{-1}(X_i)} d_M((U_X,D_X),(U,D))\right\}^2.
\end{equation}
See Section~\ref{sec:3} for precise definitions and an iterative algorithm for computing a sample PSR mean.
The PSR means can be thought of as a special case of generalized Fr\'{e}chet means, proposed in \cite{Huckemann2011b} and studied in \cite{Huckemann2011a,huckemann2021data,schotz2019convergence,schotz2022strong}. The PSR means can be mapped to $\Symp$ (via the eigen-composition map $\mathcal{F}$), and we establish some sufficient conditions under which the PSR means are \emph{equivalent} to the SR mean. These conditions are related to the strata of $\Symp$ in which the sample and means are located.


Another artifact caused by the stratification of $\Symp$ is that the distance function $\dsr$ is not continuous on $\Symp$, and in principle we do not know whether an SR mean is well-defined. We show that the distance function $\dsr$, the cost function appeared in (\ref{eq:psrmean_intro}) for the PSR means, and their squares are \emph{lower semicontinuous}, and thus are measurable, which guarantees that both the SR and PSR mean sets are well-defined.
We also show that SR and PSR mean sets exist, under mild assumptions.

PSR means are never unique, due to the fact that eigen-decompositions are not unique. In the best case, there are $2^{p-1}p!$ elements in the PSR mean set for a $\Symp$-valued sample, corresponding to the number of distinct eigen-decompositions of any SPD matrix with no repeated eigenvalues.
As a result, if a PSR mean set $\Enpsr$ consists of exactly $2^{p-1}p!$ elements, then the corresponding $\Symp$-valued mean, $\mathcal{F}(\Enpsr)$ consists of a single element, and we may say that $\mathcal{F}(\Enpsr)$ is unique. A sufficient condition to ensure such uniqueness will be given in Section \ref{uniqueness} in terms of data-support diameter.

We also show that with only a finite-variance condition the sample PSR mean set is consistent with the population PSR mean set, in the sense of \cite{Bhatt2003}, following the now standard technique laid out in \cite{Huckemann2011b} (with modifications required to the fact that the cost function in (\ref{eq:psrmean_intro}) is not continuous).
With additional conditions, needed to ensure the equivalence between PSR mean sets and the SR mean, imposed, we conclude that the sample SR mean set is consistent with the (unique) SR mean. A type of central limit theorem for the PSR mean is also developed, in which the limiting normal distribution is defined on a tangent space of the space of eigen-decompositions. See Section~\ref{sec:4} for theoretical properties of (partial) SR means, including existence, uniqueness, and asymptotic results. Although these properties are developed to cope with the unique challenges (e.g. non-uniqueness of eigen-decompositions and the resulting stratification) coming from using the SR framework, we believe the course of our technical development will be instructive for developing statistics in other stratified Riemannian spaces.

Numerical results demonstrate the subtle difference between the SR mean and the PSR mean, and the advantage of (partial) SR means over other means defined via other geometric frameworks.
The potential advantage of the SR framework with PSR means is further demonstrated in an application to multivariate TBM for testing the shape difference in lateral ventricular structure in the brains of  pre-term and  full-term infants, using data from \cite{Paquette2017}.
In particular, an approximate bootstrap test based on PSR means is found to be more powerful than that based on the affine-invariant means of \cite{Pennec2006}.
We conclude with practical advice on the analysis of SPD matrices and a discussion of potential future directions of research.  Technical details, proofs, and additional lemmas that may be useful in other contexts, are contained in Appendix \ref{sec:proofs_in_appendix}.

\section{The Scaling-Rotation Framework}

In this section we provide a brief overview of the scaling-rotation framework
\citep{Jung2015,Groisser2017,Groisser2017a} for analyzing SPD-valued data.
The motivation for the scaling-rotation framework is intuitive: Any $X \in {\rm Sym}^+(p)$ can be identified with the ellipsoid with surface coordinates $\{ y\in \Real^p : y^TX^{-1}y =1 \}$, so a measure of distance between $X$ and $Y$ can be defined as a suitable combination of the minimum amount of rotation of axes and stretching or shrinking of axes needed to deform the ellipsoid corresponding to $X$ into the ellipsoid associated with $Y$. Since the semi-axes and squared semi-axis lengths of the ellipsoid associated with an SPD matrix are its eigenvectors and eigenvalues, respectively, this \emph{scaling-rotation distance} is computed on the space of eigen-decompositions. 

\subsection{Geometry of the Eigen-Decomposition Space}\label{eigen-decomp_geometry}

Recall that any $X \in {\rm Sym}^+(p)$ has an eigen-decomposition $X=UDU^T$, where $U \in SO(p)$, the space of $p \times p$ rotation matrices, and $D \in {\rm Diag}^+(p)$, the space of $p \times p$ diagonal matrices possessing positive diagonal entries. We denote the space of eigen-decompositions as $M(p) := SO(p) \times {\rm Diag}^+(p)$.
The Lie groups $SO(p)$ and $\Diagp$ carry natural bi-invariant Riemannian metrics $g_{SO}$ and $g_{\mathcal{D}^+}$, defined as follows.
The tangent space at $U$ of $SO(p)$ is $T_U(SO(p)) = \{AU : A \in \mathfrak{so}(p)\}$, where $\mathfrak{so}(p)$ is the space of $p\times p$ antisymmetric matrices. At an arbitrary point $U \in SO(p)$ we define $g_{SO} \mid_{U}(A_1,A_2) = -\frac{1}{2} {\rm tr}(A_1U^TA_2U^T)$ for $A_1,A_2 \in \mathfrak{so}(p)$, where ${\rm tr}(A)$ is the trace of the matrix $A$.
The tangent space $T_D (\Diagp) =\{LD: L \in {\rm Diag}(p)\}$ is canonically identified with  ${\rm Diag}(p)$, the set of $p\times p$ diagonal matrices, and we define $g_{\mathcal{D}^+} \mid_{D} (L_1,L_2) = {\rm tr} (L_1D^{-1}L_2D^{-1})$, for $L_1,L_2 \in {\rm Diag}(p)$.
Given eigen-decompositions $(U_1,D_1)$ and $(U_2,D_2)$ of SPD matrices $X_1$ and $X_2$, we measure the distance between their eigen-decompositions using the following product metric: 

\begin{definition}\label{defn:geodesic_dist}
We define the \textit{geodesic distance function} $d_M$ on $M(p)$, with a weighting parameter $k>0$, by
\begin{equation}\label{geodesic_dist_eqn}
d^2_M((U_1,D_1),(U_2,D_2)) = k d^2_{SO}(U_1,U_2) + d^2_{\mathcal{D}^+}(D_1,D_2),
\end{equation}
where $d_{SO}(U_1,U_2) = \frac{1}{\sqrt{2}}\| {\rm Log}(U_2U_1^T) \|_F$, $d_{\mathcal{D}^+}(D_1,D_2) = \| {\rm Log}(D_1) - {\rm Log}(D_2) \|_F$, and $\| . \|_F$ denotes the Frobenius norm.
\end{definition}

In Definition~\ref{defn:geodesic_dist} and (\ref{SSR_curve}) below, ${\rm Exp}(A)$ stands for the matrix exponential of $A$, and ${\rm Log}(R)$ for the principal matrix logarithm of $R$.\footnote{The principal logarithm for rotation matrices is defined on the set $\{R \in SO(p) : R$ is not an involution$\}$, a dense open subset of $SO(p)$. When there exists no principal logarithm of $R$, the notation ${\rm Log}(R)$ denotes any solution $A \in \mathfrak{so}(p)$ of ${\rm Exp}(A) = R$ satisfying that $\|{A}\|_F$ is the smallest among all such choices of $A$. For such rare cases, the geodesic (\ref{SSR_curve}) is not unique, but
$\|{\rm Log}(R)\|_F$ is well defined. 
} The weighting parameter $k$ is a fixed constant throughout.

For a geometric interpretation of the geodesic distance, note that the geodesic distance between eigen-decompositions $(U_1,D_1)$ and $(U_2,D_2)$ equals the length of the $M(p)$-valued geodesic 
\begin{equation}\label{SSR_curve}
\gamma_{(U_1,D_1),(U_2,D_2)}(t) = ({\rm Exp}(t\textrm{{\rm Log}}(U_2U_1^{-1}))U_1,{\rm Exp}(t{\rm Log}(D_2D_1^{-1}))D_1)
\end{equation}
connecting $(U_1,D_1)$ and $(U_2,D_2)$, which is a minimal-length smooth curve connecting these two points when the tangent spaces of $M(p)$ are equipped with the canonical inner product $g_M= k g_{SO}\oplus g_{{\mathcal D}^+}$.
The functions $d_{\mathcal{D}^+}(D_1,D_2)$ and $d_{SO}(U_1,U_2)$ in \eqref{geodesic_dist_eqn} have the following interpretations: $d_{\mathcal{D}^+}(D_1,D_2)$ computes the Euclidean distance between ${\rm Log}(D_1)$ and ${\rm Log}(D_2)$, while $d_{SO}(U_1,U_2)$ equals the magnitude of the rotation angle of $U_2U_1^{-1}$ when $p=2,3$.

The exponential map at $(U,D) \in M(p)$ is ${\rm Exp}_{(U,D)}: T_{(U,D)}M(p) \to M(p)$, given by
\begin{equation}\label{eq:exp_map_2.2a}
  {\rm Exp}_{(U,D)}((AU,LD)) = ({\rm Exp}(A) U, {\rm Exp} (L)D).
\end{equation}
The inverse of the exponential map at $(U,D) \in M(p)$, defined for $\mathcal{U}_{(U,D)} = \{(V,\Lambda) \in  M(p): \|{\rm Log}(VU^T)\|_F < \pi \}$, is
${\rm Log}_{(U,D)} : \mathcal{U}_{(U,D)} \to T_{(U,D)}M(p)$, and is given by
\begin{equation}\label{eq:log_map_2.2b}
{\rm Log}_{(U,D)}((V,\Lambda)) = ({\rm Log}(VU^T)U, {\rm Log}(\Lambda D^{-1})D).
\end{equation}
With the Riemannian metric $g_M = k g_{SO} \oplus g_{\mathcal{D}^+}$, the induced norm on the tangent space $T_{(U,D)}M(p)$ satisfies
$$\| (A_1U, L_1D) - (A_2U, L_2D)\|^2_{(U,D)} = \frac{k}{2}{\rm tr}((A_1-A_2)(A_1-A_2)^T) + {\rm tr}((L_1-L_2)(L_1-L_2)^T),$$
for any two tangent vectors $(A_1U, L_1D), (A_2U, L_2D) \in T_{(U,D)}M(p)$.

\subsection{Minimal Smooth Scaling-Rotation Curves and Scaling-Rotation Distance}\label{MSSR_defs}

Since eigen-decompositions are not unique, any method for computing the distance between SPD matrices using the eigen-decomposition space must take this non-uniqueness into account. To address this, \cite{Jung2015} proposed the following distance for ${\rm Sym}^+(p)$:

\begin{definition}\label{def:dsr}
Let $\mathcal{F} : M(p) \to {\rm Sym}^+(p)$
denote the eigen-composition map $\mathcal{F}(U,D) = UDU^T$,
and for any $X \in {\rm Sym}^+(p)$, let $\mathcal{F}^{-1}(X)$ denote the set of eigen-decompositions of $X$. The \textit{scaling-rotation distance} between $X \in {\rm Sym}^+(p)$ and $Y \in {\rm Sym}^+(p)$ is
\[
\dsr(X,Y) = \inf_{\substack{(U_X,D_X)\in\mathcal{F}^{-1}(X), \\ (U_Y,D_Y)\in\mathcal{F}^{-1}(Y)}} d_M((U_X,D_X),(U_Y,D_Y)).
\]
Eigen-decompositions $(U^*_X,D_X^*) \in \mathcal{F}^{-1}(X)$ and $(U_Y^*,D_Y^*) \in \mathcal{F}^{-1}(Y)$ form a \textit{minimal pair} if $d_M((U^*_X,D_X^*),(U_Y^*,D_Y^*))=d_{\mathcal{SR}}(X,Y)$.
\end{definition}

\begin{remark}
Since the sets $\mathcal{F}^{-1}(X)$ and $\mathcal{F}^{-1}(Y)$ are compact for any $X,Y \in {\rm Sym}^+(p)$, 
there will always be a pair of eigen-decompositions of $X$ and $Y$ that form a minimal pair.
\end{remark}

\begin{remark}
The function $d_{\mathcal{SR}}$ is not a true metric on ${\rm Sym}^+(p)$ since there are instances in which the triangle inequality fails. It is a semi-metric and invariant under simultaneous matrix inversion, uniform scaling and conjugation by a rotation matrix \citep[][Theorem 3.11]{Jung2015}. When restricted to the subset of SPD matrices which possess no repeated eigenvalues, $d_{\mathcal{SR}}$ is a true metric \citep[][Theorem 3.12]{Jung2015}. 
\end{remark}

For SPD matrices $X, Y$ and their eigen-decompositions $(U_X,D_X) \in \mathcal{F}^{-1}(X)$, $(U_Y,D_Y)\in \mathcal{F}^{-1}(Y)$, one can create a smooth scaling-rotation (SSR) curve on ${\rm Sym}^+(p)$ connecting $X$ and $Y$ as $\chi_{X,Y}(t) = \mathcal{F}(\gamma_{(U_X,D_X),(U_Y,D_Y)}(t))$, where $\gamma_{(U_X,D_X),(U_Y,D_Y)}(t)$ is a minimal-length geodesic curve defined in (\ref{SSR_curve}).
If one considers the family of all possible geodesics in $(M(p),g_M)$ from ${\mathcal F}^{-1}(X)$ to ${\mathcal F}^{-1}(Y)$,  the scaling-rotation distance equals the length of the shortest geodesics in that family. By definition, the shortest geodesic (which may not be uniquely defined) connects a minimal pair $(U^*_X,D_X^*) \in \mathcal{F}^{-1}(X)$ and $(U_Y^*,D_Y^*) \in \mathcal{F}^{-1}(Y)$.
%
Computing $d_{\mathcal{SR}}(X,Y)$ for any dimension $p$ is straightforward when $X$ and $Y$ both have no repeated eigenvalues, since  $X$ and $Y$ then both have finitely many eigen-decompositions and therefore finitely many connecting SSR curves, or when one of $X,Y$ is a scaled identity matrix. Formulas for computing $d_{\mathcal{SR}}(X,Y)$ for all possible eigenvalue-multiplicity combinations of arguments $X$ and $Y$ are provided in \cite{Groisser2017} for $p=2,3$.

\subsection{The stratification of $\Symp$ and fibers of the eigen-composition map}\label{sec:fiber}

The space $\Symp$ is naturally stratified by the eigenvalue-multiplicity types. We will use the notation $S_p^{\rm top}$ to denote the subset of SPD matrices which have no repeated eigenvalues (the superscript ``top'' refers to the ``top stratum''). We also use the notation $S_p^{\rm lwr} := \Symp \setminus S_p^{\rm top}$, for the union of all ``lower'' strata, and $S_p^{\rm bot} \subset S_p^{\rm lwr} $ denotes the set of SPD matrices with equal eigenvalues.
The eigenvalue-multiplicity stratification of $\Symp$ is equivalent to the fiber-type stratification of $\Symp$; SPD matrices $X,Y\in \Symp$ are in the same stratum if $\mathcal{F}^{-1}(X)$ and $\mathcal{F}^{-1}(Y)$ are diffeomorphic, as we elaborate below.

Let ${\rm Part}\{1,\ldots,p\}$ be the set of partitions of $\{1,\ldots,p\}$. Recall that ${\rm Part}\{1,\ldots,p\}$ is partially ordered by the refinement relation, with ``$\textsf{J}_1 \le \textsf{J}_2$,'' meaning that $\textsf{J}_2 \in {\rm Part}\{1,\ldots,p\}$ is a refinement of $\textsf{J}_1 \in {\rm Part}\{1,\ldots,p\}$. As an example, for $p = 2$, there are only two partitions $\textsf{J}_{\rm top} := \{ \{1\}, \{2\}\}$ and $\textsf{J}_{\rm bot} := \{ \{1, 2\}\}$, and $\textsf{J}_{\rm bot} \le \textsf{J}_{\rm top}$. 

Each $D \in \Diagp$ naturally determines a partition $\textsf{J}_D \in  {\rm Part}\{1,\ldots,p\}$, depending on which diagonal elements are equal. The group $SO(p)$ acts on $\Symp$ on the left via $(U, X) \mapsto UXU^T$.
 For $D \in \Diagp$, the stabilizer subgroup $G_D$
under the $SO(p)$ action on $\Symp$ is $G_D := \{R \in SO(p): RDR^T = D\}$.
 The stabilizer $G_D$ depends only on $\textsf{J}_D$, and generally has more than one connected component. Write $G_D^0 \subset G_D$ for the connected component of $G_D$ containing the identity.

Let $\mathcal{S}_p$ be the group of permutations of $\{1,2,\ldots,p\}$. 
For a permutation $\pi \in \mathcal{S}_p$ and $D \in \Diagp$, the natural left action of  $\mathcal{S}_p$ on $\Diagp$ is denoted by $\pi \cdot D$, and is given by permuting the diagonal entries of $D$.
 Write the matrix of the linear map ``$\pi\cdot$'' by $P_\pi \in \mathbb{R}^{p\times p}$, where the entries of $P_\pi$ are $(P_\pi)_{ij} = \delta_{i,\pi(j)}$, so that $\pi\cdot D = P_\pi D P_\pi^T$.
We call a $p\times p$ matrix $P$ a \emph{signed-permutation matrix} if for some $\pi \in \mathcal{S}_p$ the entries of $P$ satisfy $P_{ij} = \pm \delta_{i,\pi(j)}$. We call such $P$ \emph{even} if ${\rm det}(P) = 1$.
Each such $P$ thus represents a permutation of coordinates in $\mathbb{R}^p$, combined with an even number of sign changes.
The set $\mathcal{G}(p)$ of all such \emph{even signed-permutation matrices} has exactly $2^{p-1}p!$ elements, and is a matrix subgroup of $SO(p)$.
The natural  left-action of  $\mathcal{G}(p)$ on $M(p)$ is given by
\begin{equation}\label{eq:group_action}
  h\cdot (U,D):= (U  h^{-1}, h \cdot D),
\end{equation}
where $h \in \mathcal{G}(p)$ and $h \cdot D := hDh^{-1}$. The action of $h$ on $(U,D)$
represents
the simultaneous permutation (by the unsigned permutation associated with $h$) of columns of $U$ and diagonal elements of $D$, and the sign-changes of the columns of $U$.
The identity element of $\mathcal{G}(p)$ is $I_p$.

It is shown in \cite{Jung2015} that the fiber $\mathcal{F}^{-1}(X)$---that is, the set of eigen-decompositions of $X$---is characterized with any $(U,D) \in \mathcal{F}^{-1}(X)$ by
\begin{equation}\label{eq:fiber}
\mathcal{F}^{-1}(X) = \{h \cdot (UR, D): R \in G_D^0, h \in \mathcal{G}(p)\}.
\end{equation}
Thus,
the left-action of $\mathcal{G}(p)$ on $M(p)$ is fiber-preserving.

The structure of fiber $\mathcal{F}^{-1}(X)$ depends on the stratum to which $X$ belongs.
If $X \in \Sptop$, then for any eigen-decomposition $(U,D) \in \mathcal{F}^{-1}(X)$, we have $G_D^0 = \{I_p\}$ and the orbit
\begin{equation}\label{eq:orbit}
  \mathcal{G}(p) \cdot (U,D) =  \{h\cdot (U,D): h \in \mathcal{G}(p)\}
\end{equation}
 is \emph{exactly} the set of eigen-decompositions of $X$. Intuitively, any eigen-decomposition of $X \in \Sptop$ can be obtained from any other by a sign-change of eigenvectors and a simultaneous permutation of  eigenvectors and eigenvalues.
In contrast, if $X \in S_p^{\rm bot}$ (i.e., $X$ is a scaled identity matrix), then  $G_X^0 = G_X = SO(p)$ and $h \cdot X = X$ for all $h \in \mathcal{G}(p)$, thus the fiber of $\mathcal{F}$ at $X$ is $\mathcal{F}^{-1}(X) = SO(p) \times \{X\}$.
 A complete characterization of fibers of $\mathcal{F}$ for other lower strata can be found in \cite{Groisser2017}.





\section{Location estimation under the scaling-rotation framework}\label{sec:3}

\subsection{Fr\'{e}chet mean}\label{frechet_mean_estimation}

An approach often used for developing location estimators for non-Euclidean metric spaces is Fr\'{e}chet mean estimation \citep{Frechet1945}, in which estimators are derived as minimizers of a metric-dependent sample mean-squared error.

\begin{definition}\label{def:3.1}
Let $M$ be a metric space with metric $\rho$ and suppose that $X, X_1, \dotsc, X_n$ are i.i.d. $M$-valued random variables with induced probability measure $P$ on $M$. The \textit{population Fr\'{e}chet mean set} is
\[\argmin_{C \in M} \int_{M} \rho^2(X,C)P(dX). \]
The \textit{sample Fr\'{e}chet mean set} is
\[\argmin_{C \in M} \frac{1}{n}\sum_{i=1}^n \rho^2(X_i,C). \]
\end{definition}


Examples of location estimators that have been developed for ${\rm Sym}^+(p)$ using the sample Fr\'{e}chet mean estimation framework include the log-Euclidean mean \citep{Arsigny2007}, affine-invariant mean \citep{Fletcher2004,Pennec2006}, Procrustes size-and-shape mean \citep{Dryden2009}, and the log-Cholesky average \citep{lin2019riemannian}.  Below, we allow ourselves to use the ``Fr\'{e}chet mean'' terminology of Definition~\ref{def:3.1} when the metric space $(M,\rho)$ is replaced by the semi-metric space $({\rm Sym}^+(p), d_{\mathcal SR})$.

\subsection{Scaling-rotation means}\label{SR_sample_mean}\label{sec:SR_sample_mean}

We now define the population and sample scaling-rotation mean sets, consisting of the Fr\'{e}chet means of SPD matrices under the scaling-rotation framework.
Let $P$ be a Borel probability measure on $\Symp$, and $X_1,\ldots,X_n$ be deterministic data points in $\Symp$. Note that Borel measures on $\Symp$ include both discrete and absolutely continuous measures, as well as mixtures of those.

\begin{definition}\label{def:esr_ensr}
The \emph{population scaling-rotation (SR) mean set} with respect to $P$ is
\begin{equation}\label{eq:esr_definition}
  \Esr := \argmin_{S \in {\rm Sym}^+(p)} \fsr(S),\quad \fsr(S) =  \int_{{\rm Sym}^+(p)} \dsr^2(X, S) P(dX).
\end{equation}
Given $X_1, \dotsc, X_n \in {\rm Sym}^+(p)$, the \emph{sample SR mean set} is
\[ \Ensr := \argmin_{S \in {\rm Sym}^+(p)} \fnsr(S), \quad \fnsr(S) =  \frac{1}{n}\sum_{i=1}^n d^2_{\mathcal{SR}}(X_i,S). \]
\end{definition}

Since,  for some $S \in \Symp$, the function $\dsr(\cdot, S): \Symp \to \Real $ has discontinuities (see Appendix~\ref{sec:discontinuityofdsr}), we must address whether the objective function $\fsr$ of (\ref{eq:esr_definition}) is well-defined. We defer this discussion to Section~\ref{sec:LSC}.

Locating a sample SR mean can be recast as solving a difficult constrained optimization problem on $M(p)^n$ since
\begin{equation}\label{eigen-decomp_obj_function}
\frac{1}{n}\sum_{i=1}^n d^2_{\mathcal{SR}}(X_i,S) = \frac{1}{n}\sum_{i=1}^n d^2_M((U^*_i,D^*_i),(U_S^{*,i},D_S^{*,i})),
\end{equation}
where  for each $i=1, \dotsc, n$,  $(U^*_i,D^*_i) \in \mathcal{F}^{-1}(X_i)$ and $(U_S^{*,i},D_S^{*,i}) \in \mathcal{F}^{-1}(S)$ are an arbitrary minimal pair.
Due to the non-uniqueness of eigen-decompositions, there may be many pairs of eigen-decompositions of $X_i$ and $S$ which form a minimal pair.

However, when $S \in S_p^{\rm top}$ the scaling-rotation distance simplifies to
\begin{equation}\label{SR_distance_simplification}
d_{\mathcal{SR}}(X,S) = \inf_{(U_X,D_X) \in \mathcal{F}^{-1}(X)} d_M((U_X,D_X),(U_S,D_S)),
\end{equation}
where $(U_S,D_S)$ is \emph{any} eigen-decomposition of $S$. 
In this case, $d_{\mathcal{SR}}(X,S)$ is easier to compute since one can select an arbitrary eigen-decomposition $(U_S,D_S)$ of $S$ and then determine the infimum of the distances between $(U_S,D_S)$ and the eigen-decompositions of $X$.
If $S$ has repeated eigenvalues (or, equivalently, $S$ is in a lower stratum), this simplification does not hold in general; there may be no eigen-decomposition of $S$ that is at minimal distance from $\mathcal{F}^{-1}(X_i)$ simultaneously for all $i$.

From the simplification in (\ref{SR_distance_simplification}), we propose to solve for minimizers of the simplified objective function
\[ (U,D) \mapsto \frac{1}{n} \sum_{i=1}^n \inf_{(U_X,D_X) \in \mathcal{F}^{-1}(X_i)} d_M^2 ((U_X,D_X),(U,D)), \]
where the argument $(U,D)$ is an arbitrarily chosen eigen-decomposition of the argument $S$ from (\ref{eigen-decomp_obj_function}).



To formally define this simplified optimization problem, we first define the following measure of distance between an SPD matrix and a given eigen-decomposition of another SPD matrix:

\begin{definition}
The \textit{partial scaling-rotation (PSR) distance} is the map $d_{\mathcal{PSR}} : {\rm Sym}^+(p) \times M(p) \to [0,\infty)$ given by
\[ d_{\mathcal{PSR}}(X,(U,D)) = \inf_{(U_X,D_X) \in \mathcal{F}^{-1}(X)} d_M((U_X,D_X),(U,D)). \]
\end{definition}

It can be checked from the definitions that for any $X \in {\rm Sym}^+(p)$ and any $(U,D) \in M(p)$
\begin{equation}\label{SR_PSR_ineq}
d_{\mathcal{SR}}(X,\mathcal{F}(U,D)) \le d_{\mathcal{PSR}}(X,(U,D)),
\end{equation}
and by  (\ref{SR_distance_simplification}), the equality in (\ref{SR_PSR_ineq}) holds  if $\mathcal{F}(U,D) \in \Sptop$. 

\begin{definition}\label{def:epsr_enpsr}
The population and sample \emph{partial scaling-rotation (PSR) mean sets} are subsets of $M(p)$ and are defined respectively by $\Epsr  := \argmin_{(U,D) \in M(p)}\fpsr(U,D)$ and
  $\Enpsr  := \argmin_{(U,D) \in M(p)} \fnpsr(U,D)$, where
\begin{align}
    \fpsr(U,D) &=  \int_{{\rm Sym}^+(p)} \dpsr^2(X,(U,D)) P(dX), \label{eq:Epsr} \\
     \fnpsr(U,D) &=  \frac{1}{n}\sum_{i=1}^n d^2_{\mathcal{PSR}}(X_i,(U,D)). \nonumber
\end{align}
\end{definition}

In Sections~\ref{sec:LSC} and \ref{sec:existence} we show that for any Borel probability measure on $\Symp$, the population mean set $\Epsr$ is well-defined and  non-empty.
 There, we also show that both $\Enpsr$ and $\Ensr$ are non-empty for any sample $X_1,\ldots,X_n$. An iterative algorithm to compute a sample PSR mean is given in Section~\ref{sec:3.3}.

The PSR means lie in $M(p)$ and can be mapped to ${\rm Sym}^+(p)$ via the eigen-composition map. 
The sample PSR mean set can be thought of as yielding an approximation of the sample SR mean set, and it is of interest to know when the two sets are ``equivalent''.
The theorem below provides conditions under which $\Enpsr \subset M(p)$ is equivalent to $\Ensr \subset {\rm Sym}^+(p)$ in the sense that every member of $\Enpsr$ is an eigen-decomposition of a member of $\Ensr$ and vice-versa.
Define  $M^{\rm top}(p) = \mathcal{F}^{-1}( S_p^{\rm top} )$, the subset of $M(p)$ consisting of all elements $(U,D) \in M(p)$ in which the diagonal elements of $D$ are not repeated.

\begin{theorem}\label{thm:SRvsPSRequivalence}
Let $\Enpsr$ and $\Ensr$ be defined with deterministic data points $X_1,\ldots,X_n \in \Symp$.
\begin{itemize}
  \item[(a)] $E_n^{(\mathcal{PSR})} \supset \mathcal{F}^{-1}(E_n^{(\mathcal{SR})} \cap S_p^{\rm top} )$.
  \item[(b)] If $E_n^{(\mathcal{SR})} \cap S_p^{\rm top} \neq \emptyset$, then $\mathcal{F}(E_n^{(\mathcal{PSR})}) \subset E_n^{(\mathcal{SR})}$ and $E_n^{(\mathcal{PSR})} \subset \mathcal{F}^{-1}(E_n^{(\mathcal{SR})})$.
  \item[(c)] If $E_n^{(\mathcal{SR})} \cap S_p^{\rm top} \neq \emptyset$ and
        $E_n^{(\mathcal{PSR})} \subset M^{\rm top}(p)$, then
         $E_n^{(\mathcal{PSR})} = \mathcal{F}^{-1}(E_n^{(\mathcal{SR})} \cap S_p^{\rm top} )$.
\end{itemize}
In particular, since $\Ensr \neq \emptyset$ (see Corollary~\ref{thm:existence_sampleSR_PSRmeans}),
parts (a) and (b) together imply that
 if
 $E_n^{(\mathcal{SR})} \subset S_p^{\rm top}$,
then
$$E_n^{(\mathcal{PSR})}=
\mathcal{F}^{-1}(E_n^{(\mathcal{SR})})\ \mbox{ and } \ \mathcal{F}(E_n^{(\mathcal{PSR})}) = E_n^{(\mathcal{SR})}.$$
Moreover, the statements above hold when $\Ensr$ and $\Enpsr$ are replaced by $\Esr$ and $\Epsr$, respectively, provided that $\fpsr(U,D) < \infty$ for some $(U,D) \in M(p)$.
\end{theorem}

The previous theorem suggests that in many realistic situations, there may be no cost to using the PSR means in place of the SR means, which are more difficult to compute in practice.
If minimizing $\fsr$  or $\fnsr$ over $S_p^{\rm lwr}$ (the union of lower strata of $\Symp$)
is feasible, then the following result can be used to tell whether a PSR mean is equivalent to an SR mean.

For the rest of the paper, we generally use the notation $m$ rather than $(U,D)$ for an arbitrary element of $M(p)$  if there is no explicit need for writing out the eigenvector and eigenvalue matrices separately.

\begin{theorem}\label{thm:how_to_tell_whether_PSR=SR}
Let $m^{\mathcal{PSR}} \in M(p)$ be a PSR mean with respect to a probability measure $P$ on $\Symp$.
\begin{itemize}
  \item[(a)] If $\fsr(\mathcal{F}(m^{\mathcal{PSR}})) \le \min_{S \in S_p^{\rm lwr}} \fsr(S)$, then $\mathcal{F}(m^{\mathcal{PSR}}) \in \Esr$. 
  \item[(b)] If $\fsr(\mathcal{F}(m^{\mathcal{PSR}})) > \min_{S \in S_p^{\rm lwr}} \fsr(S)$, then $\mathcal{F}(m^{\mathcal{PSR}}) \notin \Esr$ and $\Esr \subset S_p^{\rm lwr}.$
\end{itemize}
Let $\hat{m}^{\mathcal{PSR}} \in M(p)$ be a sample PSR mean with respect to a given sample $X_1,\ldots,X_n \in \Symp$. Similarly to the statements above,
\begin{itemize}
  \item[(c)] If   $\fnsr(\mathcal{F}(\hat{m}^{\mathcal{PSR}})) \le \min_{S \in S_p^{\rm lwr}} \fnsr(S)$, then $\mathcal{F}(\hat{m}^{\mathcal{PSR}}) \in \Ensr$.
  \item[(d)] If $\fnsr(\mathcal{F}(\hat{m}^{\mathcal{PSR}})) > \min_{S \in S_p^{\rm lwr}} \fnsr(S)$, then $\mathcal{F}(\hat{m}^{\mathcal{PSR}}) \notin \Ensr$ and $\Ensr \subset S_p^{\rm lwr}.$
\end{itemize}
\end{theorem}

We remark that for $p = 2$, $S_p^{\rm lwr} = \{ c I_2 : c > 0\}$ and the function $\fnsr$ can be efficiently minimized over $S_p^{\rm lwr}$ by a one-dimensional numerical optimization.

 A key condition to ensure the equivalence of the SR means to PSR means is that all SR means have no repeated eigenvalues (i.e., $E^{(\mathcal{SR})} \subset S_p^{\rm top}$), which in fact depends on the distribution $P$.
Below, we give a sufficient condition
for 
$\Esr \subset S_p^{\rm top}$ or $E_n^{(\mathcal{SR})} \subset S_p^{\rm top}$.
Let
$\delta: {\rm Sym}^{+} (p) \to [0,\infty)$ be
$\delta(S) = \inf\{d_{\mathcal{SR}}(S, S'): S' \in S_p^{\rm lwr}\}$.
Thus, $\delta(S)$ is a ``distance'' from $S$ to lower strata of ${\rm Sym}^{+}(p)$.
 (Because $S_p^{\rm lwr}$ is closed, $\delta(S) > 0 $ for any $S \in S_p^{\rm top}$.)

\begin{theorem}\label{thm:avoid_low_strata_SR}
Let $X$ be a $ {\rm Sym}^+(p)$-valued random variable  with distribution $P$.
Assume that there exists $S_0 \in S_p^{\rm top}$ and $r \in (0, \delta(S_0)/3)$ such that
$$P( d_{\mathcal{SR}}(X,S_0) \le r ) = 1.$$
Then $\Esr \subset S_p^{\rm top}$.

Similarly, let $X_1, \dotsc, X_n \in {\rm Sym}^+(p)$, and assume that there exists $S_0 \in S_p^{\rm top}$ and $r \in (0, \delta(S_0)/3)$ satisfying
$d_{\mathcal{SR}}(X_i,S_0) \le r$
for $i = 1,\ldots,n$. Then $E_n^{(\mathcal{SR})} \subset S_p^{\rm top}$.
\end{theorem}
The condition of Theorem \ref{thm:avoid_low_strata_SR} requires that the sample lie in a ball that is sufficiently far from   lower strata of ${\rm Sym}^{+}(p)$, but this condition is by no means necessary.  The condition, however, cannot be replaced by the weaker condition that all data lie in $S_p^{\rm top}$; there are examples in which this weaker condition is met, but $\Ensr \subset S_p^{\rm lwr}$. In Section \ref{sec:Appendix_numerical_examples}, we provide numerical examples where the PSR means are equivalent (or not equivalent) to the SR means.

\subsection{Sample PSR Mean Estimation Algorithm}\label{sec:3.3}

Given a sample $X_1, \dotsc, X_n \in \Symp$, we propose an algorithm for approximating a member of $\Enpsr$, that is to find a minimizer of $\fnpsr$. The algorithm is similar to the generalized Procrustes algorithm \citep{Gower1975}.


\begin{procedure}[Sample PSR Mean]\label{proc:estimation}
Set tolerance $\varepsilon > 0$ and pick initial guess $(\hat{U}^{(0)},\hat{D}^{(0)}) \in M(p)$. Set $j = 0$.
\begin{enumerate}
\item[Step 1.] For $i=1, \dotsc, n$, find $(U^{(j)}_i,D^{(j)}_i) \in \mathcal{F}^{-1}(X_i)$ that has the smallest geodesic distance from $(\hat{U}^{(j)},\hat{D}^{(j)})$.
\item[Step 2.] Compute $(\hat{U}^{(j+1)},\hat{D}^{(j+1)}) \in \argmin_{(U,D) \in M(p)} \frac{1}{n} \sum_{i=1}^n d^2_M((U^{(j)}_i,D^{(j)}_i),(U,D)).$
\end{enumerate}
 If $|\fnpsr(\hat{U}^{(j+1)},\hat{D}^{(j+1)}) - \fnpsr(\hat{U}^{(j)},\hat{D}^{(j)})|>\varepsilon$, increment $j$ and repeat Steps 1 and 2.
 Otherwise, $(\hat{U}_{\mathcal{PSR}},\hat{D}_{\mathcal{PSR}}) = (\hat{U}^{(j+1)},\hat{D}^{(j+1)})$ is the approximate sample PSR mean produced by this algorithm,  given the tolerance $\varepsilon$ and initial guess $(\hat{U}^{(0)},\hat{D}^{(0)})$.
\end{procedure}
%

\begin{remark}
The above procedure will always terminate since $\fnpsr(U,D) \ge 0$ for any $(U,D) \in M(p)$ and $\fnpsr(\hat{U}^{(j)},\hat{D}^{(j)}) \ge \fnpsr(\hat{U}^{(j+1)},\hat{D}^{(j+1)})$ for any $j \ge 0$.
\end{remark}

If $X_i$ lies in $S_p^{\rm top}$, performing  Step 1 will simply require searching over the $2^{(p-1)}p!$ distinct eigen-decompositions 
of $X_i$ to find one that attains the minimal geodesic distance from $(\hat{U}^{(j)},\hat{D}^{(j)})$.
Solving for the minimizing eigen-decomposition of $X_i$ is also easy if $X_i$ is a scaled identity matrix ($X_i \in S_p^{\rm bot}$), since the fact that $X_i = cI_p = U(cI_p)U^T$ for any $U \in SO(p)$ implies that $(\hat{U}^{(j)}, cI_p)$ will be the eigen-decomposition of $X_i$ with minimal geodesic distance from $(\hat{U}^{(j)},\hat{D}^{(j)})$.
Determining the minimizing eigen-decomposition of $X_i$ when $p=3$ and $X_i$ has two distinct eigenvalues can be done by comparing  three closed-form expressions, as
described in \cite{Groisser2017}. 
For $p > 3$, there are no known corresponding closed-form expressions for determining a minimizing eigen-decomposition of $X_i \in S_p^{\rm lwr} \setminus S_p^{\rm bot}$.

The optimization problem over $M(p)$ in Step 2 can be divided into separate minimization problems over ${\rm Diag}^+(p)$ and $SO(p)$:
\begin{align}
&\hat{D}^{(j+1)} =  \argmin_{D \in {\rm Diag}^+(p)} \frac{1}{n}\sum_{i=1}^n \| {\rm Log}(D_i^{(j)})-{\rm Log}(D) \|_F^2, \nonumber \\
&\hat{U}^{(j+1)} \in \argmin_{U \in SO(p)} \frac{1}{n}\sum_{i=1}^n \|{\rm Log}(U^{(j)}_iU^{-1})\|^2_F. \nonumber
\end{align}
The solution $\hat{D}^{(j+1)}$ is uniquely given by 
$\hat{D}^{(j+1)} = {\rm Exp} \{\frac{1}{n}\sum_{i=1}^n {\rm Log}(D_i^{(j)})\}$,
while $\hat{U}^{(j+1)}$ usually must be approximated via numerical procedures.
It is shown in \cite{Manton2004} that when the rotation matrices $U^{(j)}_1, \dotsc, U^{(j)}_n$ lie within a geodesic ball of radius $\frac{\pi}{2}$, there is a unique minimizer $\hat{U}^{(j+1)}$, and this minimizer can be approximated by a globally convergent gradient descent algorithm on $(SO(p),g_{SO})$. It is highly unlikely that one would be able to de-couple estimation of the eigenvalue and eigenvector means in this manner while solving for a sample SR mean.
 
\section{Theoretical Properties of Scaling-Rotation Means}\label{sec:4}

\subsection{Lower semicontinuity and other properties of $\dsr$ and $\dpsr$}\label{sec:LSC}
One of the complications in using the SR framework is that the symmetric function $\dsr$ is not continuous in either variable (see Appendix~\ref{sec:discontinuityofdsr} for an example).  Unfortunately, $\dpsr$ is also not continuous at every point of ${\rm Sym}^+(p) \times M(p)$, as illustrated by the following example. Let $X(\varepsilon) := \textrm{diag}(e^\varepsilon, e^{-\varepsilon})$ and $(U,D) = (R(\theta),I_2)$, where $R(\theta)$ is the $2 \times 2$ rotation matrix corresponding to a counterclockwise rotation by angle $\theta$. Then for any $\varepsilon \ne 0$ and $0<|\theta|<\pi/4$,
\[ d_{\mathcal{PSR}}(X(\varepsilon),(U,D)) = (k\theta^2 + 2\varepsilon^2)^{1/2}, \]
which implies that $d_{\mathcal{PSR}}(X(\varepsilon),(U,D)) \rightarrow \sqrt{k}|\theta|$ as $\varepsilon \rightarrow 0$. Since $d_{\mathcal{PSR}}(X(0),(U,D)) = 0$, it follows that $d_{\mathcal{PSR}}$ is not continuous at $(I_2,(U,D))$, and therefore $d_{\mathcal{PSR}}$ is not continuous on ${\rm Sym}^+(p) \times M(p)$. Nevertheless, $d_{\mathcal{PSR}}$ is continuous with respect to the second variable in $M(p)$, and is jointly continuous on $S_p^{\rm top} \times M(p)$, as we state below.

\begin{lemma}\label{lem:dpsr_continuity}
\begin{itemize}
  \item[(a)] $d_{\mathcal{PSR}}$ is continuous on $S_p^{\rm top} \times M(p)$.
  \item[(b)] For each $S \in \Symp$, the function $\dpsr(S,\cdot): M(p) \to [0,\infty)$ is Lipschitz, with Lipschitz-constant 1. That is, for all $m_1,m_2 \in M(p)$,
$$|d_{\mathcal{PSR}}(S,m_1) - d_{\mathcal{PSR}}(S,m_2)| < d_M(m_1,m_2).$$
In particular, $\dpsr(S,\cdot)$ is uniformly continuous for each $S$.
\end{itemize}
\end{lemma}

Since both   $\dsr$ and $\dpsr$ are not continuous, in principle we do not know yet whether the integrals of $\dsr^2(\cdot, \Sigma)$ and $\dpsr^2(\cdot, (U,D))$, for $\Sigma \in \Symp$ and $(U,D)\in M(p)$, in Definitions~\ref{def:esr_ensr} and \ref{def:epsr_enpsr}, are well defined.
A related question is: under which conditions do the population (partial) scaling-rotation means exist? A key observation in answering these questions is that these functions
$\dsr^2(\cdot,\Sigma)$ and $\dpsr^2(\cdot; (U,D))$
 are \emph{lower semicontinuous} (LSC).
(Recall that a function $f: \mathcal{X} \to \Real$, where $\mathcal{X}$ is a topological space, is LSC at a point $x_0 \in \mathcal{X}$ if for all $\epsilon >0$, there exists an open neighborhood $\Uc$ of $x_0$ such that $f(x) > f(x_0) - \epsilon$ for all $x \in \Uc$. If $f$ is LSC at each $x_0 \in \mathcal{X}$, we say that $f$ is LSC.)

\begin{definition}\label{def:lsc}
Let $\mathcal{X}$ be a topological space and $\mathcal{Y}$ be a set, and let $f: \mathcal{X} \times \mathcal{Y} \to \Real$.
\begin{itemize}
  \item[(a)] We say that $f$ is \emph{LSC in its first variable, uniformly with respect to its second variable}, if for all $x_0 \in \mathcal{X}$ and $\epsilon>0$, there exists an open neighborhood $\Uc$ of $x_0$ such that
        \begin{equation}\label{lsc_unif2'}
f(x,y)> f(x_0,y)-\epsilon \ \ \ \mbox{for all $x\in \Uc$ and all $y\in \mathcal{Y}$.}
        \end{equation}
  \item[(b)] If $\mathcal{Y}$ is also a topological space, we say that {\em $f$ is LSC in its first variable,
\underline{\em  locally} uniformly with respect to its second variable},
if every $y_0\in \mathcal{Y}$ has an open neighborhood $\Vc$ such that $f|_{\mathcal{X}\times \Vc}$
is LSC in the  first variable, uniformly with respect to the second.  If $\mathcal{Y}$ is locally compact,
this property is equivalent to:  for every compact set $K \subset \mathcal{Y}$, $f|_{\mathcal{X}\times K}$
is LSC in the  first variable, uniformly with respect to the second.
\end{itemize}
\end{definition}

Any finite-dimensional manifold (in particular, $M(p)$) is locally compact.

\begin{theorem}\label{prop:lsc}
\begin{itemize}
  \item[(a)] Let $S_0 \in \Symp$, and $m_0 \in M(p)$. Then the functions $\dsr^2(\cdot, S_0)$ and $\dpsr^2(\cdot, m_0)$ and their square-roots are LSC.
  \item[(b)] The functions $\dsr(\cdot,\cdot)$, $\dsr^2(\cdot,\cdot)$, $\dpsr(\cdot,\cdot)$ and $\dpsr^2(\cdot,\cdot)$ are LSC in the first variable, locally uniformly with respect to the second variable.
\end{itemize}
\end{theorem}

In this theorem, part (a) is actually redundant; it is a special case of part (b), with
the one-point set $\{S_0\}$ playing the role of the compact set in Definition \ref{def:lsc}(b).
Also, for $\dsr$ and $\dsr^2$, the terms ``first variable'' and ``second variable'' in Theorem~\ref{prop:lsc} can be interchanged, since $\dsr$ is symmetric.
Verifying Theorem~\ref{prop:lsc} requires substantial background work regarding the geometry of the eigen-decomposition space $M(p)$ and the eigen-composition map $\mathcal{F}$.
The following lemma is the key technical result used in proving Theorem~\ref{prop:lsc}. The radius-$r$ open ball centered at $m_0 \in M(p)$ is $B_{r}^\td(m_0) := \{ m \in M(p): d_M(m, m_0) < r\}$.

\begin{lemma}\label{lsclemma-unif} Let $K\subset \sympp$ be a compact set.  Let $\e>0$ and let $S\in \sympp$. There exists $\d_1=\d_1(S,K,\e)>0$
such that for  all $S_0\in K$,   all $m_0\in \Fc^{-1}(S_0)$, all $m \in \Fc^{-1}(S)$, and all
$S'\in \Fc\big(B_{\d_1}^\td(m )\big)$,
\be\label{final_est}
\dsr(S',S_0)^2 > \dsr(S,S_0)^2 -\e
\ee
 and
\be\label{final_est_psr}
\dpsr(S',m_0)^2 > \dpsr(S,m_0)^2 -\e.
\ee
\end{lemma}

Lemma \ref{lsclemma-unif}
does not immediately imply that $\dsr(\cdot, S_0)$
or $\dpsr(\cdot, m_0)$ is LSC at $S$, because the set $\mathcal{F}(B_{\d_1}^\td(m))$
in the lemma is not always open in $\sympp$ ($\Fc$ does not map arbitrary open sets to open sets).
However, as we show in an appendix, there exists an open ball centered at $\mathcal{F}(m)$ in $\Symp$ with radius smaller than $\delta_1$ that is contained in $\mathcal{F}(B_{\delta_1}^{d_M}(m))$ (Corollary~\ref{cor:a11}).
The background and our proofs of these supporting results and Theorem~\ref{prop:lsc} are provided in Appendix \ref{sec:prop:lsc_proof}.

Semicontinuous real-valued functions are (Borel) measurable, so an immediate   consequence of Theorem~\ref{prop:lsc}(a) is that the integrals defining the objective functions $\fsr$ and $\fpsr$ for the population (partial) scaling-rotation means exist  in $\Real \cup \{\infty\}$. This establishes the following. 

\begin{prop}\label{prop:fsr_fpsr is well defined}
Let $P$ be any
Borel
probability measure on $\Symp$.
\begin{itemize}
  \item[(i)] For any $S \in \Symp$, the integral $\int_{\Symp} \dsr^2(\cdot, S) dP$ is well-defined in $[0,\infty]$.
  \item[(ii)] For any $m \in M(p)$, the integral $\int_{\Symp} \dpsr^2(\cdot, m) dP$ is well-defined in $[0,\infty]$.
\end{itemize}
\end{prop}

(Proof for Proposition~\ref{prop:fsr_fpsr is well defined} is omitted.)

A finite-variance condition for the random variable $X \in \Symp$ with respect to the (partial) scaling-rotation distance (already needed to define $\Esr$ and $\Epsr$) is required
to establish (semi-)continuity of $\fsr$ and $\fpsr$, non-emptiness of $\Esr$ and $\Epsr$ (discussed in Section~\ref{sec:existence}), and relationships between these sets (in Section~\ref{sec:SR_sample_mean}).
For a probability measure $P$ on $\Symp$, we say $P$ has \emph{finite SR-variance} if $\fsr(S) < \infty$ for all $S \in \Symp$. Likewise, $P$ has   \emph{finite PSR-variance} if  $\fpsr(U,D) < \infty$ for all $(U,D) \in M(p)$.
The following result shows that such a condition
needs to be assumed only at a single point, rather than at all points.

\begin{lemma}\label{lem:finite_variance}
 Let $P$ be a Borel probability measure on $\Symp$ and let $\fpsr$ and $\fsr$  be the corresponding objective functions defined in equations \eqref{eq:Epsr} and \eqref{eq:esr_definition}.
\begin{itemize}
  \item[(a)]  If $\fpsr(m) < \infty$ for some $m \in M(p)$, then $\fpsr(m) < \infty$ for any $m \in M(p)$, and $\fsr(S) < \infty$ for any $S \in \Symp$.
  \item[(b)] If $\fsr(S) < \infty$ for some $S \in \Symp$, then $\fsr(S) < \infty$ for any $S \in \Symp$.
\end{itemize}
\end{lemma}

By (\ref{SR_PSR_ineq}), any probability measure with   finite PSR-variance always has   finite SR variance.

We conclude this background section by answering a natural question: Are the SR and PSR mean functions $\fsr$ and $\fpsr$ (semi-)continuous?

\begin{lemma}\label{lem:contfsr}\label{lem:contfpsr}
Let $P$ be a Borel probability measure on $\Symp$.
\begin{itemize}
 \item[(i)] If $P$ is supported in a compact set $K \subset \Symp$, then $\fsr: \Symp \to \Real$ is LSC.
  \item[(ii)] If $P$ has finite PSR-variance, then $\fpsr: M(p)\to \Real$ is continuous.
\end{itemize}
\end{lemma}

 The preceding result also implies that for any finite sample $X_1,\ldots,X_n$, $\fnsr$ (or $\fnpsr$) is LSC (or continuous, respectively). Lemma~\ref{lem:contfsr} plays an important role in developing theoretical properties of SR and PSR means, which we present in the subsequent sections.

\subsection{Existence of scaling-rotation means}\label{sec:existence}

The SR mean set $\Esr$ consists of the minimizers of the function $\fsr$. To prove existence of SR means (or, equivalently, non-emptiness of  $\Esr$) we use the fact that any LSC function on a compact set attains a minimum . For this purpose, we first verify  \emph{coercivity} of $\fsr$ (and $\fpsr$).
\begin{prop}\label{coercivity}
Let $P$ be a  Borel probability measure on $\sympp$.
\begin{itemize}
\item[(a)] There exists a compact set $K\subset \sympp$ such that
\be\label{coerce_sr}
\inf_{S\in \sympp} \fsr(S) =
\inf_{S\in K} \fsr(S).
\ee

\item[(b)] There exists a compact set $\widetilde{K}\subset \tim(p)$ such that
\be\label{coerce_psr}
\inf_{m \in \tim(p)} \fpsr(m) =
\inf_{m \in \widetilde{K}} \fpsr(m).
\ee
\end{itemize}
\end{prop}

Proposition~\ref{coercivity} says that
$\fsr$ (and $\fpsr$, respectively) is coercive, i.e. uniformly large outside some compact set, and, under the finite-variance condition, has a (non-strictly) smaller value somewhere inside that compact set.
Using this fact and the lower semicontinuity of $\fsr$ (respectively, $\fpsr$), we show in Theorem \ref{thm:existence_popSR_PSRmeans} that the SR and PSR mean sets are non-empty.
In this theorem, the bounded support condition for $P$ is used only to ensure the lower semicontinuity of $\fsr$. 

\begin{theorem}\label{thm:existence_popSR_PSRmeans}
Let $P$ be a Borel probability measure on $\Symp$.
\begin{itemize}
  \item[(a)]  If $P$ is supported on a compact set, then
$\Esr \neq \emptyset$.
  \item[(b)]  If $P$ has   finite PSR-variance, then
$\Epsr \neq \emptyset$.
\end{itemize}
\end{theorem}
%
%

Since the conditions of Theorem~\ref{thm:existence_popSR_PSRmeans} are met for any empirical measure defined from a finite set $\{ X_1,\ldots,X_n\}  \subset \Symp$, a corollary of the population SR mean result is the existence of sample SR means:

\begin{corollary}\label{thm:existence_sampleSR_PSRmeans}
 For any finite $n$, and any $X_1,\ldots,X_n \in \Symp$, $\Ensr \neq \emptyset$ and $\Enpsr \neq \emptyset$.
\end{corollary}

\begin{remark}
For any Borel-measurable $\Symp$-valued random variable with finite PSR-variance, the PSR mean set $\Epsr$ is closed. In particular, every sample PSR mean set is closed. To verify this, recall from Lemma~\ref{lem:contfpsr} that $\fpsr$ is continuous. The PSR mean set is a level set of a continuous function, and therefore is closed.

Moreover, as seen in Proposition~\ref{coercivity}, the closed set $\Epsr$ is a subset of a compact set, thus is compact as well.
\end{remark}

\subsection{Uniqueness of PSR means}\label{uniqueness}

Much work has been done on the question of uniqueness of the Fr\'{e}chet mean of Riemannian manifold-valued observations. It is known that the Fr\'{e}chet mean is unique as long as the support of the probability distribution lies within a geodesic ball of a certain radius (see, for example, \cite{Afsari2011}). Although $d_{\mathcal{SR}}$ is not a geodesic distance on ${\rm Sym}^+(p)$,
we can obtain a similar result for a kind of uniqueness of the PSR mean.

For any $X \in \Symp$, recall from (\ref{eq:fiber}) that $\mathcal{F}^{-1}(X) = \{ h\cdot(UR,D): R \in G_D^0, h\in \mathcal{G}(p)\}$ for an eigen-decomposition $(U,D)$ of $X$. Since the finite group $\mathcal{G}(p)$ acts freely and isometrically on $M(p)$,  for any $h \in \mathcal{G}(p)$ and $m \in M(p)$,
\begin{align*}
 \dpsr(X, m ) &=\inf_{R \in G_D^0, h'\in \mathcal{G}(p) } d_M( h' \cdot(UR,D) , m) \\
            & = \inf_{R \in G_D^0, h \cdot h' \in \mathcal{G}(p) } d_M( h \cdot h' \cdot(UR,D) , h \cdot m) = \dpsr(X, h \cdot m ).
 \end{align*}
For a sample $X_1, \dotsc, X_n \in \Symp$, we have thus
\begin{equation}\label{sample_PSR_mean_nonunique}
\fnpsr(m) = \frac{1}{n}\sum_{i=1}^n d^2_{\mathcal{PSR}}(X_i, m) = \frac{1}{n}\sum_{i=1}^n d^2_{\mathcal{PSR}}(X_i, h\cdot m)) = \fnpsr(h\cdot m)
\end{equation}
for any $h \in \mathcal{G}(p)$ and $m \in M(p)$.
It follows from (\ref{sample_PSR_mean_nonunique}) that for any $m \in E_n^{(\mathcal{PSR})}$, the remaining members of its orbit $\mathcal{G}(p) \cdot m$  (see (\ref{eq:orbit})) also belong to $E_n^{(\mathcal{PSR})}$.
Thus, $E_n^{(\mathcal{PSR})}$ will contain at least $2^{p-1}p!$ elements.
In the case where $E_n^{(\mathcal{PSR})}$ only contains $2^{p-1}p!$ elements (necessarily belonging to the same orbit), we will say that the sample PSR mean is \emph{unique up to the action of $\mathcal{G}(p)$}.  The notion of uniqueness (up to the action of $\mathcal{G}(p)$) for the population PSR mean in $\Epsr$ is defined similarly.

The following lemma yields a useful lower bound on the distance between distinct eigen-decompositions of an SPD matrix in  $S_p^{\rm top}$.  (Note that for any $X \in S_p^{\rm lwr}$, the set of eigen-decompositions of $X$ is not discrete, so two eigen-decompositions of $X$ may be arbitrarily close to each other.) 

\begin{lemma}\label{eig_decomp_dist_bound}
(a) For any $(U,D) \in M(p)$ and
 for any $h \in \mathcal{G}(p)\setminus \{ I_p \}$,
\[ d_{M}((U,D), h\cdot (U,D) ) \ge \sqrt{k}\beta_{\mathcal{G}(p)} \]
where
$ \beta_{\mathcal{G}(p)} := \min_{h \in \mathcal{G}(p) \setminus \{ I_p \}} d_{SO}(I_p,h) = \min_{h \in \mathcal{G}(p) \setminus \{I_p\}} \frac{1}{\sqrt{2}}\| {\rm Log}(h) \|_F.$

(b) The quantity $\beta_{\mathcal{G}(p)}$ satisfies $\beta_{\mathcal{G}(p)} \le \frac{\pi}{2}$ for any $p \ge 2$.

(c) For any $X \in S_p^{\rm top}$, any two distinct eigen-decompositions $(U_X,D_X)$ and $(U'_X,D'_X)$ of $X$ satisfy
$ d_{M}((U_X,D_X),(U'_X,D'_X)) \ge \sqrt{k}\beta_{\mathcal{G}(p)}.$
\end{lemma}

\begin{remark}
It is easily checked that $\beta_{\mathcal{G}(p)} = \frac{\pi}{2}$ when $p=2,3$.
\end{remark}

In Theorem~\ref{thm:uniqueness} below, we provide a sufficient condition for  uniqueness (up to the action of $\mathcal{G}(p)$) of the PSR means.
In preparation, we first provide a sufficient condition for a distribution on $M(p)$ to have a unique Fr\'{e}chet mean. Recall that $(M(p), g_M)$ is a Riemannian manifold, which in turn implies that $(M(p), d_M)$ is a metric space. The Fr\'{e}chet mean set for a probability distribution $P$ on $M(p)$ is thus well-defined.


\begin{lemma}\label{lem:geom}
 Let $\tilde{P}$ be a Borel probability measure on $M(p)$. Suppose that ${\rm supp}(\tilde{P})$, the support of $\tilde{P}$, satisfies
\begin{equation}\label{eq:condition_unique_3}
 {\rm supp}(\tilde{P})  \subseteq B_r^{d_M}(m_0)
\end{equation}
for some $r\le\sqrt{k} \beta_{\mathcal{G}(p)}$ and some $m_0 \in M(p)$. Then there exists a unique Fr\'{e}chet mean
$\bar{m}(\tilde{P}):= \argmin_{m \in M(p)} \int_{M(p)} d_M^2(\tilde{X}, m) \tilde{P}(d \tilde{X})$ of $P$, and
$\bar{m}(\tilde{P}) \in B_r^{d_M}(m_0)$.
\end{lemma}

Similarly to Lemma~\ref{lem:geom}, if a deterministic sample $m_1,\ldots,m_n \in M(p)$ lies in $B_r^{d_M}(m_0)$ $(i = 1,\ldots,n)$ for some $r\le\sqrt{k} \beta_{\mathcal{G}(p)}$ and some $m_0 \in M(p)$, then the sample Fr\'{e}chet mean $\bar{m}:= \argmin_{m \in M(p)} \frac{1}{n}\sum_{i=1}^n d_M^2(m_i, m)$ is unique and lies in $B_r^{d_M}(m_0)$.

\begin{theorem}\label{thm:uniqueness}
Suppose the probability measure $P$ on $\Symp$ is absolutely continuous with respect to volume measure and that for two independent $\Symp$-valued random variables $X_1, X_2$ whose distribution is $P$,
\begin{equation}\label{eq:condition_unique}
P( \dsr(X_1, X_2) < r'_{cx}) = 1, \quad \mbox{\rm where}\ \  r'_{cx} := \frac{\sqrt{k}\beta_{\mathcal{G}(p)}}{4}.
\end{equation}
Then the population PSR mean set $\Epsr$ is unique up to the action of $\mathcal{G}(p)$.
\end{theorem}

The number $r'_{cx}$ is a lower bound on the regular convexity radius
of the quotient space $M(p) / \mathcal{G}(p)$ with the induced Riemannian structure, as shown in  \cite{Groisser2018}.  This ensures that a ball in $M(p) / \mathcal{G}(p)$ with radius less than $r'_{cx}$ is convex. The quotient space $M(p) / \mathcal{G}(p)$ ``sits'' between $M(p)$ and $\Symp$; any $X \in \Sptop$ coincides with an element in $M(p) / \mathcal{G}(p)$, but there are multiple (in fact, infinitely many) elements in $M(p) / \mathcal{G}(p)$  corresponding to any $X \in \Splwr$ (cf. (\ref{eq:fiber})). Lemma~\ref{eig_decomp_dist_bound}  shows that $r'_{cx} \le \sqrt{k}\pi/8$.
 In contrast, the regular convexity radius of $(M(p),g_M)$ is $\sqrt{k}\pi/2$, which is much larger than $r'_{cx}$.
 Even though we work with the eigen-decomposition space $M(p)$, in Theorem~\ref{thm:uniqueness} we require  data-support diameter at most $r'_{cx} < \sqrt{k}\pi/2$ since, if
 $\dsr(S_1,S_2) \ge r'_{cx}$ for some $S_1,S_2 \in \Symp$,
 then there may be two or more eigen-decompositions of $S_1$ that are both closest to an eigen-decomposition of $S_2$.

The assumption of absolute continuity of $P$ in Theorem~\ref{thm:uniqueness} enables us to restrict our attention to the probability-1 event for which the random variables lie in the top stratum $S_p^{\rm top}$ of $\Symp$, since the complement of $S_p^{\rm top}$ has volume zero in $\Symp$. Corollary \ref{uniqueness_theorem} below explicitly states this restriction as a sufficient condition for the uniqueness of  sample PSR means of a deterministic sample. We also show that the estimation procedure (Procedure \ref{proc:estimation}) will yield the unique (up to the action of $\mathcal{G}(p)$) sample PSR mean.

\begin{corollary}\label{uniqueness_theorem}
Assume $X_1, \dotsc, X_n \in S_p^{\rm top}$. If
\begin{equation}\label{eq:condition_unique2}
d_{\mathcal{SR}}(X_i,X_j) < r'_{cx}
\end{equation}
for all $i,j = 1,\ldots,n$, then
\begin{itemize}
\item[(a)] the sample PSR mean is unique up to the action of $\mathcal{G}(p)$;
\item[(b)] choosing an eigen-decomposition of any observation from the sample as the initial guess will lead Procedure \ref{proc:estimation} to converge to the sample PSR mean after one iteration.
\end{itemize}
\end{corollary}

\begin{remark}\label{remark:4.9}
 The data-diameter condition \eqref{eq:condition_unique2} in Corollary \ref{uniqueness_theorem} is  satisfied under either of the following two conditions (in the presence of the assumption $X_i \in \Sptop$):
  \begin{itemize}
    \item[(i)] There exists an $S_0 \in \Sptop$ such that $\dsr(S_0,X_i) < r'_{cx}/2$ for all $i = 1,\ldots, n$.
    \item[(ii)] There exists an $m \in M(p)$ such that $\dpsr(X_i,m) < r'_{cx}/2$ for all $i= 1,\ldots, n$.
  \end{itemize}
\noindent
 Similarly, the condition that $\dsr(X_1,X_2) < r'_{cx}$ almost surely in Theorem \ref{thm:uniqueness} is guaranteed by either (i) or (ii) above, when the latter two conditions are modified probabilistically; see Appendix~\ref{remark:4.9proof}.
In condition (i) above, it is necessary for the center of the open ball (data support) to lie in the top stratum, due to the fact that the functions $\dsr(\cdot, X_i)$ are, in general, only LSC (not continuous) at points belonging to $\Splwr$. For an $S_0 \in \Splwr$, even if a condition $\dsr(S_0,X_i) < \epsilon$ ($i = 1,\ldots,n$) is satisfied for arbitrarily small $\epsilon$, $d_{\mathcal{SR}}(X_i,X_j)$ may be larger than $r'_{cx}$.

 Proof of the statements given in this remark can be found in Appendix~\ref{remark:4.9proof}.

%
\end{remark}

 If the data-support is small enough to satisfy \eqref{eq:condition_unique} and also is far from the lower stratum (satisfying the conditions in Theorem~\ref{thm:avoid_low_strata_SR}), then the SR mean is unique, as the following corollary states.

\begin{corollary}\label{cor:uniqueness}
Let $X$ be a $\Symp$-valued random variable following the distribution $P$. Assume that there exist $S_0 \in \Sptop$ and $r < \min\{\delta(S_0)/3, r'_{cx}/2 \}$ satisfying
$P(\dsr(S_0,X_i) \le r) = 1$. Then, (i) the PSR mean is unique up to the action of $\mathcal{G}(p)$, (ii) $\Esr \subset \Sptop$, and (iii) $\Esr = \mathcal{F}(\Epsr)$ is a singleton set.
\end{corollary}

\subsection{Asymptotic properties of the sample PSR means}\label{asymptotics}

This subsection addresses two aspects of the asymptotic behavior of the sample PSR mean $\Enpsr$: (i) strong consistency of $\Enpsr$ with the population PSR mean set $\Epsr$ and (ii) the large-sample limiting distribution of a sample PSR mean. Much work has been done to establish consistency and central limit theorem-type results for sample Fr\'{e}chet means on Riemannian manifolds and metric spaces (\cite{Bhatt2003}, \cite{Bhatt2005}, \cite{Bhatt2017}, \cite{eltzner2019smeary}). Estimation of the PSR mean does not fit into the context of estimation on Riemannian manifolds or metric spaces since the sample space ${\rm Sym}^+(p)$ and \emph{parameter space} $M(p)$ are different.
Moreover, as we have seen, the PSR means are never unique.
With this in mind, we apply the framework of generalized Fr\'{e}chet means on general product spaces in \cite{Huckemann2011a} and \cite{Huckemann2011b} to our PSR mean estimation context, enabling us to establish conditions for strong consistency and for a central limit theorem.

We now establish a strong-consistency result for $E_n^{\mathcal{PSR}}$.
Throughout this subsection, let $X, X_1,\ldots$ be  independent random variables mapping from a complete probability space $(\Omega, \mathcal{A}, \mathcal{P})$ to $ {\rm Sym}^+(p)$ equipped with its Borel $\sigma$-field, and let $P$ be the induced Borel probability measure on $\Symp$. The
 sets $\Epsr$ and $\Enpsr$ denote the population and sample PSR-mean sets defined by $P$ and $X_1, \dots, X_n$, respectively.

\begin{theorem}\label{thm:consistency}
Assume that $P$ has finite PSR-variance.
Then
 \begin{equation}\label{eq:thm:consistency}
 \lim_{n \to \infty} \sup_{ m \in \Enpsr }d_M( m ,\Epsr) = 0
\end{equation}
almost surely.
\end{theorem}

Our proof of Theorem~\ref{thm:consistency} is contained in Appendix~\ref{app:proof_of_consistency}. There, we closely follow the arguments of  \cite{Huckemann2011b} used in verifying the conditions required to establish strong consistency of the generalized Fr\'{e}chet means.
However, the theorems of \cite{Huckemann2011b} are not directly applied since the function $\dpsr$ is not continuous.
Nevertheless,  the Fr\'{e}chet-type objective function $\fpsr: M(p)\to \Real$ is continuous, as shown in Lemma~\ref{lem:contfpsr},
a fact that
 plays an important role in the proof of Theorem~\ref{thm:consistency}.
\cite{schotz2022strong} extends the results of \cite{Huckemann2011b} by, among other things and in our notation, allowing for $\dpsr(X,\cdot)$ to be only LSC. However, this is not actually helpful for $\dpsr$ either, because $\dpsr$ is actually \emph{continuous} with respect to the second variable (it is LSC  with respect to the \emph{first} variable);
see Lemma \ref{lem:dpsr_continuity} and Theorem \ref{prop:lsc}.

In the proof of Theorem~\ref{thm:consistency}, we first show that with probability 1
\begin{equation}\label{eq:Ziezold}
  \cap_{k=1}^\infty \overline{ \cup_{n=k}^\infty \Enpsr} \subset \Epsr.
\end{equation}
 In the terminology of \cite{Huckemann2011b}, (\ref{eq:Ziezold}) is called strong consistency of $\Enpsr$ as an estimator of  $\Epsr$ in the sense of \cite{ziezold1977expected}. Our result (\ref{eq:thm:consistency}) is equivalent to strong consistency in the sense of \cite{Bhatt2003} (again using the terminology of \cite{Huckemann2011b}), as shown in Lemma~\ref{lem:BPconsistency-vs-sup} in the appendix.
 \cite{schotz2022strong} classified three types of convergence for a sequence of sets, referring to (\ref{eq:Ziezold})
  as  convergence \emph{in the outer limit}, and to
 (\ref{eq:thm:consistency}) as convergence \emph{in the one-sided Hausdorff distance}. The last type of convergence is convergence \emph{in Hausdorff distance}. Recall that for a metric space $(M,d)$ the Hausdorff distance between non-empty sets $A,B \subset M$ is $d_H(A,B) := \max\{\sup_{m \in A}d(m,B), \sup_{ m \in B }d( A , m) \}$.

Theorem~\ref{thm:consistency} states that, with probability 1,  any sequence $m_n \in \Enpsr$ of sample PSR means will eventually  lie in an arbitrarily small neighborhood of  the population PSR mean set as the sample size $n$ increases. But, conceivably there could be a population PSR mean in $\Epsr$ with no sample PSR mean nearby even for large $n$, in which case $d_H(\Enpsr,\Epsr)$ would not approach zero.
In other words, $\Enpsr$ would be a strongly consistent estimator of $\Epsr$ only with respect to {\em one-sided} Hausdorff distance, not (two-sided) Hausdorff distance.
However,
if the population PSR mean is unique up to the action of $\mathcal{G}(p)$, then $\Enpsr$ \emph{is} a strongly consistent estimator of $\Epsr$ with respect to Hausdorff distance on $(M(p),d_M)$, as shown next.

\begin{corollary}\label{cor:Hausdorff}
Assume that $P$ has finite PSR-variance, and that $\Epsr = \mathcal{G}(p) \cdot \mu$ for some $\mu \in M(p)$. Then with probability 1,
\begin{equation}\label{eq:cor:Hausdorff1}
 \lim_{n \to \infty} \sup_{ m \in \Epsr }d_M( \Enpsr , m) = 0
\end{equation}
and
\begin{equation}\label{eq:cor:Hausdorff2}
\lim_{n \to \infty} d_H( \Enpsr , \Epsr) = 0.
\end{equation}
\end{corollary}

The strong consistency of sample PSR means with the population PSR means can be converted to strong consistency of sample PSR means with the \emph{population SR means}, as follows.  For $S \in \Symp$ and a set $\mathcal{E} \subset \Symp$, we define $\dsr(S, \mathcal{E}) := \inf_{E \in \mathcal{E}}\dsr(S, E)$.

\begin{corollary}\label{cor:consistency}
Assume that $P$ has finite PSR-variance.  Then,

(i) $\lim_{n\to\infty} \sup_{S \in \mathcal{F}(\Enpsr)} \dsr(S, \mathcal{F}(\Epsr)) = 0$ almost surely.

(ii) If, in addition, $\Esr \subset \Sptop$, then $\lim_{n\to\infty} \sup_{S \in \mathcal{F}(\Enpsr)} \dsr(S, \Esr) = 0$ almost surely.

(iii) If $\Esr \subset \Sptop$ and the population SR mean is unique with $\Esr = \{\mu^{(\mathcal{SR})}\}$, then  $\lim_{n\to\infty}  \dsr(\mathcal{F}(\Enpsr) , \mu^{(\mathcal{SR})}) = 0$ almost surely.
\end{corollary}

Note that in establishing a strong consistency property of $\Enpsr$ with respect to population (partial) SR means, we assumed only that the \emph{population} mean set $\Epsr$ is unique up to the action of $\mathcal{G}(p)$, not that the \emph{sample} mean sets $\Enpsr$ have this  uniqueness property. We also did not assume that $\Enpsr \subset M_p^{\rm top}$.

We next establish a central limit theorem for our estimator $\Enpsr$.
Our strategy is to closely follow the arguments in \cite{Bhatt2005,Huckemann2011a,Bhatt2017,eltzner2021stability},
 for deriving central limit theorems for (generalized) Fr\'{e}chet means on a Riemannian manifold. In particular, our central limit theorem is expressed in terms of charts and the asymptotic distributions of "linearized" estimators.

Our parameter space of interest $M(p) =SO(p) \times {\rm Diag}^+(p)$ is a Riemannian manifold of dimension $d:=\tfrac{(p-1)p}{2} + p$. As defined in Section~\ref{eigen-decomp_geometry}, the tangent space at $(U,D) \in M(p)$ is $T_{(U,D)}M(p) =\{(AU, LD): A \in \mathfrak{so}(p), L \in {\rm Diag}(p)\}$, which can be canonically identified with $\mathfrak{so}(p) \oplus {\rm Diag}(p)$, a vector space of dimension $d$.

At $(U,D) \in M(p)$, we use the local chart $(\mathcal{U}_{(U,D)}, \tilde\varphi_{(U,D)})$, where $$\mathcal{U}_{(U,D)} = \{(V,\Lambda) \in  M(p): \|{\rm Log}(VU^T)\|_F < \pi \},$$
and where
$\tilde\varphi_{(U,D)}: \mathcal{U}_{(U,D)} \to \mathfrak{so}(p) \oplus {\rm Diag}(p)$ is defined by
\begin{equation}\label{eq:local_chart}
   \tilde\varphi_{(U,D)}(V, \Lambda)  = ({\rm Log}(VU^T), {\rm Log}(\Lambda D^{-1})).
\end{equation}
Observe that $ \tilde\varphi^{-1}_{(U,D)}(A,L)  = ({\rm Exp}(A) U, {\rm Exp} (L)D).$
 The maps $\tilde\varphi_{(U,D)}$ and $\tilde\varphi_{(U,D)}^{-1}$ are the Riemannian logarithm and exponential maps to (and from) the tangent space $T_{(U,D)}M(p)$,
composed with the right-translation isomorphism between $T_{(U,D)}M(p)$ and $T_{(I,I)}M(p)=\mathfrak{so}(p)\oplus {\rm Diag}(p)$.

We also write the elements of $\mathfrak{so}(p) \oplus {\rm Diag}(p)$ in a coordinate-wise vector form. For each $(A,L) \in \mathfrak{so}(p) \oplus {\rm Diag}(p)$, define  a suitable vectorization operator ${\rm vec}$,
\begin{equation}\label{eq:vectorization}
 {\rm vec}(A,L) := \begin{pmatrix}
 \sqrt{k} \ x_{SO}(A) \\
                   x_{\mathcal{D}}(L) \\
                   \end{pmatrix}
  \in \Real^d,
\end{equation}
 where $x_{SO}(A)\in \Real^{(p-1)p/2}$ consists of the upper triangular entries of $A$ (in the lexicographical ordering) and $x_{\mathcal{D}}(L) = (L_{11},\ldots, L_{pp})^T \in \Real^p$ consists of the diagonal entries of $L$.
 The inverse vectorization operator ${\rm vec}^{-1}$ is well-defined as well.
 We use the notation $\phi_{(U,D)}(\cdot,\cdot):= {\rm vec}\circ \tilde\varphi_{(U,D)}(\cdot,\cdot)$ and $\phi^{-1}_{(U,D)}(\cdot) := \tilde\varphi^{-1}_{(U,D)} \circ {\rm vec}^{-1}(\cdot)$.

Assume the following.

(A1) The probability measure $P$ induced by $X$ on $\Symp$ is absolutely continuous with respect to volume measure, and has finite PSR-variance.

(A2) $\Epsr$ is unique up to the action of $\mathcal{G}(p)$. With probability 1, so is $\Enpsr$ (for every $n$).

(A3) For some $m_0 \in \Epsr$, $P( d_{\mathcal{PSR}}(X,m_0)<r'_{cx}) = 1$.

 The absolute continuity assumption (A1) ensures that any volume-zero (Lebesgue-measurable) subset of $\Symp$ has probability zero. In particular, $P(X \in \Sptop) = 1 - P(X \in \Splwr) = 1$. This fact greatly simplifies our theoretical development.

The uniqueness assumption (A2) ensures that $\Enpsr$ converges  almost surely to $\Epsr$ with respect to the Hausdorff distance (by Corollary~\ref{cor:Hausdorff}). Therefore, for any $m_0 \in \Epsr$, there exists a sequence $m_n \in \Enpsr$ satisfying $ d_M(m_n, m_0)\to 0$ (or, equivalently, $\phi_{m_0}(m_n) \to \phi_{m_0}(m_0) = 0$)  as $n \to \infty$ almost surely. Assumption (A2) also guarantees that if (A3) is true for some PSR mean $m_0 \in \Epsr$ then it is true for any other PSR mean in $\Epsr$.

The radius $r'_{cx} = \sqrt{k}\beta_{\mathcal{G}(p)}/4$ in Assumption (A3) previously appeared in Theorem~\ref{uniqueness_theorem}, where the bounded-support assumption was used to ensure
uniqueness of one element of a minimal pair (see Definition~\ref{def:dsr}) when the other element is fixed.
Similarly, assumptions (A1) and (A3) ensure that, with probability 1, for each $X_i$ there exists a unique $m_i \in \mathcal{F}^{-1}(X_i)$ such that $m_i \in B^{d_M}_{r'_{cx}}(m_0)$, a radius-$r'_{cx}$ ball in $M(p)$ centered at $m$.
A stronger version of this fact will be used (in the proof of Theorem~\ref{thm:PSR_mean_CLT}, to be given shortly) to rewrite the objective function $\fnpsr$ involving $\dpsr$ as a Fr\'{e}chet objective function $m \mapsto \frac{1}{n}\sum_{i=1}^n d_M^2(m_i, m)$, with probability 1.

In addition, the bounded support condition (A3) ensures that with probability 1 the function $\dpsr^2(X, \cdot)$ is smooth ($C^\infty)$ and convex on a convex set. Using this fact and geometric results from given in \cite{Afsari2011} and \cite{afsari2013convergence}, we show in the proof that
the gradient
$${\rm grad}_x\, \dpsr^2(X, \phi^{-1}_{m_0}(x)) := \left(\frac{\partial}{\partial x_i} \dpsr^2(X, \phi^{-1}_{m_0}(x))\right)_{i=1,\ldots,d}$$
at $x = 0$ has mean zero, and has a finite covariance matrix $\Sigma_P := {\rm Cov}({\rm grad}_x\, \dpsr^2(X, \phi^{-1}_{m_0}(0)))$. (We conjecture that (A1) guarantees that $\Sigma_P$ is also positive-definite.)
 Likewise, as we will see in the proof,
the differentiability  and (strict) convexity of $\dpsr^2(X, \phi^{-1}_{m_0}( \cdot ))$ guarantee that the expectation of the Hessian
$ H_P(x) := E \left( \mathbf{H}\dpsr^2(X, \phi^{-1}_{m_0}( x ))\right) $
 exists and is positive definite at $x = 0$. Write $H_P:= H_P(0)$.

In summary, Assumptions (A1)---(A3) enable us to use a second-order Taylor expansion for $\fnpsr$, to which the classical central limit theorem and the law of large numbers are applied. Such an approach was used in \cite{Bhatt2005} and \cite{Huckemann2011a}.  In particular, our proof for part (b) of Theorem \ref{thm:PSR_mean_CLT}  (in Appendix \ref{sec:proof_of_thm:PSR_mean_CLT}) closely follows the proof of Theorem 6 of \cite{Huckemann2011a}.

\begin{theorem}\label{thm:PSR_mean_CLT}
Suppose that Assumptions (A1)---(A3) are satisfied, and let $m_0 \in \Epsr$ be a PSR mean.
Let
 $\{m_n' \in \Enpsr\}$ be any choice of sample PSR mean sequence. For each $n$, let $m_n \in \argmin_{m \in \mathcal{G}(p)\cdot m_n'} d_M(m, m_0)$. Then, with probability 1, the sequence $\{m_n\}$  is determined  uniquely.
  Furthermore,
\begin{itemize}
  \item[(a)] $m_n \to m_0$ almost surely as $n \to \infty$, and
  \item[(b)] $\sqrt{n}\phi_{m_0}(m_n) \to N_d(0, H_P^{-1}\Sigma_{P} H_P^{-1})$
in distribution as $n \to \infty$.
\end{itemize}
\end{theorem}

Estimating the covariance matrix ($H_P^{-1}\Sigma_{P} H_P^{-1}$ in our case) of the limiting Gaussian distribution (for Riemannian manifold-valued Fr\'{e}chet means) is a difficult task.
For general Riemannian manifold-valued Fr\'{e}chet means, \cite{Bhatt2005} and \cite{bhattacharya2012nonparametric} suggest using  a moment estimator for $H_P$ and $\Sigma_P$. 
This however requires specifying the second derivatives of $\dpsr^2(X, \phi_{m_0}(\cdot))$.
We note that in the literature, explicit expressions for $H_P$ and $\Sigma_P$ are only available for geometrically very simple manifolds, with a high degree of symmetry, such as spheres.
 As an instance, see \cite{hotz2015intrinsic} and Section 5.3 of \cite{bhattacharya2012nonparametric} for the cases where the data and the Fr\'{e}chet mean lie in the unit circle $S^1$ and the more general unit sphere $S^d$, respectively.
 Others, including \cite{eltzner2019smeary}, simply use the sample covariance matrix of  $\{\phi_{m_n}(m_{_{X_i}}): i =1,\ldots,n\}$ (in our notation) as an estimate of $H_P^{-1}\Sigma_{P} H_P^{-1}$.
 In Section~\ref{sec:5}, we will use a bootstrap estimator of ${\rm Var}(\phi_{m_0}(m_n))$, the variance of $\phi_{m_0}(m_n)$, instead of directly estimating  $H_P^{-1}\Sigma_{P} H_P^{-1}$. Out bootstrap estimator is defined as follows.

Choose a PSR mean $m_n$ computed from the original sample $\{X_1,\ldots,X_n\}$. For the $b$th bootstrap sample (that is, a simple random sample of size $n$ from the set $\{X_1,\ldots,X_n\}$, treated as a fixed set, with replacement), let $m_b^*$ be the PSR mean of the bootstrap sample that is closest to $m_n$.
(For the purpose of defining the bootstrap estimator, we are assuming that such an $m_b^*$ is unique.)
The bootstrap estimator of ${\rm Var}(\phi_{m_0}(m_n))$ is then defined to be
$$ \widehat{\rm Var}_{\rm boot}(\phi_{m_0}(m_n)) := \frac{1}{B}\sum_{b=1}^B \phi_{m_n}(m_b^*) \cdot (\phi_{m_n}(m_b^*))^T,$$
where
$B$ is the number of bootstrap replicates.


\section{Numerical examples}\label{sec:5}

\subsection{Numerical examples of scaling-rotation means}
\label{sec:Appendix_numerical_examples}
In this subsection, we provide an example where the SR means are equivalent to the PSR means, and an example where they are not. Consider a random variable $X \in \textrm{Sym}^+(2)$,
\begin{equation}\label{eq:dim2_spd_matrix_model}
 X = R(\theta) {\rm diag}(\exp(D_1),\exp(D_2)) R(\theta)^T,
\end{equation}
where $\theta$ follows the normal distribution with mean 0, standard deviation $\sigma_\theta$, truncated to lie in $(-\pi, \pi)$,
and independently $(D_1,D_2)$  follow a normal distribution with mean $(\mu_1,\mu_2)$ and covariance matrix $\sigma_D^2 I_2$. From this model, we generate two samples of size $n = 200$ with different choices of model parameters.

For each sample, a PSR mean, denoted $\hat{m}^{\mathcal{PSR}}$, is computed using the algorithm of Section~\ref{sec:3.3}, and we also numerically compute the minimizer of $\fnsr$ over $S_p^{\rm lwr}$, and denote it by $\hat{M}^{\mathcal{SR}}_{\rm lwr}$. Throughout we set $k = 1$.
By Theorem~\ref{thm:how_to_tell_whether_PSR=SR},
 if
 $\fnsr(\mathcal{F}(\hat{m}^{\mathcal{PSR}})) \le  \fnsr(\hat{M}^{\mathcal{SR}}_{\rm lwr})$,
$\mathcal{F}(\hat{m}^{\mathcal{PSR}})$ is an SR mean, and otherwise $\hat{M}^{\mathcal{SR}}_{\rm lwr}$ is a SR mean.

\begin{itemize}
  \item Case I: Set $\sigma_\theta = \pi/12$, $(\mu_1,\mu_2) = (2,0)$ and $\sigma_D = 0.2$. See the left panels of Fig.~\ref{fig:SRvsPSR}.
  \item Case II:  Set $\sigma_\theta = \pi/3$, $(\mu_1,\mu_2) = (1,0)$ and $\sigma_D = 0.2$. See the right panels of Fig.~\ref{fig:SRvsPSR}.
\end{itemize}

For Case I, the sample are relatively far from the lower stratum $S_2^{\rm lwr} =\{cI_2: c>0\}$ (shown as the green axis in the top row of Fig.~\ref{fig:SRvsPSR}). In this particular instance, $82 \approx \fnsr(\mathcal{F}(\hat{m}^{\mathcal{PSR}})) <  \fnsr(\hat{M}^{\mathcal{SR}}_{\rm lwr}) \approx   458$, and $\mathcal{F}(\hat{m}^{\mathcal{PSR}})$ is an SR mean.

\begin{figure}[pt]
\centering
\includegraphics[width = 1\textwidth]{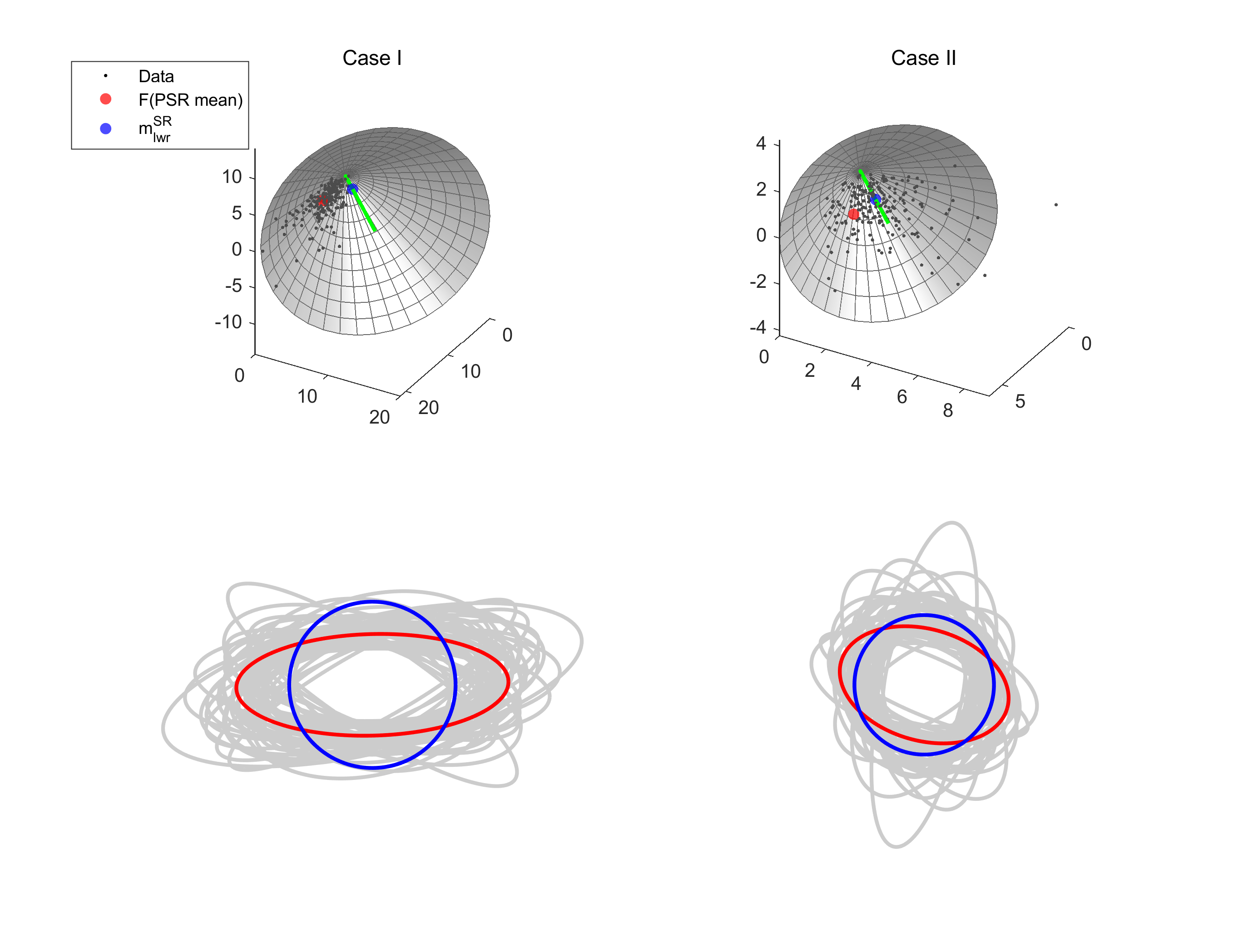}
\vspace{-.5in}
\caption{Data on ${\rm Sym}^+(2)$ shown in the cone of ${\rm Sym}^+(2)$ (top row) and as ellipses (bottom row), compared with their PSR means $\mathcal{F}(\hat{m}^{\mathcal{PSR}})$ (red) and $\hat{M}^{\mathcal{SR}}_{\rm lwr}$ (blue). Left panel shows data from Case I, where $\mathcal{F}(\hat{m}^{\mathcal{PSR}})$ is indeed an SR mean; Right panel shows data from Case II, where $\hat{M}^{\mathcal{SR}}_{\rm lwr}$ is an SR mean. See Section~\ref{sec:Appendix_numerical_examples} for details.
}\label{fig:SRvsPSR}
\end{figure}

For Case II, $ 196  \approx \fnsr(\mathcal{F}(\hat{m}^{\mathcal{PSR}})) >  \fnsr(\hat{M}^{\mathcal{SR}}_{\rm lwr}) \approx  173$, and
$\hat{M}^{\mathcal{SR}}_{\rm lwr}$ is an SR mean.

\subsection{Comparison with other geometric frameworks%
}\label{sec:simu}

In analyzing SPD matrices, the scaling-rotation (SR) framework has an advantage in interpretation as it allows describing the changes of SPD matrices in terms of rotation and scaling of the corresponding ellipsoids. For example, it is shown in  \cite{Jung2015} that only the SR framework yields interpolation curves which consist of pure rotation when the endpoints differ only by rotation, when compared to the commonly used log-Euclidean \citep{Arsigny2007} and affine-invariant \citep{Fletcher2004,Pennec2006} interpolation curves.

In this subsection, we illustrate situations under which averaging via the scaling-rotation framework has similar interpretive advantages over the affine-invariant mean. %
The affine-invariant (AI) mean $\bar{X}^{\rm (AI)}$ for a sample of SPD matrices $X_1,\ldots,X_n\in \Symp$ is the sample Fr\'{e}chet mean with respect to the affine-invariant metric $d_{AI}$:
$$
\bar{X}^{\rm (AI)} = \argmin_{M \in \Symp} \sum_{i=1}^n d_{AI}^2(M, X_i),$$
where $d_{AI}(X,Y) = \|{\rm Log}(X^{-1/2}YX^{-1/2})\|_F$. The AI mean exists and is unique  for any finite sample  \citep{Pennec2006}.

In numerical experiments, we randomly generated SPD matrices from the model (\ref{eq:dim2_spd_matrix_model}) with parameters set as in Case I but with $\sigma_\theta = \pi/6$. A sample of size $n = 200$ is plotted in Fig. \ref{fig:Sec2Comparison_110}. There, we have used two different types of ``linearizations'' of ${\rm Sym}^+(2)$, as explained below.

\begin{figure}[t]
\centering
\includegraphics[width = 0.95\textwidth]{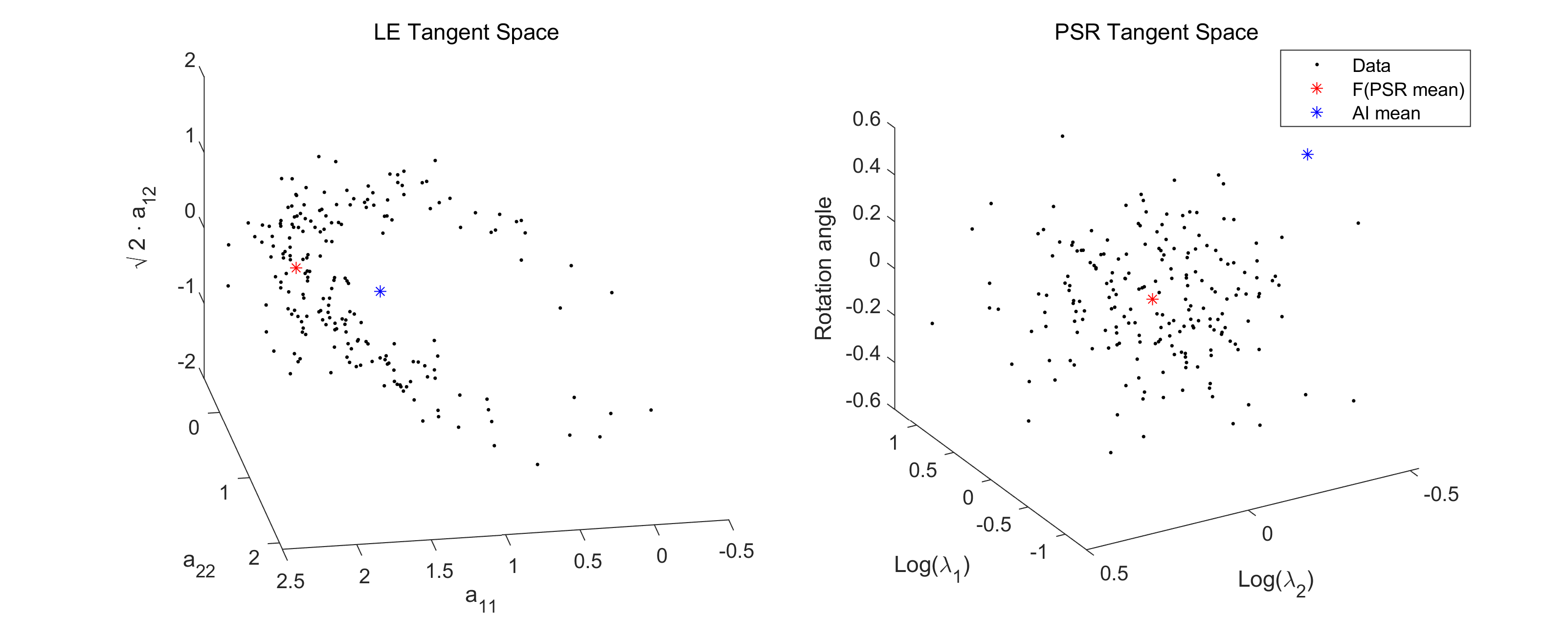}
\caption{A sample of SPD matrices (sampled from the model (\ref{eq:dim2_spd_matrix_model})) shown in the Log-Euclidean (LE) coordinates (left) and the PSR coordinates (right), overlaid with the PSR mean and AI mean. For this data set, PSR mean appears to be a better representative for the data, while the AI mean does not lie in the data-dense region.
See Section~\ref{sec:simu} for details. }\label{fig:Sec2Comparison_110}
\end{figure}

The \emph{Log-Euclidean coordinates} on ${\rm Sym}^+(2)$ are given by the three free parameters $y_{11}$, $y_{22}$ and $\sqrt{2}y_{12}$ of $Y = (y_{ij}) = {\rm Log}(X) \in {\rm Sym}(2)$. Write ${\rm vecd}(Y):= (y_{11}, y_{22}, \sqrt{2}y_{12})^T \in \Real^3$ . These coordinates are chosen so that for any two vectors $({\rm vecd}(X),{\rm vecd}(Y)) = (x,y)$, the usual inner product $\langle x,y \rangle = x^Ty$ corresponds to the Riemannian metric when $X, Y \in {\rm Sym}(2)$ are viewed as tangent vectors in the affine-invariant framework.  The left panel of  Fig. \ref{fig:Sec2Comparison_110} plots the data on the Log-Euclidean coordinates.

The \emph{PSR coordinates}, used in the right panel of the figure for the same data, come from the coordinates defined on a tangent space of the eigen-decomposition space $M(p)$. More precisely, given a reference point $(U,D) \in M(p)$, we use the local chart $(\mathcal{U}_{(U,D)}, \phi_{(U,D)})$ defined in (\ref{eq:local_chart}), followed by the vectorization via ${\rm vec}$ (see (\ref{eq:vectorization})), to determine a coordinate system.
To illustrate this concretely, let $p = 2$. Then
$\tilde\varphi_{(U,D)}(V,\Lambda) = ({\rm Log}(VU^T), {\rm Log}(\Lambda D^{-1}))=: (A,L) \in  \mathfrak{so}(p) \oplus {\rm Diag}(p)$.
 The first coordinate of $x_{V,\Lambda} := {\rm vec}(\phi_{(U,D)}(V,\Lambda))  \in \Real^3$ is the free parameter $a_{21}$ of $A$ (multiplied by the scale parameter $\sqrt{k}$), and corresponds to the rotation angle of $VU^T$ in radians (scaled by $\sqrt{k}$). The second and last coordinates of $x_{V,\Lambda}$ are the diagonal entries of $L$.
 Multiplying by $\sqrt{k}$ as above affords us the convenience that for any two $x,y$, the usual inner product $\langle x,y \rangle = x^T y$ corresponds to the Riemannian metric $g_M$ we have assumed on the tangent spaces of $M(p)$.

When representing SPD-valued data $X_1,\ldots,X_n \in {\rm Sym}^+(2)$ in PSR coordinates, we choose the reference point $(U,D)$  to be an arbitrarily chosen  PSR mean $\hat{m}^{\mathcal PSR}$ of the data. Care is needed since there are multiple eigen-decompositions corresponding to each observation $X_i$. For each $X_i$, an eigen-decomposition $m_i \in \mathcal{F}^{-1}(X_i) \subset M(p)$ is chosen so that $m_i$ has the smallest geodesic distance from $\hat{m}^{\mathcal PSR}$ among all elements of $\mathcal{F}^{-1}(X_i)$.
The right panel of  Fig. \ref{fig:Sec2Comparison_110} plots the same data as in the left panel, but in these PSR coordinates.

The AI mean and a PSR mean for this data set  are also plotted in Fig. \ref{fig:Sec2Comparison_110}. It can be seen that major modes of variation in the data are well described in terms of rotation angles and scaling, while the variation appears to be highly non-linear in  Log-Euclidean (LE) coordinates. As one might expect from this non-linearity, we observe that the AI mean is located far from the data, while the PSR mean appears to be a better representative of the data.

In the opposite direction, we also considered a data set sampled from an SPD-matrix log-normal distribution  \citep{Schwartzman2016}. Note that the SPD-matrix log-normal distributions on ${\rm Sym}^+(p)$ correspond to a multivariate normal distribution in {Log-Euclidean coordinates}. The data and their AI and PSR means are plotted in Fig. \ref{fig:Sec2Comparison_210}.
While the AI mean  is well approximated by the average of data in LE coordinates, the PSR mean (in LE coordinates) is also not far from this average.
Similarly, the PSR mean is approximately the average in PSR coordinates, and the AI mean is also not far.
Therefore, we may conclude that using the SR framework and PSR means is beneficial especially if variability in the sample (or in a population) is pronounced in terms of rotations, while the cost of using the SR framework is small for the log-normal case.

\begin{figure}[t]
\centering
\includegraphics[width = 0.95\textwidth]{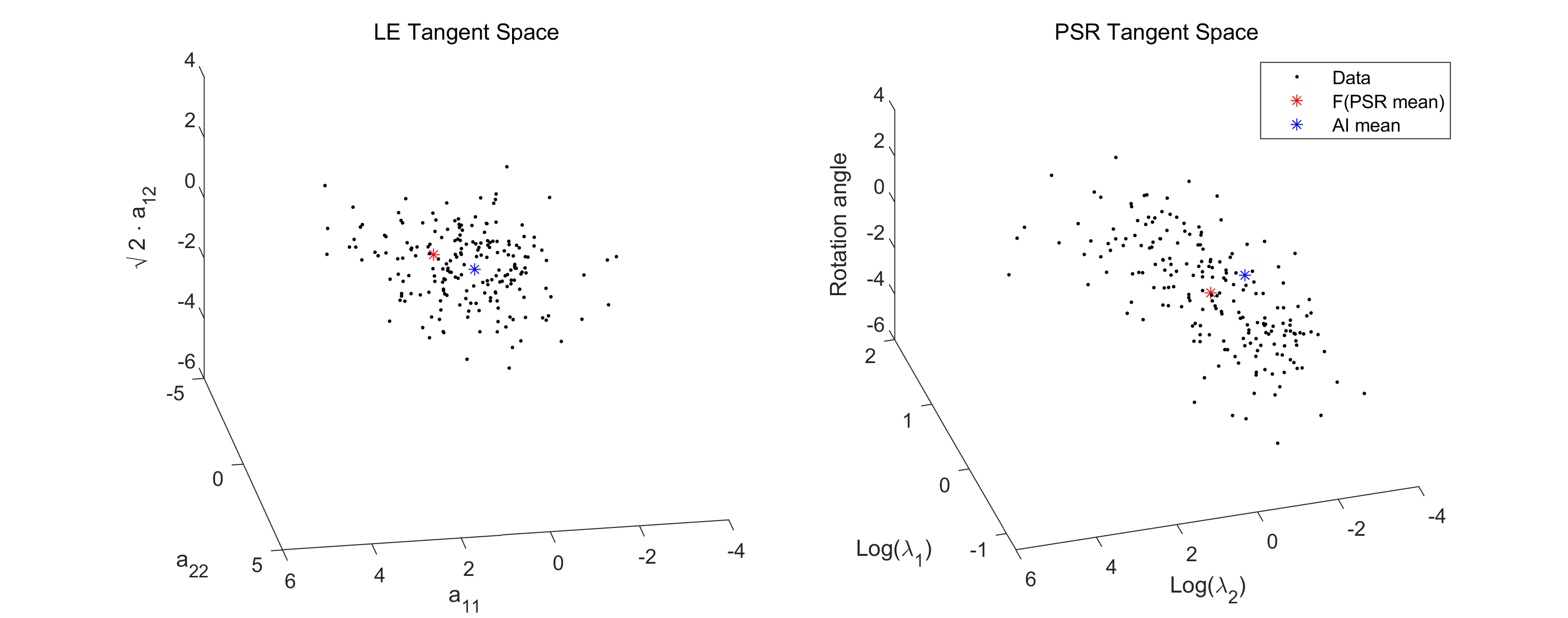}
\caption{A sample of SPD matrices shown in the Log-Euclidean (LE) coordinates (left) and the PSR coordinates (right), overlaid with the PSR mean and AI mean. For this data set, PSR mean appears to be a better representative for the data, while the AI mean does not lie in the data-dense region.
See Section~\ref{sec:simu} for details. }\label{fig:Sec2Comparison_210}
\end{figure}

\subsection{An application to multivariate tensor-based morphometry}\label{sec:tbm}
In \cite{Paquette2017}, the authors compared the lateral ventricular structure in the brains of 17 pre-term and 19 full-term infant children. After an MRI scan of a subject's brain was obtained and processed through an image processing pipeline, the shape data collected at 102816 vertices on the surfaces of their left and right ventricles were mapped onto the left and right ventricles of a template brain image, after which the $2 \times 2$ Jacobian matrix $J$ from that surface registration transformation was computed at each vertex for each subject. The deformation tensor $X=(J^TJ)^{1/2} \in {\rm Sym}^+(2)$ was then computed at each vertex for each subject. To summarize the structure of the data, there are 102816 vertices along the surfaces of the template ventricles, and at each vertex there are deformation tensors ($2 \times 2$ SPD matrices) from $n_1= 17$ pre-term and $n_2 = 19$ full-term infants. We will call these group 1 and group 2, respectively.

One way that the authors tested for differences in ventricular shape between the two groups was by performing two-sample location tests at each vertex via use of the log-Euclidean version of Hotelling's $T^2$ test statistic introduced in \cite{Lepore2008}. The log-Euclidean (LE) version of the $T^2$ test statistic is the squared Mahalanobis distance between the full-term and pre-term log-Euclidean sample means on the LE coordinates (defined in Section~\ref{sec:simu}). %

Similarly, one could also measure separation between groups by comparing their respective PSR means in PSR coordinates. For this two-group context, the reference point for the PSR coordinates is given by a PSR mean computed from pooled sample (with sample size $n_1+n_2$).

We have chosen vertex 75412 as an example to illustrate a scenario in which two groups have little separation in the LE coordinates but are well-separated in the PSR coordinates.
In the top row of \autoref{fig:TBM_tangent_space_obs}, tensors from the two groups as well as the group-wise LE and PSR means are plotted in their respective coordinates.
There is little visible separation between the two groups in the LE coordinates, while there is near-total separation in the PSR coordinates.

\begin{figure}[t]
\centering
\includegraphics[width = 1\textwidth, trim = 10 350 0 10]{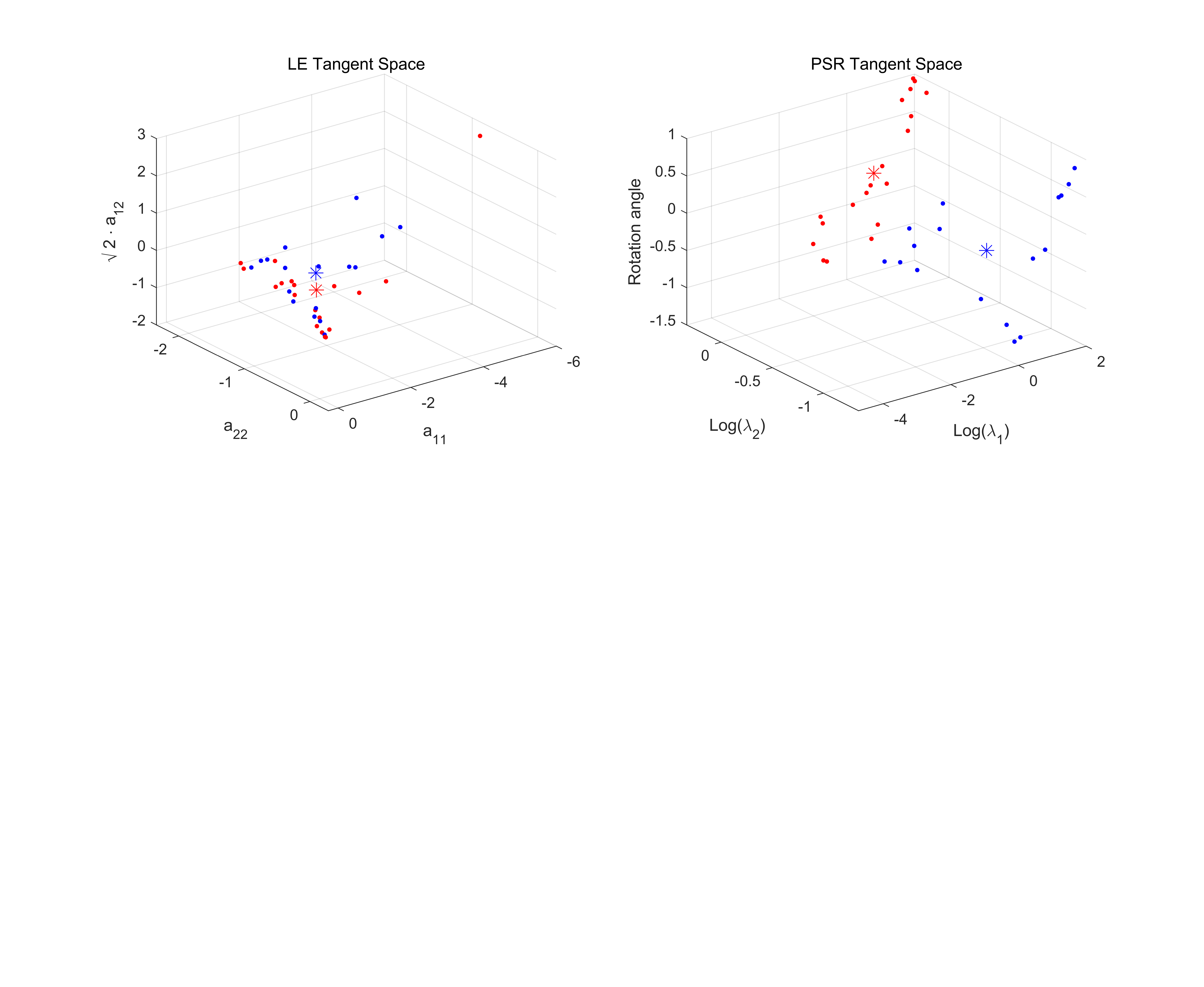}
\includegraphics[width = 1\textwidth, trim = 10 350 0 10]{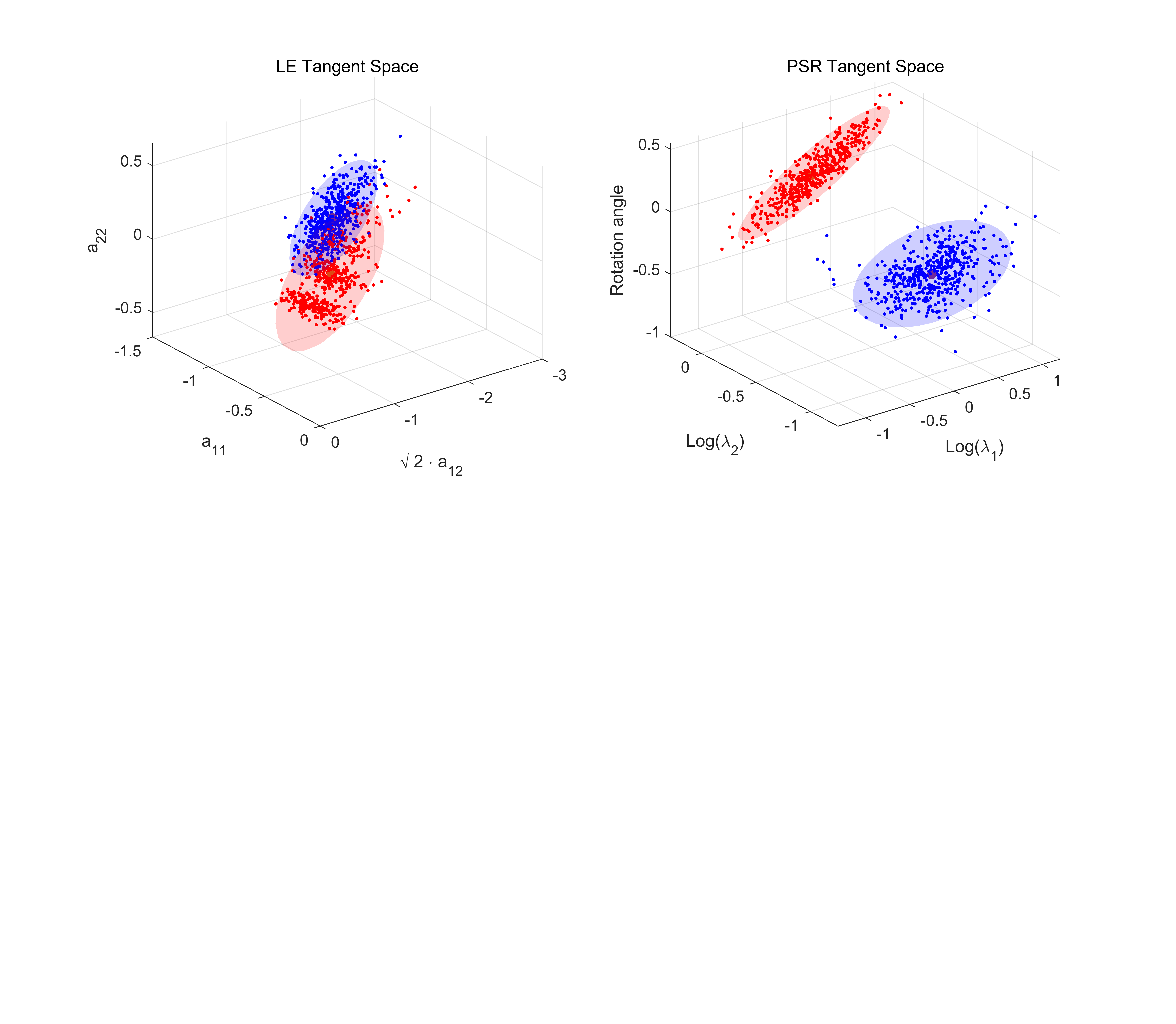}
\caption{Real data example. (Top row) Observations corresponding to Group 1 (and Group 2) are shown in blue (and red, respectively) dots.  The group-wise LE and PSR means are shown as the asterisks. (Bottom row) Bootstrap replications of the LE and PSR means (left and right panels, respectively) for each group, with the 95\% approximate confidence regions shown as transparent ellipsoids. See Section~\ref{sec:tbm} for details.
}\label{fig:TBM_tangent_space_obs}\label{fig:TBM_bootstrap_means}
\end{figure}

To visualize the sampling distributions of the group-wise means under the log-Euclidean and scaling-rotation frameworks, we computed 500 bootstrap sample means for each group, under both geometric frameworks. These 
are plotted in their respective tangent spaces in the bottom row of \autoref{fig:TBM_bootstrap_means}. (See also \autoref{fig:Bootstrap_PSR_Mean_QQ_plots} in Appendix \ref{sec:app_QQplot}, in which one can see that the (bootstrap) sampling distributions of the group-wise PSR means are approximately normal.)
The nonparametric bootstrap provides an estimate of standard errors of the sample means, from which (bootstrap-approximated) parametric 95\% confidence regions are obtained. For this, we assume normality, as suggested by the central limit theorem, Theorem~\ref{thm:PSR_mean_CLT}, and obtain an approximate 95\% confidence region given by $\{x \in \mathbb{R}^3: x \hat{\Sigma}^{-1} x^T \le \chi^2_{0.05,2}\}$, for each sample mean. Here, $\hat\Sigma$ is the sample covariance matrix of the bootstrap (group-wise LE or PSR) sample means, and $\chi^2_{0.05,2}$ is the 95\% quantile of the $\chi^2_2$ distribution. The resulting confidence regions are overlaid in the bottom row of  \autoref{fig:TBM_bootstrap_means} as well.
 As in the top row,
 there is considerable overlap between the LE confidence regions, while there is complete separation between the two confidence regions for the group-wise PSR sample means, especially along the direction of rotation angles. This example suggests that the scaling-rotation framework may be better at detecting group differences than other frameworks when most of the variability between the groups is due to rotation.

\section{Discussion}
We have presented the first statistical estimation methods for ${\rm Sym}^+(p)$ based on the scaling-rotation framework of \cite{Jung2015}. These estimation methods are intended to set the foundation for the development of scaling-rotation-framework-based statistical methods, such as testing the equality of two or more PSR means, testing for a variety of eigenvalue and eigenvector patterns of SPD matrices, and an analogue of principal component analysis for SPD-valued data.
The scaling-rotation framework should also be particularly useful for diffusion tensor processing since the eigenvectors and eigenvalues of a diffusion tensor model the principal directions and intensities of water diffusion at a given voxel, and are thus the primary objects of interest.

We recommend using the scaling-rotation estimation procedure presented here for $p=2,3$, since the number of eigen-decompositions of an SPD matrix from $S_p^{\rm top}$ grows rapidly with $p$.
One avenue for interesting future work is to develop computational procedures for higher $p$. Another avenue is to develop a proper two-sample or multi-sample testing framework, and dimension-reduction and regression methods using the eigen-decomposition spaces, and to establish asymptotic and non-asymptotic properties of these statistical methods, reflecting the structure of ${\rm Sym}^+(p)$ as a stratified space under eigen-decomposition. 

\begin{appendix}


\section{Discontinuity of $\dsr$}\label{sec:discontinuityofdsr}
While the scaling-rotation ``distance'' function $\dsr: \Symp \times \Symp \to [0,\infty)$ is continuous when restricted to $S_p^{\rm top} \times S_p^{\rm top}$, it is not so in general. Even the one-variable function $\dsr(\cdot, S)$, with a fixed $S \in S_p^{\rm top}$, has many discontinuities in lower strata. While it may be of interest to characterize the set of discontinuity, here we provide just an example.
For $0 < \theta < \theta' < \pi/4$, and $\lambda > 1$, let $S = R(\theta') {\rm diag}(e^\lambda,e^{-\lambda})R(\theta')^T$, where $R(\theta)$ is the $2\times 2$ rotation matrix corresponding to the counterclockwise rotation by angle $\theta$. For $n = 1,2,\ldots$, let $S_n = R(\theta) \textrm{diag}(e^{1/n}, e^{-1/n})R(\theta)^T$. For every $n$, it can be checked that
$(R(\theta'), {\rm diag}(e^\lambda,e^{-\lambda})) \in \mathcal{F}^{-1}(S)$ and
$(R(\theta), \textrm{diag}(e^{1/n}, e^{-1/n})) \in \mathcal{F}^{-1}(S_n)$ form a minimal pair, which implies that $\dsr(S_n,S)^2 = k(\theta'-\theta)^2 + 2(\lambda - \tfrac{1}{n})^2$. On the other hand, $\lim_{n\to\infty}S_n = I$ and $\dsr(I,S)^2 = 2\lambda^2$.
Thus,
$$
\lim_{n\to\infty} \dsr(S_n,S) = \{k(\theta'-\theta)^2 + 2\lambda^2\}^{1/2}  > 2\lambda^2 = \dsr(\lim_{n\to\infty}S_n,S),$$
and the function $\dsr(\cdot,S)$ is not continuous at $I$.

\section{Technical details, additional lemmas and proofs}\label{sec:proofs_in_appendix}

\subsection{Proofs for Section \ref{sec:3}}
\subsubsection{Proof of Theorem~\ref{thm:SRvsPSRequivalence}}

\begin{proof}[Proof of Theorem~\ref{thm:SRvsPSRequivalence}]


Let $Y \in E_n^{(\mathcal{SR})} \cap S_p^{\rm top}$, let $\tilde{Y} \in \mathcal{F}^{-1}(Y)$
 be an arbitrary eigen-decomposition of $Y$, let $\tilde{Z}=(U,D)\in M(p)$ be arbitrary, and let
 $Z=\mathcal{F}(\tilde{Z})$.
Since $Y$ has no repeated eigenvalues, it follows from \eqref{SR_distance_simplification} and \eqref{SR_PSR_ineq} that

%
\begin{equation}\label{SR_PSR_var_ineq}
\sum_{i=1}^n d^2_{\mathcal{PSR}}(X_i,\tilde{Y})
= \sum_{i=1}^n d^2_{\mathcal{SR}}(X_i,Y) \le \sum_{i=1}^n d^2_{\mathcal{SR}}(X_i,Z)
\le \sum_{i=1}^n d^2_{\mathcal{PSR}}(X_i,\tilde{Z}),
\end{equation}
implying that $\tilde{Y}\in E_n^{(\mathcal{PSR})}$.  Since the case where $E_n^{(\mathcal{SR})} \cap S_p^{\rm top} = \emptyset$ is trivial, we have shown (a).

 (b)  Suppose now
$\tilde{Z}\in E_n^{(\mathcal{PSR})}$. Then the first and fourth sums in
\eqref{SR_PSR_var_ineq} must be equal, so the two  inequalities must be equalities.
In particular, the second and third sums are equal, so $Z\in E_n^{(\mathcal{SR})}$
and hence $\tilde{Z}\in \mathcal{F}^{-1}(E_n^{(\mathcal{SR})})$.
 This shows $\mathcal{F}^{-1}(E_n^{(\mathcal{SR})}) \supset E_n^{(\mathcal{PSR})}$, which immediately implies
$\mathcal{F}(E_n^{(\mathcal{PSR})}) \subset E_n^{(\mathcal{SR})}$.

(c)  Assume that $E_n^{(\mathcal{SR})} \cap S_p^{\rm top} \neq \emptyset$ and $E_n^{(\mathcal{PSR})} \subset M^{\rm top}(p)$.
Then, using (b) and (a),
\begin{eqnarray}
\label{9.1prime-1}
E_n^{(\mathcal{PSR})} = E_n^{(\mathcal{PSR})} \cap  M^{\rm top}(p)
&\subset &\mathcal{F}^{-1}(E_n^{(\mathcal{SR})}) \cap  M^{\rm top}(p)
\\
\nonumber
&=& \mathcal{F}^{-1}(E_n^{(\mathcal{SR})}) \cap  \ \mathcal{F}^{-1}(S_p^{\rm top})\\
\nonumber
 &=&       \mathcal{F}^{-1} (E_n^{(\mathcal{SR})} \cap S_p^{\rm top}  )\\
&\subset& E_n^{(\mathcal{PSR})}.
\label{9.1prime-2}
\end{eqnarray}
Hence the inclusions in \eqref{9.1prime-1} and\eqref{9.1prime-2} are equalities.

Finally, note that  (\ref{SR_PSR_var_ineq}) holds with the finite summation replaced by the integration with respect to the probability measure $P$, provided that $\fpsr(U,D) < \infty$
is finite for any $(U,D) \in M(p)$,
which also guarantees that
$\fsr(Z) < \infty$
  for any $Z \in \Symp$. Since these conditions are satisfied by Lemma~\ref{lem:finite_variance}, the statements (a)--(c) hold for $\Ensr$ and $\Enpsr$ replaced by $\Esr$ and $\Epsr$, respectively.
\end{proof}

\subsubsection{Proof of Theorem~\ref{thm:how_to_tell_whether_PSR=SR}}
\begin{proof}[Proof of Theorem~\ref{thm:how_to_tell_whether_PSR=SR}] We give a proof for (a) and (b). Assertions (c) and (d) can be  verified similarly.

For (a), consider the case where  the inequality is strict, i.e.,    $\fsr(\mathcal{F}(m^{\mathcal{PSR}})) < \min_{\Sigma \in S_p^{\rm lwr}} \fsr(\Sigma)$. Then no scaling-rotation mean can lie in $S_p^{\rm lwr}$, but since scaling-rotation means always exist, $\Esr \subset S_p^{\rm top}$. By Theorem~\ref{thm:SRvsPSRequivalence}, $\mathcal{F}(m^{\mathcal{PSR}}) \in \Esr$. Now consider the situation where
\begin{equation}\label{eq:condition_a_equal_PSR=SR=low}
\fsr(\mathcal{F}(m^{\mathcal{PSR}})) = \min_{\Sigma \in S_p^{\rm lwr}} \fsr(\Sigma).
\end{equation}
Suppose that no scaling-rotation mean lies in $S_p^{\rm lwr}$. Then $\Esr \subset S_p^{\rm top}$ and, by Theorem~\ref{thm:SRvsPSRequivalence}, $\mathcal{F}(m^{\mathcal{PSR}}) \in \Esr$. Since these arguments and (\ref{eq:condition_a_equal_PSR=SR=low}) contradict, there must be a scaling-rotation mean in $S_p^{\rm lwr}$. Moreover, by (\ref{eq:condition_a_equal_PSR=SR=low}), $\mathcal{F}(m^{\mathcal{PSR}}) \in \Esr$ as well.

The hypothesis of (b) leads that $\mathcal{F}(m^{\mathcal{PSR}}) \notin \Esr$. By the inverse of Theorem~\ref{thm:SRvsPSRequivalence}(b), $\Esr \cap S_p^{\rm top} = \emptyset$.
\end{proof}

\subsubsection{Proof of Theorem~\ref{thm:avoid_low_strata_SR}}

We need several technical lemmas.
For $(U,D) \in M(p)$, define
$$\tilde{\delta}(U,D) = \inf\{d_M( (U,D), (V, \Lambda)) : (V, \Lambda) \in M(p) \setminus M^{\rm top}(p)\},$$
where the infimum can be replaced by minimum. The minimum is achieved since $M(p) \setminus M^{\rm top}(p)$ is closed in $M(p)$, and as a finite-dimensional manifold, $M(p)$ is locally compact.

\begin{lemma}\label{lem:13.1}
For $(U,D) \in M(p)$, $\tilde{\delta}(U,D) = \min\{d_{\mathcal{D}^+}(D, \Lambda): \Lambda \in {\rm Diag}^+(p) \setminus D_+^{\rm top}(p) \}$, where $D_+^{\rm top}(p)$ is the subset of ${\rm Diag}^+(p)$ consisting of distinct diagonal entries.
\end{lemma}

\begin{proof}
Write $M^{\lwr}(p) := M(p) \setminus M^{\rm top}(p)$ and $ D_+^{\rm lwr}(p) := {\rm Diag}^+(p) \setminus D_+^{\rm top}(p)$.
Clearly $ \inf\{d_M( (U,D), (V, \Lambda)) : (V, \Lambda) \in M^{\lwr}(p)\} \ge \inf\{d_{\mathcal{D}^+}(D, \Lambda): \Lambda \in D_+^{\rm lwr}(p) \}$. Conversely, if $\Lambda \in D_+^{\rm lwr}(p)$, then $(U,\Lambda) \in M^{\lwr}(p)$. So,
\begin{align*}
  \inf\{d_M( (U,D), (V, \Lambda)) : (V, \Lambda) \in M^{\lwr}(p))\} & \le
    \inf\{d_M( (U,D), (U, \Lambda)) :   \Lambda \in D_+^{\lwr}(p)\}   \\
   &  = \inf\{d_{\mathcal{D}^+}( D,  \Lambda) :   \Lambda \in D_+^{\lwr}(p)\} \\
   &  = \min\{d_{\mathcal{D}^+}( D,  \Lambda) :   \Lambda \in D_+^{\lwr}(p)\},
\end{align*}
where the minimum is achieved since $D_+^{\lwr}(p)$ is a closed subset of the locally compact metric space $({\rm Diag}^+(p), d_{\mathcal{D}^+})$.
\end{proof}

\begin{lemma}\label{lem:13.2}
The function
$\tilde\delta(U,D)$
is constant on fibers of $\mathcal{F}$. That is, for each $S \in \Symp$, the value of $\tilde{\delta}(U,D)$ is independent of the choice of $(U,D)\in F^{-1}(S)$.
%
\end{lemma}

\begin{proof}
Let $S\in \Symp$ and let $(U,D), (U_1,D_1)\in F^{-1}(S).$ Then
  $D_1 = h\cdot D$ for some $h \in \mathcal{G}(p)$. But the set $D_+^{\lwr}(p)$ and the metric $d_{\mathcal{D}^+}$ are invariant under the action of $\mathcal{G}(p)$, defined in Section~\ref{sec:fiber}, so
\begin{align*}
  \{ d_{\mathcal{D}^+} (h \cdot D, \Lambda): \Lambda \in  D_+^{\lwr}(p) \} & =  \{ d_{\mathcal{D}^+} (h \cdot D, h \cdot \Lambda): \Lambda \in D_+^{\lwr}(p) \} \\
    &  =  \{ d_{\mathcal{D}^+} (D, \Lambda): \Lambda \in  D_+^{\lwr}(p) \} .
\end{align*}
Hence by Lemma~\ref{lem:13.1}, $\tilde\delta(U_1,D_1)$ and $\tilde\delta(U,D)$ are the infimum of the same set of real numbers.
\end{proof}

The following lemma shows a relation between $\delta(S)$ and $\tilde{\delta}(U,D)$.
\begin{lemma}\label{lem:13.3}
For any $S \in \Symp$, $\delta(S) = \tilde\delta(U,D)$ for any $(U,D) \in \mathcal{F}^{-1}(S)$.
\end{lemma}

\begin{proof}[Proof  of Lemma~\ref{lem:13.3}]
Recall that we write $S_p^{\lwr} = {\rm Sym}^{+}(p) \setminus S_p^{\rm top}$.
\begin{align*}
  \delta(S) & = \inf\{ d_{\mathcal{SR}}(S,S') : S' \in S_p^{\lwr} \} \\
   & = \inf\{ \inf \{ d_M( (U,D) , (V, \Lambda)): (U,D) \in \mathcal{F}^{-1}(S), (V, \Lambda) \in \mathcal{F}^{-1}(S')   \} : S' \in S_p^{\lwr} \} \\
   & = \inf \{ d_M( (U,D) , (V, \Lambda)):(U,D) \in \mathcal{F}^{-1}(S), (V, \Lambda) \in  M^{\lwr}(p)  \} \\
   & =  \inf \{ \inf \{ d_M( (U,D) , (V, \Lambda)): (V, \Lambda) \in  M^{\lwr}(p)  \} :(U,D) \in \mathcal{F}^{-1}(S)  \} \\
   & = \inf \{ \tilde\delta(U,D) :(U,D) \in \mathcal{F}^{-1}(S)  \}.
\end{align*}
The above and Lemma~\ref{lem:13.2} give the result.
\end{proof}

By Lemmas~\ref{lem:13.1}---\ref{lem:13.3}, we have $\delta(S) > 0$ if and only if $S \in S_p^{\rm top}$.



\begin{lemma}\label{lem:13.6}
 Let $S_0 \in S_p^{\rm top}$, let $r>0$, and write $\bar{B}_r^{\dsr}(S_0) = \{Y \in \Symp:  d_{\mathcal{SR}}(Y,S_0) \le r \}$.
 \begin{itemize}
\item [(a)] If $S \in \bar{B}_r^{\dsr}(S_0)$, then $\delta(S) \ge \delta(S_0) - r$.
\item [(b)] If $r < \delta(S_0)$, then $\bar{B}_r^{\dsr}(S_0) \subset S_p^{\rm top}$.
\item [(c)] If $r < \delta(S_0)/3$, then for any $S, S' \in \bar{B}_r^{\dsr}(S_0)$, and $S_{\rm lwr} \in {\rm Sym}^{+}(p) \setminus S_p^{\rm top}$,
    $$d_{\mathcal{SR}}(S,S') < d_{\mathcal{SR}}(S,S_{\rm lwr}).$$
\end{itemize}
\end{lemma}

\begin{proof}
(a) Let $S \in \bar{B}_{r}^{\dsr}(S_0)$. If $(U,D) \in \mathcal{F}^{-1}(S)$ and $\tilde{S}_0 := (U_0, D_0) \in \mathcal{F}^{-1}(S_0)$, then
by Lemmas~\ref{lem:13.1} and \ref{lem:13.3}, $\tilde{\delta}(U,D) = \min\{d_{\mathcal{D}^+}(D, \Lambda): \Lambda \in   D_+^{\lwr}(p) \}$. Since
$({\rm Diag}^+(p), d_{\mathcal{D}^+})$ is a metric space, $d_{\mathcal{D}^+}(D, \Lambda) \ge d_{\mathcal{D}^+}(D_0, \Lambda) - d_{\mathcal{D}^+}(D_0, D)$ for any $D, D_0, \Lambda \in {\rm Diag}^+(p)$. Thus
\begin{align*}
  \delta(S) & \ge \inf\{ d_{\mathcal{D}^+}(D_0, \Lambda) - d_{\mathcal{D}^+}(D_0, D) : \Lambda \in   D_+^{\lwr}(p) \}  \\
    & = \inf\{ d_{\mathcal{D}^+}(D_0, \Lambda) : \Lambda \in   D_+^{\lwr}(p) \}  - d_{\mathcal{D}^+}(D_0, D)   \\
    & = \delta(S_0) -  d_{\mathcal{D}^+}(D_0, D)  \\
    & \ge \delta(S_0) -  d_{\mathcal{SR}}(S_0, S) \\
    & \ge \delta(S_0) - r.
\end{align*}

(b) If $r < \delta(S_0)$ and $S \in \bar{B}_r^{\dsr}(S_0)$, then by part (a),
   $\delta(S) \ge \delta(S_0) - r > 0$, so $S \in S_p^{\rm top}$.

(c) By part (b), since $r < \delta(S_0)/3 < \delta(S_0)$, $\bar{B} := \bar{B}_r^{\dsr}(S_0) \subset  S_p^{\rm top}$, and $\bar{B}$ is a closed ball in the metric space $(S_p^{\rm top}, d_{\mathcal{SR}})$. Hence, for any $S,S' \in \bar{B}$,
$$d_{\mathcal{SR}}(S,S') \le d_{\mathcal{SR}}(S,S_0) + d_{\mathcal{SR}}(S_0, S') \le 2r < 2\delta(S_0)/3.$$
But by Lemma~\ref{lem:13.3} and  part (a), for $S_{\rm lwr} \in {\rm Sym}^{+}(p) \setminus S_p^{\rm top}$,
$$d_{\mathcal{SR}}(S,S_{\rm lwr}) \ge \delta(S) \ge \delta(S_0) - r > \delta(S_0) - \delta(S_0)/3 = 2\delta(S_0)/3.$$
Hence $d_{\mathcal{SR}}(S,S') \le 2\delta(S_0)/3 < d_{\mathcal{SR}}(S,S_{\rm lwr})$.
\end{proof}

The proof of Theorem~\ref{thm:avoid_low_strata_SR} heavily depends on Lemma~\ref{lem:13.6}(c).

\begin{proof}[Proof of Theorem~\ref{thm:avoid_low_strata_SR}]
The random variable $X$ in the hypothesis of the theorem lies in $\bar{B} := \bar{B}_r^{\dsr}(S_0)$ with probability 1. Thus, by  Lemma~\ref{lem:13.6}(c), for any $S \in \bar{B}$ and $S_{\rm lwr} \in {\rm Sym}^{+}(p) \setminus S_p^{\rm top}$,
 $$\fsr(S) < \fsr(S_{\rm lwr}).$$
Hence, no element of ${\rm Sym}^{+}(p) \setminus S_p^{\rm top}$ can be a minimizer of $f_{{\mathcal{SR}}}$. Since the set of  minimizers of the function $f_{{\mathcal{SR}}}$ is exactly $\Esr$, and $\Esr$ is non-empty, $\Esr \subset S_p^{\rm top}$.

The second part of the theorem can be shown similarly  by Lemma~\ref{lem:13.6}(c), but with the function $\fnsr(\cdot)$ 
defined with respect to the data $X_1,\ldots, X_n$.
\end{proof}

\subsection{Proofs and technical details for Section 4.1} \label{sec:prop:lsc_proof}
\subsubsection{Proof of Lemma~\ref{lem:dpsr_continuity}}
\begin{proof}[Proof of Lemma~\ref{lem:dpsr_continuity}]
(a) Define the map $\rho: M(p) \times M(p) \to [0,\infty)$ as
\[ \rho((U',D'),(U,D)) = \min_{h \in \mathcal{G}(p)} d_M((U' h^{-1},h\cdot D'),(U,D)). \]
Note that $\rho((U',D'),(U,D)) = \rho(h \cdot (U',D')),(U,D))$ for any $h \in \mathcal{G}(p)$.  Hence for each $(X,(U,D)) \in S_p^{\rm top} \times M(p)$, the function
$\rho$ is constant on the set $\mathcal{F}^{-1}(X)\times \{(U,D)\}$.  Therefore the restriction of
$\rho$
to the domain $M(p)^{\rm top} \times M(p)$
induces a function on $S_p^{\rm top} \times M(p)$, which by definition
is precisely our function $d_{\mathcal{PSR}}$. (Here, $M(p)^{\rm top} = \mathcal{F}^{-1}(S_p^{\rm top})$.)
For each $h\in \mathcal{G}(p)$, the function $((U',D'),(U,D)) \mapsto
d_M((U' h^{-1},h\cdot D'),(U,D))$ is continuous on $M(p) \times M(p)$.
Therefore
$\rho$ is also continuous on $M(p) \times M(p)$
since it is the minimum of a finite number of
continuous functions,
which implies that the restriction of $\rho$ to $M(p)^{\rm top} \times M(p)$ is continuous.
Hence the induced function $d_{\mathcal{PSR}}$ on
$S_p^{\rm top} \times M(p)$
is also continuous.

(b)
For any non-empty subset $A$ of a metric space $(M,d)$,
the triangle inequality implies that $|d(A,y) - d(A,y')| \le d(y,y')$
for any $y,y' \in M$, where
 $d(A,y) = \inf_{x \in A} d(x,y)$.
For any $S\in {\rm Sym}^+(p)$,
applying the above fact to the subset $\mathcal{F}^{-1}(S)$ of the metric space
$(M(p), d_M)$,
and noting that $d_{\mathcal{PSR}}(S,m) = d_M(\mathcal{F}^{-1}(S), m)$, the conclusion follows.
\end{proof}

\subsubsection{Background work on semicontinuous functions}
Recall Definition~\ref{def:lsc}.

\begin{prop}\label{complsc_unif}
Let $X$ be a topological space, let $Y$ be a set, and let $f:X\times Y \to\bfr$.
Let $J\subset \bfr$
be a set containing ${\rm range}(f)$, and let $g:J\to\bfr$ be a
(non-strictly) increasing, uniformly continuous
function.

\begin{itemize}
\item[(a)]
Assume that
$f:X\times Y \to\bfr$ is LSC in its first variable,
uniformly with respect to its second variable.
Then so is
$g\circ f: X\times Y\to \bfr.$

\item[(b)] Assume that $Y$ is a topological space and that
$f:X\times Y \to\bfr$ is LSC in its first variable,
locally uniformly with respect to its second variable.
Then so is
$g\circ f: X\times Y\to \bfr.$
\end{itemize}

\end{prop}
\begin{proof}[Proof of Proposition~\ref{complsc_unif}]
(a) Let $x_0\in X$ and let $\e>0$.
Since $g$ is uniformly continuous, we may select $\d>0$ such that whenever $z_1,z_2\in J$
and $|z_1-z_2|<\d$, we have $g(z_2)>g(z_1)-\e$.
By the hypothesis on $f$, we may select
an open neighborhood $U$ of $x_0$ such that
$f(x,y)> f(x_0,y)-\d$ for all $x\in U$ and all $y\in Y$.

Let $x\in U$ and $y\in Y$.  Then either (i) $f(x_0,y)-\d<f(x,y)\leq f(x_0,y)$ or (ii) $f(x,y)>f(x_0,y)$.  In case (i),
$|f(x,y)-f(x_0,y)|<\d,$ so $g(f(x,y))>g(f(x_0,y))-\e.$
 In case (ii), since $g$ is
increasing, $g(f(x,y))\geq g(f(x_0,y))>g(f(x_0,y))-\e$.  Hence in both cases, $g(f(x))>g(f(x_0))-\e$.

Thus $g\circ f$ is LSC in its first variable, uniformly with respect to its second.

(b) Follows immediately from part (a) and Definition \ref{def:lsc}(ii).
\end{proof}

\begin{corollary} \label{sqrtlsc_unif}
Let $X$ be a topological space, let $Y$ be a set, and let $f:X\times Y \to [0,\infty)$.

\begin{itemize}
\item[(a)]
Assume that
$f:X\times Y \to\bfr$ is LSC in its first variable,
uniformly with respect to its second variable.
Then so is $\sqrt{f}$.

\item[(b)] Assume that $Y$ is a topological space and that
$f:X\times Y \to\bfr$ is LSC in its first variable,
{\em locally} uniformly with respect to its second variable.
Then so is $\sqrt{f}$.
\end{itemize}

\end{corollary}

\begin{proof}[Proof of Corollary~\ref{sqrtlsc_unif}]
  The square-root function $[0,\infty)\to [0,\infty)$ is uniformly continuous (since
$\sqrt{x+\d}-\sqrt{x}\leq \sqrt{\d}$ for $x,\d\geq 0$).  Hence the results follow from
Proposition \ref{complsc_unif}.
\end{proof}
\subsubsection{Background work on $M(p)$ and $\mathcal{F}$}
The strata of $\dpp$ (and the strata of $\tim(p)$)
are partially ordered by identifying
a stratum with the corresponding partition of $\pset$.  If $\T_\J\subset\dpp$ denotes the stratum labeled
by $\J$, then we have the following relations (the first of which is definition)
\be
\T_{\J_1}\leq \T_{\J_2} \iff \J_1\leq \J_2
\iff G_{\J_1}\supset G_{\J_2}.
\ee
\cite[See][Section 2.2.]{Groisser2017}

In Lemma~\ref{lem:delta_strat} and throughout, $B_\d^\D(D) = \{ \Lambda \in \dpp: d_{\mathcal{D}^+}(D,\Lambda) < \delta \}$ denotes the open ball in the metric space $(\dpp, d_{\mathcal{D}^+})$.

\begin{lemma}\label{lem:delta_strat} (a) Every $D\in \dpp$.
has an open neighborhood that intersects only strata that are at least as high as the stratum of $D$.
I.e. for any $D\in\dpp$ there is an open ball $B_\d^\D(D)$ such
that
\be\label{stratlsc-2}
\mbox{\rm if $\T$ is a stratum of $\dpp$ for which  $\T\cap B_\d^\D(D)\neq\emptyset$, then $\T\geq \T_D$};
\ee
equivalently,
\be\label{stratlsc}
\mbox{if $\J\in \partpset$ and $\T_{\J}\cap B_\d^\D(D)\neq\emptyset$, then $\J\geq \J_D$}.
\ee

(b) There is a function $\dstrat: \sympp\to (0,\infty)$
such that for all $S\in \sympp$ and all $(U,D)\in \Fc^{-1}(S)$,
\eqref{stratlsc-2} (equivalently, \eqref{stratlsc}) holds with $\d=\dstrat(S)$.
\end{lemma}

\begin{proof}[Proof of Lemma~\ref{lem:delta_strat}]
(a) This follows from the fact that any strict
eigenvalue-inequalities holding at $D$ persist on a small enough open neighborhood of $D$.

 (b)  Let $S\in \sympp$, and let $D$ be a diagonal matrix appearing in
some eigencomposition of $S$.  The set of such diagonal matrices is
$\{\pi\cdot D: \pi\in\cals_p\}$.
Because the action of $\cals_p$ on $\dpp$ is isometric, if $\d>0$ is such that \eqref{stratlsc}
holds for a given $D$, then for any $\pi\in\cals_p$, \eqref{stratlsc} holds with $D$ replaced by
$\pi\cdot D$ (with the same $\d$).  Thus any such $\d$ depends only on $S$,
not on any chosen eigendecomposition.
\end{proof}

\begin{defn}[just notational] \label{def_delta_strat}
\rm For each $S\in \sympp$, let $\dstrat(S)$
be as in Lemma \ref{lem:delta_strat}(b).
\end{defn}

\begin{prop}\label{fproper} $\Fc$ is a proper map (i.e. the inverse of any compact set is compact).
\end{prop}

\begin{proof}[Proof of Proposition~\ref{fproper}]
  Let $\l_{\max}, \l_{\min}:\sympp\to\bfr$ be the functions carrying $S\in\sympp$ to its largest and smallest eigenvalues,
respectively.  As is well known, these functions are continuous.

Let $K\subset\sympp$ be a nonempty compact set.  Then $K$ is closed, and since $\Fc$ is continuous, $\Fc^{-1}(K)$ is closed.

Let $\l_{\max}^K$ (respectively $\l_{\min}^K$) denote the maximum (resp. minimum) value of $\l_{\max}$ (resp. $\l_{\min}$) achieved on $K$, and let $\tilde{K}_{\D}=\{D\in \dpp: D_{ii}\in
[\l_{\min}^K, \l_{\max}^K], \ 1\leq i\leq p\}.$  Note that $\Fc^{-1}(K)\subset \sop\times \tilde{K}_\D$, a compact subset of $\tim(p).$
Hence $\Fc^{-1}(K)$ is a closed subset of a compact set, and is therefore compact.
\end{proof}

Next few results are needed because $\Fc$ is not an open map. (A map is {\em open} if it carries open sets to open sets.)

\begin{lemma}[{\bf ``Slice lemma''}]\label{slicelemma}  Let $\L\in \dpp,$ let $\gl\subset\lsop$ be the Lie algebra of $G_\L$,
and let $\glp\subset\lsop$ be the orthogonal complement of $\gl$ in $\lsop$ (with respect to $g_\sop|_I$, a multiple of
the Frobenius inner product).
Define $\nu_\L:=\glp\plus \fd(p)$  and $n_\L:=\dim(\nu_\L)=\dim(\glp)+p,$
and define
$\Psi:\nu_\L\to \sympp$ by
\ben
\Psi(A,L)= e^A \L e^L e^{-A}\ ;
\een
note that $\Psi$ is $C^\infty$ and that $\Psi(0,0)=\L$.
On $\fgl(p,\bfr)$ or any of its subspaces
let $\|\ \|_\fr$ denote the Frobenius norm;
on $\lsop$ let $\|\ \|_\so=\frac{1}{\sqrt{2}}\|\ \|_\fr$\,;
and on $\nu_\L$ let $\|\ \|_\tge $ be the norm defined by
$\|(A,L)\|_\tge=(k\|A\|_\so^2+\|L\|_\fr^2)^{1/2}.$

\indenum There exist ${\d_2}>0, c>0$,
and an open neighborhood $\tilde{\H}_\L$ of $(0,0)$ in $\nu_\L$,
such that the $(\Sym(p),\|\ \|_\fr)$-open ball $ B_{\d_2}^\frob(\L)$
lies in $\sympp$ and

\begin{itemize}
\item[(a)] $\Psi|_{\tilde{\H}_\L}$ is an embedding;

\ss
\item[(b)] $\H_\L:=\Psi(\tilde{\H}_\L)$  is an $n_\L$-dimensional
submanifold of $\sympp$ containing $\L$;

\ss
\item[(c)]  $\H_\L={\rm image}(\Psi)\cap B_{\d_2}^\frob(\L);$


\ss
\item[(d)]   letting $\Phi=\Psi|_{\tilde{\H}_\L}$, viewed as a map
$\tilde{\H}_\L\to \H_\L$,
\be\label{hest1}
\|\Phi^{-1}(S')\|_\tge \leq c\|S'-\L\|_\fr\ \ \ \mbox{\rm for all
$S'\in \H_\L$}\ ;
\ee
and

\ss
\item[(e)]
for every $S'\in B_{\d_2}^\frob(\L)\subset\sympp$, there exist $R\in G_\L^0$, $A\in \glp$, and
\lnbrk
$L\in \fd(p)$ such that
\bestar
\nonumber
S'&=&R\, \Psi(A,L)\, R^T, \\
\|A\|_\so&=&\dso(e^A,I),  \ \ \mbox{and}\\
\|(A,\L)\|_\tge &\leq & c\|S'-\L\| .
\nonumber
\eestar

\end{itemize}
\end{lemma}

\begin{proof}[Proof of Lemma~\ref{slicelemma}]
For any $p\times p$ symmetric matrix $S$ and antisymmetric matrix $A$, the commutator $[S,A]$ is antisymmetric.
 Hence for the diagonal matrix $\L$, the map $\glpr\to\glpr$
defined by $A\mapsto [\L,A]$ restricts to a linear map $\ad_\L: \lsop\to\Sym(p)$.
Recall that the subalgebra $\gl\subset\lsop$ consists precisely of those elements of $\lsop$ that commute with $\L$.
Thus $\gl=\ker(\ad_\L)$, and the further-restricted map $\ad'_\L:=\ad_\L|_{\glp}$ is injective.

The  derivative of $\Psi$ at $(0,0)$ is the linear map
$d\Psi|_{(0,0)}:\nu_\L\to \Sym(p)$
given by
\be\label{dpsi}
d\Psi|_{(0,0)}(A,L) =[A,\L] +\L L
 = -\ad_\L'(A)+ \L L.
\ee

Let $A\in\glp$ and $L\in \fd(p)$.
Since the diagonal entries of $A$ are all zero, so are the diagonal entries of $A\L, \ \L A$, and $[A,\L]$.  Hence $[A,\L]$ is Frobenius-orthogonal to the diagonal matrix $\L L$.
Thus if $d\Psi|_{(0,0)}(A,L) =0$, equation \eqref{dpsi} implies that $\ad_\L'(A)=0$ and $\L L=0$.  Since $\ad_\L'$ is injective and $\L$ is
invertible, the latter pair of equations implies $A=0$ and $L=0$.  Thus $d\Psi|_{(0,0)}$ is injective.

Since $\Psi$ is continuously differentiable and $d\Psi|_{(0,0)}$ is injective, and $\sympp$ is an open subset of
the vector space $\Sym(p)$,
a standard application of the Inverse Function Theorem implies the existence of $\d_2, c,$ and $\tilde{\H}_\L$
for which properties (a)--(d) hold and for which  $B_{\d_2}^\frob(\L) \subset\sympp$. Note that, modulo the value of $c$, conclusion (d) does not depend on
our choices of norms,
since all norms on a finite-dimensional vector space are equivalent.

For conclusion (e), let $S'\in B_{\d_2}^\frob(\L)$ and let $(U,D)\in \Fc^{-1}(S')$.  Since $G_\L^0$ is compact, there exists an element  $R\in G_\L^0$ achieving the $\dso$-distance from $U$ to $G_\L^0$.  Since the Riemannian exponential map
$\exp_R: T_R(\sop)\to \sop$ is surjective, and the tangent space $T_R(G_\L^0)$
is
$\{RA\in \glpr : A\in \gl\}$,  the minimal-distance condition (together with our choice of inner
product on  $\lsop$) implies that $U=Re^A$ for some
$A\in \glp$ with $\|A\|=\dso(e^A,I).$  Letting $L=\log(D\L^{-1})$, we then have $S'=Re^A \L e^L e^{-A} R^{-1}=R\Psi(A,L)R^{-1}$.
Since the Frobenius norm on $\Sym(p)$ is invariant under
the action of $\sop$ (the map $(R,S)\mapsto RSR^T$),
\bearray
\nonumber
{\d_2}\ > \ \| S'-\L\|_\fr
&=& \| R\Psi(A,L) R^{-1}-\L\|_\fr \\
\nonumber
&=& \| \Psi(A,L) -R^{-1}\L R\|_\fr \\
&=& \| \Psi(A,L) -\L\|_\fr \ \ \ \ \ \mbox{(since $R\in G_\L)$}.
\label{getridofR}
\eearray

Thus $\Psi(A,L)\in B_{\d_2}^\frob(S) \cap {\rm image}(\Psi)=\H_\L$, and $(A,L)=\Phi^{-1}(\Psi(A,L))$.
Hence \eqref{hest1} and \eqref{getridofR} imply that
\ben
|(A,L)\|_\tge<c\| \Psi(A,L) -\L\|_\fr
=c\|S'-\L\|_\fr. \qedhere
\een
\end{proof}

\begin{remark}\label{nul_rmk}
The geometric significance of the space $\nu_\L$ in Lemma \ref{slicelemma} is the following. The manifold $\tim(p)=\sop\times\dpp$ is a Lie group
with identity element $e=(I,I)$ and Lie algebra $T_e(\tim(p))=\lsop\plus\fd(p)$.
For any $S\in \sympp$ and $(V,\L)\in \Fc^{-1}(S)$, let $\nu_{(V,\L)}(\Fc^{-1}(S))$ be the normal space to the fiber $\Fc^{-1}(S)$ at $(V,\L)$---i.e. the orthogonal
complement of $T_{(V,\L)}(\Fc^{-1}(S))\subset T_{(V,\L)}\tim(p)$
w.r.t. $\tg_{(V,\L)}$.
The space $\nu_\L\subset \lsop\plus\fd(p)=T_e(\tim(p))$ in Lemma \ref{slicelemma} is simply the image of $\nu_{(V,\L)}(\Fc^{-1}(S))$ under the map
$ T_{(V,\L)}(\tim(p))) \to T_e (\tim(p))$ induced by left-translation by $(V^{-1},\L^{-1}).$
\end{remark}

\begin{notation}  \rm Given any metric space $(X,d_X),$  any $Y\subset X$, and
any $\e>0$, we let $N_\e(Y)$ denote the $\e$-neighborhood of $Y$ in $X$:
\ben
N_\e(Y)=\{x\in X : d_X(x,Y)<\e\}
\een
\end{notation}

It is easily seen that
\be\label{e-neighborhoods}
N_\e(Y)=\Union_{y\in Y} B_\e^X(y).
\ee

\begin{corollary}\label{cor:a11}
\label{slicecor}  Let $S\in\sympp,$  let $(V,\L)\in \Fc^{-1}(S)$,
and let $\calc(V,\L)$ denote the connected component of $\Fc^{-1}(S)$ containing $(V,\L)$.
Let $VG_\L^0$ denote the
set \lnbrk $\{VR : R\in G_\L^0\}$.

There exist $\d_2=\d_2(\L)>0$ and $c_1=c_1(\L)>0$ (depending only on $\L$, not
$V$)
such that for all $\d\in (0,\d_2]$,
\bearray
\underbrace{B_{\d}^\frob(S)}_{\mbox{\scriptsize \rm in $\sympp$}}
\label{f(unionprodballs)}
&\subset&
F\big(\Union_{R\in G_\L^0}\big(B_{c\d/\sqrt{k}}^{SO}(VR)\times B_{c\d}^\D(\L) \big)\ \big)\\
\label{f(normneighborhood_x_ball)}
&=&
 \F\big( \underbrace{N_{c\d/\sqrt{k}}(VG_\L^0)}_{\mbox{\rm
\scriptsize in $\sop$}}\times B^\D_{c\d}(\L) \big)\\
&\subset &
\label{f(unionballs)}
 F\big(\Union_{\tS\in \calc(V,\L)}B_{c_1\d}^{\td}(\tS)\ \big)\\
&=&
F(\underbrace{N_{c_1\d}(\calc(V,\L))}_{\mbox{\rm
\scriptsize in $\tim(p)$}})
\label{fneighborhoodcc}
\eearray
\end{corollary}

\begin{proof}[Proof of Corollary~\ref{cor:a11}]
 Let $ \glp,\d_2,c$, and all relevant norms be as in Lemma \ref{slicelemma},
let $c_1=c\sqrt{2}$,
and let
$\d\in (0,\d_2]$.

 We will prove \eqref{f(normneighborhood_x_ball)} last.
 First,
the equality
\eqref{fneighborhoodcc} follows from \eqref{e-neighborhoods}, as does the
equality
\be\label{latterequality}
N_{\e}(VG_\L^0)=\Union_{U\in VG_\L^0} B_{\e}^{SO}(U)
=\Union_{R\in G_\L^0} B_{\e}^{SO}(VR)
\ee
for any $\e>0$.  But \eqref{latterequality} implies that for any $\e,\e'>0,$
 \ben
 N_{\e}(VG_\L^0)\times B^\D_{\e'}(\L)=\big( \Union_{R\in G_\L^0} B_{\e}^{SO}(VR)\ \big) \times B_{\e'}^\D(\L)
=\Union_{R\in G_\L^0} \big( B_{\e}^{SO}(VR)\times B_{\e'}^\D(\L) \big),
\een
yielding the equality
\eqref{f(normneighborhood_x_ball)}.

Next,
$\calc(V,\L)=\{(VR,\L) : R\in G_\L^0\}$ \citep[See][Appendix A]{Groisser2017a},
and for any $R\in G_\L^0$ and
$(U,D)\in B_{c\d/\sqrt{k}}^{SO}(VR)\times B_{\d}^\D(\L),$
\ben
\td((U,D),(V,\L))^2=k\dso(U,VR)^2+\ddpp(D,\L)^2
<2(c\d)^2=(c_1\d)^2.
\een
Thus $B_{c\d/\sqrt{k}}^{SO}(VR)\times B_{\d}^\D(\L)\ \subset
\ B_{c_1\d}^{\td}(VR,\L)$.
Hence the RHS of \eqref{f(unionprodballs)} (and therefore
the equal RHS of \eqref{f(normneighborhood_x_ball)})  is contained in the RHS
of  \eqref{f(unionballs)}.

It remains only to establish the inclusion   \eqref{f(unionprodballs)}.
Let  $S'\in B_\d^\frob(S).$
Then
\ben
\d>\|S'-S\|_\fr=
\|S'-V\L V^T\|_\fr =\|V^{-1}S'V-\L \|_\fr,
\een
 so $S'':=V^{-1}S'V\in B^\frob_{\d}(\L)$.
By Lemma \ref{slicelemma}(e),
there exist $R\in G_\L$, $A\in \glp$ and
\lnbrk
$L\in \fd(p)$ such that
$S''=R e^A \L e^L e^{-A} R^T,$ $\|A\|_\so=\dso(e^A,I)$, and
\ben
(k\|A\|_\so^2+\| L\|_\fr^2)^{1/2}=\|(A,L)\|_\tge \leq \| e^A \L e^L e^{-A} -\L\|_\fr
< c\d.
\een
Then $S'=VS''V^T= VRe^AD(VRe^A)^T=F(VRe^A,D),$ where
and $D=\L e^L$.
Since $\dso(VRe^A,VR)=\dso(e^A, I)=\|A\|_\so<c\d/\sqrt{k}$ and
$\ddpp(D,\L)
=\ddpp(\L e^L, \L) =\|L\|_\fr<c\d,$   the pair $(VRe^A,D)$ lies in  $B_{c\d/\sqrt{k}}^{SO}(VR)\times B_{c\d}^\D(\L).$

 Hence $S'\in \F\big( B_{c\d/\sqrt{k}}^{SO}(VR)\times B_{c\d}^\D(\L) \big)$,
establishing  \eqref{f(unionprodballs)}.

\end{proof}

\subsubsection{Proof of Lemma~\ref{lsclemma-unif}}
Since the eigenvector and eigenvalue matrices $(U,D) \in M(p)$ have distinct roles in the proof of Lemma~\ref{lsclemma-unif}, we restate the lemma with a different notation:

\begin{lemma}[Lemma~\ref{lsclemma-unif}]
Let $K\subset \sympp$ be a compact set.  Let $\e>0$ and let $S\in \sympp$. There exists $\d_1=\d_1(S,K,\e)>0$
such that for  all $S_0\in K$,   all $ \tS_0 \in \Fc^{-1}(S_0)$, all $ (V,\L) \in \Fc^{-1}(S)$, and all
$S'\in \Fc\big(B_{\d_1}^\td(  (V,\L) )\big)$,
\be\tag{\ref{final_est}}
\dsr(S',S_0)^2 > \dsr(S,S_0)^2 -\e
\ee
 and
\be\tag{\ref{final_est_psr}}
\dpsr(S',\tS_0)^2 > \dpsr(S,\tS_0)^2 -\e.
\ee
\end{lemma}

\begin{proof}[Proof of Lemma~\ref{lsclemma-unif}]
For any
$A\subset \sympp$, let $\tilde{A}_\D$ denote the  image of $\Fc^{-1}(A)$ under the natural projection $\tim(p)\to\dpp.$  (Thus $\tilde{A}_\D$ is the set of diagonal matrices occurring in eigendecompositions of
elements of $A$.) We will need this only when $A$ is either $K$ or a one-element set. In the latter case, for $Y\in \sympp$
we  write $\tilde{Y}_\D$ for $\tilde{\{Y\}}_\D$.

For each $h\in \G(p)$ the function $\dpp\times \dpp\to\bfr$ given by $(D_1,D_2)\mapsto \ddpp(D_1, \pi_h\cdot D_2)^2$
is locally uniformly continuous.  Hence, for
each $(\L,h,D_0)\in
\lnbrk
\tS_\D\times \G(p)\times \tilde{K}_\D$ there are numbers
$\tilde{\d}_3(\L,h;D_0)\in (0,\dstrat(S)])$ and
$\tilde{\d}_4(\L,h;D_0)>0$
such that for all $\L'\in B_{\d_3(\L,h;D_0)}^\D(\L)$
and $D_0'\in B_{\d_4(\L,h;D_0)}^\D(D_0),$

\be\label{ddpp2}
\ddpp(\L',\pi_h\cdot D_0')^2 > \ddpp(\L,\pi_h\cdot D_0')^2 - \e/2.
\ee

\ssn
 Choose such numbers $\tilde{\d}_3(\L,h;D_0),$
$\tilde{\d}_4(\L,h;D_0)$ for every $(\L,h,D_0)\in \tS_\D\times \G(p)\times \tilde{K}_\D.$

  Since $\G(p)$ is finite, and the set $\tilde{Y}_\D$ is finite for every $Y\in \sympp$,
  given any $S_0\in K$
  we may choose $\d_3(S_0), \d_4(S_0)>0$ such that \eqref{ddpp2} holds
simultaneously for all $(\L,h,D_0)\in \tS_\D\times\G(p)\times(\Tilde{S_0})_\D$,
$\L'\in B_{\d_3(S_0)}^\D(\L),$ and
$D_0'\in B_{\d_4(S_0)}^\D(D_0).$  (The numbers $\d_3(S_0),\d_4(S_0)$ depend on $S$ and $\e$ as well, but $S$ and $\e$
were fixed in the
hypotheses of the proposition.)  Without loss of generality, we impose the additional restriction
$\d_3(S_0)\leq \dstrat(S).$

By Proposition 3.5 of \cite{Groisser2017a},  for any $(V',\L'), (U_0',D_0')\in \tim(p)$,
\be\label{dga_prop3.5-1}
\dsr(\F(V',\L'),F(U_0',D_0'))^2 =
\min_{h\in \G(p)}\left\{
k\, \hat{d}_h\big( (V',\L'), (U_0', D_0') \big)^2 +\ddpp(\L', \pi_h\cdot D_0')^2 \right\},
\ee
where
\be\label{defdhath}
\hat{d}_h\big((V',\L'), (U_0', D_0') \big)
= \min_{R_1\in G^0_{\L'}, R_2\in G^0_{D_0'}}
\big\{\dso(V'R_1, U_0'R_2h^{-1})\big\} \ \leq\ \diam(\sop).
\ee

 Similarly,
\be
\dpsr(\F(V',\L'), (U_0',D_0'))^2 =
\min_{h\in \G(p)}\left\{
k\, \hat{\hat{d}}_h\big( (V',\L'), (U_0', D_0') \big)^2 +\ddpp(\L', \pi_h\cdot D_0')^2 \right\},
\ee
where
\be\label{defdhathath}
\hat{\hat{d}}_h\big((V',\L'), (U_0', D_0') \big)
= \min_{R_1\in G^0_{\L'}}
\big\{\dso(V'R_1, U_0'h^{-1})\big\} \ \leq\ \diam(\sop).
\ee

Let
\ben
\d_2=\min\left\{\diam(\sop),\frac{\e}{6k\,\diam(\sop)}\right\}.
\een
By definition of $\dstrat(S),$ for all $(V,\L)\in \Fc^{-1}(S)$ and all
$\L'\in B_{\dstrat(S)}^\D(\L)$ we have
$\J_{\L'}\geq \J_\L$, implying $G_{\L'}^0\subset G_\L^0$.
 Hence
for all $S_0\in K$, \ $(U_0,D_0)\in \Fc^{-1}(S_0),$\ $(U_0',D_0')\in B^{SO}_{\d_2}(U_0)\times
B_{\d_4(S_0)}^\D(D_0),$
$(V,\L)\in \Fc^{-1}(S), \ (V',\L')\in B^{SO}_{\d_2}(V)\times B_{\d_3(S_0)}^\D(\L),$
and $h\in \G(p),$
we have
\bearray\nonumber
\min_{R_1\in G^0_{\L'}, R_2\in G^0_{D_0'}}\big\{
\dso(V'R_1, U_0'R_2h^{-1})\big\}
&\geq &
\min_{R_1\in G^0_\L, R_2\in G^0_{D_0'}}\big\{
\dso(V'R_1, U_0'R_2h^{-1})\big\} 
\eearray
 and
\bearray\nonumber
\min_{R_1\in G^0_{\L'}}\big\{
\dso(V'R_1, U_0'h^{-1})\big\}
&\geq &
\min_{R_1\in G^0_\L}\big\{
\dso(V'R_1, U_0'h^{-1})\big\};
\eearray
i.e.
\be\label{dhatdhat}
\hat{d}_h \big( (V',\L'), (U_0', D_0') \big) \geq \hat{d}_h \big( (V',\L), (U_0', D_0') \big) 
\ee
 and
\be\label{dhathatdhathat}
\hat{\hat{d}}_h \big( (V',\L'), (U_0', D_0') \big) \geq \hat{\hat{d}}_h \big( (V',\L), (U_0', D_0') \big). 
\ee

With all data as above, observe that for all $R_1,R_2\in\sop,$
\bestar
\nonumber
\big|
\dso(V'R_1, U_0'R_2h^{-1})-\dso(VR_1, U_0'R_2h^{-1})
\big|
&\leq & \dso(V'R_1,VR_1)
\nonumber
\\
&=&\dso(V',V)
\\
&<&
\d_2.
\eestar
\noi Hence
\bestar
\lefteqn{
\big|
\hat{d}_h\big( (V',\L), (U_0',D_0') \big) -
\hat{d}_h\big( (V,\L), (U_0',D_0') \big)
\big| =
} \\
&& \big|
\min_{R_1\in G^0_\L, R_2\in G^0_{D_0}}\big\{
\dso(V'R_1, U_0'R_2h^{-1})\big\}
-
\min_{R_1\in G^0_\L, R_2\in G^0_{D_0}}\big\{
\dso(VR_1, U_0'R_2h^{-1})\big\}
\big| \\
&& < \d_2\,,
\eestar
implying
$\hat{d}_h\big( (V',\L), (U_0',D_0') \big)>\hat{d}_h\big( (V,\L), (U_0',D_0') \big) -\d_2$.
 Similarly,
$\hat{\hat{d}}_h\big( (V',\L), (U_0',D_0') \big)>\hat{\hat{d}}_h\big( (V,\L), (U_0',D_0') \big) -\d_2$.
Combining  these last two inequalities with \eqref{dhatdhat} and
\eqref{dhathatdhathat}, we find
\be\label{d2bound}
\hat{d}_h\big( (V,\L),  (U_0',D_0') \big) <\hat{d}_h\big( (V',\L'), (U_0', D_0') \big) ) +\d_2 
\ee
and
\be\label{d2bound_psr}
\hat{\hat{d}}_h\big( (V,\L),  (U_0',D_0') \big) <\hat{\hat{d}}_h\big( (V',\L'), (U_0', D_0') \big) ) +\d_2\,.
\ee

\noi Letting $S_0'=F(U_0',D_0'),$ the bounds \eqref{d2bound} and \eqref{ddpp2}
then yield
\bearray
\nonumber
\lefteqn{
k\, \hat{d}_h\big( (V,\L), (U_0', D_0') \big)^2 +\ddpp(\L, \pi_h\cdot D_0')^2} \\
\nonumber
&<& k\, [\hat{d}_h ((V',\L'), (U_0', D_0'))+\d_2]^2 +\ddpp(\L',\pi_h\cdot D_0')^2 + \e/2
\ \ \ \ \  \mbox{(by \eqref{ddpp2})} \\
\nonumber
&=&k\, \hat{d}_h ((V',\L'), (U_0', D_0'))^2 +\ddpp(\L',\pi_h\cdot D_0')^2 \\
\nonumber
&&+k\d_2\left(2\, \hat{d}_h  ((V',\L'), (U_0', D_0'))+\d_2\right) +\e/2\\
\nonumber
&<&k\, \hat{d}_h ((V',\L'), (U_0', D_0'))^2 +\ddpp(\L',\pi_h\cdot D_0')^2 \\
\nonumber
&&+k\d_2\left(3\, \diam(\sop)\right) +\e/2\\
&<& k\, \hat{d}_h ((V',\L'), (U_0', D_0'))^2 +\ddpp(\L',\pi_h\cdot D_0')^2+\e
\label{almostthere}
\eearray
(by our definition of $\d_2$). Since \eqref{almostthere} holds for every $h\in \G(p)$, it follows from
\eqref{dga_prop3.5-1}
that $\dsr(S,S_0')^2<\dsr(S', S_0')^2+\e,$ where $S'=\Fc(V',\L')$.
 Additionally writing $\tS_0'=(U_0',D_0')$, the bounds \eqref{d2bound_psr}
and \eqref{ddpp2} similarly imply that
$\dpsr(S,\tS_0')^2<\dpsr(S', \tS_0')^2+\e.$

Thus
\eqref{final_est}  and \eqref{final_est_psr} hold for all
$S'\in F\big(B^{SO}_{\d_2}(V)\times B_{\d_3(S_0)}^\D(\L)\big),$
$S_0'\in F\big(B^{SO}_{\d_2}(U_0)\times B_{\d_4(D_0)}^\D(D_0)\big),$
 and
$\tS_0'\in B^{SO}_{\d_2}(U_0)\times B_{\d_4(D_0)}^\D(D_0).$

Since $\Fc$ is a proper map (Proposition \ref{fproper}),
$\Fc^{-1}(K)$ is compact.
The collection
$\left\{B^{SO}_{\d_2}(U_0)\times B_{\d_4(D_0)}^\D(D_0)\right\}_{(U_0,D_0)\in \Fc^{-1}(K)}$
is an open cover of $\Fc^{-1}(K),$ and hence has a finite subcover
$\left\{(B^{SO}_{\d_2}(U_0^{(i)})\times B_{\d_4(D_0^{(i)})}^\D(D_0^{(i)})\right\}_{i=1}^n$
with ``centers" $\tS_i=(U_0^{(i)}, D_0^{(i)}), 1\leq i\leq n.$

Define $\d_5=\min\{\d_3(\tS^{(i)}): 1\leq i\leq n\}.$ Then \eqref{final_est}  and \eqref{final_est_psr} hold whenever
$(V,\L)\in \Fc^{-1}(S),$ 
$S'\in  \Fc \big(B^{SO}_{\d_2}(V)\times B_{\d_5}^\D(\L)\big),$
$S_0\in K,$
 and $\tS_0\in \Fc^{-1}(S_0).$

Finally, let $\d_1=\min\{\sqrt{k}\,\d_2,\d_5\}.$
 Then $B_{\d_1}^\td(V,\L) \subset B^{SO}_{\d_2}(V)\times B_{\d_5}^\D(\L),$
so  \eqref{final_est}  and \eqref{final_est_psr}
hold for all $S'\in \Fc\big(B_{\d_1}^\td(V,\L)\big),$
$S_0\in K,$  and $\tS_0\in \Fc^{-1}(S_0).$
\end{proof}

\subsubsection{Proof of Theorem~\ref{prop:lsc}}

\begin{proof}[Proof of Theorem~\ref{prop:lsc}]
As mentioned after the statement of the theorem, part (a) is a special case of part (b), so it suffices to prove part (b).

Observe that if $\tilde{K}\subset \tim(p)$ is compact, then so are $\Fc(\tilde{K})$ and
(since $\Fc$ is proper [Proposition \ref{fproper}]) also $\Fc^{-1}(\Fc(\tilde{K}))$. Since $\tilde{K}\subset
\Fc^{-1}(\Fc(\tilde{K}))$, any property that is uniform over $\Fc^{-1}(\Fc(\tilde{K}))$ is uniform over $K$.
Hence, to prove the desired result for $\dpsr$, it suffices to consider compact subsets of $\tim(p)$ of
the form $\Fc^{-1}(K)$, where $K\subset\sympp$ is compact.

Let $S\in \sympp,$ let $K\subset \sympp$  be a compact set, and let $\e>0.$
Let $\d_1=\d_1(S,K,\e)$ be as in Lemma \ref{lsclemma-unif}.
Let $(V,\L)\in \Fc^{-1}(S)$, let $c_1=c_1(\L)$ and $\d_2=\d_2(\L)$ be as in Corollary \ref{slicecor},
and let $\d=\min\{\d_1/c_1,\d_2\}$.

Let $S'\in B_\d^\frob(S),$ and
let $S_0\in K,$  and let $\tS_0\in \Fc^{-1}(K)$.
   Since $\d\leq \d_2$, relations
\eqref{f(unionprodballs)}--\eqref{f(unionballs)} in Corollary \ref{slicecor},
ensure that $S'\in \Fc\left( B_{c_1\d}^\td(\tS) \right)$
for some $\tS\in  \Fc^{-1}(S).$
Since $c_1\d\leq \d_1$\,,
Lemma \ref{lsclemma-unif} implies that
$\dsr(S',S_0)^2 > \dsr(S,S_0)^2 -\e$
 and $\dpsr(S',\tS_0)^2 > \dsr(S,\tS_0)^2 -\e.$

This proves that $\dsr^2$  and $\dpsr^2$ are LSC in their first variables, locally uniformly with respect to their second variables.   The analogous
result for the unsquared functions $\dsr$  and $\dpsr$ then follow from Corollary \ref{sqrtlsc_unif}.
\end{proof}

\subsubsection{Proof of Lemma~\ref{lem:finite_variance}}

We first develop an inequality for $\dsr$ which plays the role of the triangle inequality.
\begin{lemma} \label{lem:sr_triangle}
 Let $S_2 \in \Symp$. Then there exists a constant $C = C(S_2) \in (0, \infty)$, depending only on the eigenvalues of $S_2$, such that for all $S_0,S_1\in \Symp$,
   $$\dsr(S_1,S_2) \le \dsr(S_1,S_0) + \dsr(S_0,S_2) + C(S_2).$$
\end{lemma}

\begin{proof}
Note from Proposition 3.5 of \cite{Groisser2017a} that given any two $S', S'' \in \Symp$ and a connected component $\mathcal{C}'$ of $\mathcal{F}^{-1}(S')$, there is a minimal pair
$\big( (U',D'), (U'',D'') \big) \in \mathcal{C}' \times \Fc^{-1}(S'')$. Note also that if both $(U',D')$ and $(U,D)$ are in the same connected component $\mathcal{C}'$ of $S'$, then $D = D'$. Moreover, for any $(U',D'), (U,D) \in \mathcal{F}^{-1}(S')$,
\begin{align}
  d_M((U',D'), (U,D)) & \le  d_M((U',D'), (U,D')) + d_M((U,D'), (U,D)) \nonumber \\
                      & \le \sqrt{k}\mbox{diam}(SO(p)) +d_{\mathcal{D}^+}(D',D). \label{eq:1lem:sr_triangle}
\end{align}

Let $\mathcal{C}_0$ be a connected component of $\mathcal{F}^{-1}(S_0)$. For $i= 1,2$, let $\big( (U_i,D_i) ,(U_0^{(i)},D_0^{(i)}) \big) \in \mathcal{F}^{-1}(S_i) \times \mathcal{F}^{-1}(S_0) $ be minimal pairs with $(U_0^{(i)},D_0^{(i)})$ both lying in $\mathcal{C}_0$. Let $\mathcal{C}_1$ be the connected component of $\mathcal{F}^{-1}(S_1)$ containing $(U_1,D_1)$, and let $ \big( (U'_1,D'_1), (U_2',D_2') \big) \in \mathcal{F}^{-1}(S_1) \times \mathcal{F}^{-1}(S_2)$ be a minimal pair for $(S_1,S_2)$ with $(U'_1,D'_1) \in  \mathcal{C}_1$. Then,
$D_0^{(1)} = D_0^{(2)} =: D_0$, and $D_1' = D_1$.
Moreover,
$$\dsr(S_i,S_0) =d_M(U_i,D_i),(U_0^{(i)},D_0^{(i)})) = d_M(U_i,D_i),(U_0^{(i)},D_0 ) ),$$
and
$$\dsr(S_1,S_2) =d_M(U_1',D_1'),(U_2',D_2')) = d_M(U_1',D_1),(U_2',D_2' ) ).$$
Hence
\begin{align*}
 \dsr& (S_1,S_2) = d_M(U_1',D_1),(U_2',D_2' ) ) \\
    &\le d_M(U_1',D_1),(U_1,D_1 ) ) + d_M((U_1,D_1 ), (U_0^{(1)},D_0 ) ) + d_M( (U_0^{(1)},D_0 ) , (U_0^{(2)},D_0 ) ) \\
    &  \quad + d_M((U_0^{(2)},D_0 ), (U_2, D_2)) + d_M((U_2, D_2) , (U_2',D_2' ) )\\
    &\le \sqrt{k} \mbox{diam}(SO(p)) + \dsr(S_1,S_0) + \sqrt{k} \mbox{diam}(SO(p)) \\
    & \quad + \dsr(S_0,S_2) + (\sqrt{k} \mbox{diam}(SO(p)) + d_{\mathcal{D}^+}(D_2,D_2'))\\
    & \le \dsr(S_1,S_0) +  \dsr(S_0,S_2) + 3\sqrt{k} \mbox{diam}(SO(p)) +   d_{\mathcal{D}^+}(D_2,D_2'). \qedhere
\end{align*}
\end{proof}

\begin{proof}[Proof  of Lemma~\ref{lem:finite_variance}]
(a) Suppose that for some $(U_0,D_0) \in M(p)$, $\fpsr(U_0,D_0)<\infty$.  Then for any given $(U,D) \in M(p)$,
$$\fpsr (U,D)  = \int_{\Symp} \inf_{(U_X,D_X) \in \mathcal{F}^{-1}(X)} [ kd^2_{SO}(U_X,U) + d^2_{\mathcal{D}^+}(D_X,D)] dP.$$
By the triangle inequality, we have
\begin{align*}
  d^2_{SO}(U_X,U) & \le \{d_{SO}(U_X, U_0) + d_{SO}(U_0,U)\}^2 \\
    & \le 2 d^2_{SO}(U_X, U_0) + 2d^2_{SO}(U_0,U),
\end{align*}
and similarly $d^2_{\mathcal{D}^+}(D_X,D) \le 2d^2_{\mathcal{D}^+}(D_X,D_0) + 2d^2_{\mathcal{D}^+}(D_0,D)$.
Thus,
$$
\fpsr  (U,D)   \le 2 \int_{\Symp}  d^2_{\mathcal{PSR}}(X,(U_0,D_0)) dP + C < \infty,$$
where $C = 2k d^2_{SO}(U_0,U) + 2d^2_{\mathcal{D}^+}(D_0,D)$. Moreover, for any given $\Sigma \in \Symp$, choosing any $(U',D') \in \mathcal{F}(\Sigma)$, we have  $\fsr (\Sigma) \le \fpsr(U',D') < \infty$  by (\ref{SR_PSR_ineq}).

(b) Suppose that for some $S_0 \in \Symp$, $\fsr(S_0)<\infty$. For any given $S \in \Symp$, Lemma~\ref{lem:sr_triangle} gives
$\dsr(X,S) \le \dsr(X,S_0) + C'(S_0,S)$ for any $X \in \Symp$,  where $C'(S_0,S) = \dsr(S_0,S) + C(S) < \infty$. Thus,
\begin{align*}
 \fsr(S)& = \int_{\Symp} \dsr^2(X,S) dP \le \int_{\Symp} 2\dsr^2(X,S_0) dP + 2C'(S_0,S)^2 \\
        & = 2\fsr(S_0) + 2C'(S_0,S)^2< \infty. \qedhere
 \end{align*}
\end{proof}

\subsubsection{Proof of Lemma~\ref{lem:contfsr}}

\begin{proof}[Proof of Lemma~\ref{lem:contfsr}]
(i) By Theorem \ref{prop:lsc}, $\dsr^2\mid_{\Symp \times K}$ is LSC in the first variable, uniformly with respect to the second. Let $S_0 \in \Symp$ and $\epsilon >0$ be arbitrary, and let $U$ be the open set containing $S_0$ as in Definition~\ref{def:lsc}(i). Then for all $S \in U$ and $S' \in K$, $\dsr^2(S,S') > \dsr^2(S_0,S') - \epsilon.$ Hence
$$\fsr(S)  = \int_{K} \dsr^2(S, \cdot) dP > \int_{K} \left( \dsr^2(S_0,\cdot) - \epsilon \right) dP
          =  \fsr(S_0) - \epsilon.$$
Hence $\fsr$ is LSC at the arbitrary point $S_0$.

(ii) We will show that $(\fpsr)^{1/2}$ is Lipschitz with constant 1:
\be\label{diff_fpsrp}
\big|\fpsr(m_1)^{1/2}-\fpsr(m_2)^{1/2}\big|
\leq \td(m_1,m_2) \ \ \ \mbox{\rm for all $m_1,m_2\in \tim(p).$}
\ee
This will imply uniform continuity of $(\fpsr)^{1/2}$, and therefore continuity of $\fpsr$.
Let $m_1,m_2 \in M(p)$. Utilizing Lemma~\ref{lem:dpsr_continuity}(ii),
\begin{align}
  \big|\fpsr&(m_1)  -\fpsr(m_2)\big| \label{diff_fpsrp-1} \\
  \leq&\int_{\sympp}\big|\dpsr(X,m_1)- \dpsr(X,m_2)\big| \
\big[\dpsr(X,m_1)+ \dpsr(X,m_2)\big]\, P(dX) \nonumber \\
  \leq& d_M(m_1,m_2) \int_{\sympp} \
\big[\dpsr(X,m_1)+ \dpsr(X,m_2)\big]\, P(dX)\nonumber\\
\leq& \td(m_1,m_2) \left[ \fpsr(m_1)^{1/2} +\fpsr(m_2)^{1/2} \right].\nonumber
\end{align}
If $\dpsr(m_1)=\dpsr(m_2)=0$, then \eqref{diff_fpsrp} is true trivially. Otherwise, dividing both sides of
\eqref{diff_fpsrp-1} by $\left[ \fpsr(m_1)^{1/2} +\fpsr(m_2)^{1/2} \right]$ yields \eqref{diff_fpsrp}.
\end{proof}

\subsection{Proofs for Section 4.2}
\subsubsection{Proof of Proposition~\ref{coercivity}}

\begin{proof}[Proof of Proposition~\ref{coercivity}]
For $r>0$ let $$K_r=\{S\in\sympp : \mbox{every eigenvalue $\l$ of $S$
satisfies $|\log \l|\leq r$} \},$$ and let $\kappa_r =\diam(\Fc^{-1}(K_r)).$
For each $r$, the set $K_r$ is compact, and hence so is $\tilde{K}_r:=F^{-1}(K_r)$ (by Proposition \ref{fproper}).
Note also that $I\in K_r$ and $(I,I)\in \tilde{K}_r$ for every $r>0$.

Suppose $r_2>r_1>0$ and that $S_2\in\sympp\minus K_{r_2}$.
Then for any $m_1=(U_1,D_1)\in \tilde{K}_{r_1}$ and $m_2=(U_2,D_2)\in \Fc^{-1}(S_2)$,
the matrix $D_2$ has some eigenvalue $\l_2$ with $|\log\l_2|>r_2$, while
every eigenvalue $\l_1$ of $D_1$ satisfies $|\log\l_1|\leq r_1$, implying
$\td(m_1, m_2)\geq \ddpp(D_1,D_2)\geq r_2-r_1.$  Hence $\dsr(S_1,S_2)\geq r_2-r_1\,,$
and thus
\be\label{coercive_sr}
\fsr(S_2) \geq  \int_{K_{r_1}} \dsr(S_1, S_2)^2 P(dS_1)
\geq
(r_2-r_1)^2 P(K_{r_1}).
\ee

Now choose $r_1$  large enough that $P(K_{r_1})>0$; such
$r_1$ exists since $\Union_{r>0} K_r=\sympp$.  Let $r_2>r_1$ be large enough that
$(r_2-r_1)^2 P(K_{r_1})>\fsr(I).$ Then by \eqref{coercive_sr}, for every
$S\in \sympp\minus K_{r_2}$,
\be\label{faraway-2_sr}
\fsr(S)\geq (r_2-r_1)^2 P(K_{r_1}) >\fsr(I).
\ee
Since $I\in K_{r_2}$, \eqref{faraway-2_sr} implies that \eqref{coerce_sr} holds with $K=K_{r_2}$,
establishing part (a).

For (b), let $r_1$ be as above, consider a (new) arbitrary $r_2>r_1$, and  let
$m_1, m_2$ be as above. Then essentially the same argument as above shows that
 $\dpsr(S_1,m_2)
 \geq r_2-r_1$ and hence that
$\fpsrp(m_2)
\geq
(r_2-r_1)^2 P(K_{r_1}).$
%
Now let $r_2>r_1$  be large enough that
$(r_2-r_1)^2 P(K_{r_1})>\fpsrp(I,I).$ Then
for every $m\in \tim(p)\minus
\tilde{K}_{r_2}$ we have
$\fpsrp(m)
>\fpsrp(I,I).$
Since $(I,I)\in \tilde{K}_{r_2}$,
this implies that \eqref{coerce_psr} holds with $K=K_{r_2}$.
\end{proof}

\subsubsection{Proof of Theorem~\ref{thm:existence_popSR_PSRmeans}}

\begin{proof}[Proof of Theorem~\ref{thm:existence_popSR_PSRmeans}]
(a) By Proposition \ref{coercivity}(a), there exists a compact set $K \subset \Symp$ such that equation \eqref{coerce_sr} holds. But by Lemma~\ref{lem:contfpsr}, $\fsr$ is LSC, hence
achieves a minimum value on  $K$, say at $S_0$. By \eqref{coerce_sr},
$\fsr(S_0)$ is the minimum value
of $\fsr$ on all of $\Symp$.  Hence $\Esr$  is nonempty.

The proof for (b) is almost identical to the proof for (a), except that the finite PSR-variance condition actually ensures \emph{continuity} (not just semi-continuity) of $\fpsr$.
\end{proof}

\subsection{Proofs for Section 4.3}

\subsubsection{Proof of Lemma~\ref{eig_decomp_dist_bound}}

\begin{proof}[Proof of Lemma~\ref{eig_decomp_dist_bound}]
(a) Since $h \cdot (U,D) = (Uh^{-1},h \cdot D )$, and $h\cdot D = h  D h^{-1}$,
\begin{align*}
d_{M}((U,D),(Uh^{-1},hDh^{-1}))
&=\left\{d^2_{\mathcal{D}^+}(D,hDh^{-1}) + kd^2_{SO}(U,Uh^{-1})\right\}^{1/2} \nonumber\\
&\ge \sqrt{k}d_{SO}(U,Uh^{-1}) 
=\sqrt{k}d_{SO}(I_p,h^{-1}) \nonumber\\
&\ge \sqrt{k}\beta_{\mathcal{G}(p)}. \nonumber
\end{align*}

(b) The set $\mathcal{G}(p)$ contains the block-diagonal signed permutation matrix
\[ B =
\begin{pmatrix}
 I_{p-2} & 0 \\
0 & R(\frac{\pi}{2})
\end{pmatrix}
\]
where
\[ R(\tfrac{\pi}{2}) =
\begin{pmatrix}
\cos(\frac{\pi}{2}) & -\sin(\frac{\pi}{2}) \\
\sin(\frac{\pi}{2}) & \cos(\frac{\pi}{2})
\end{pmatrix}
=
\begin{pmatrix}
0 & -1 \\
1 & 0
\end{pmatrix}
.\]
It can be shown that
\[ {\rm Log}(B) =
\begin{pmatrix}
0 \hspace{1em} & 0 \\
0 \hspace{1em} & {\rm Log}(R(\frac{\pi}{2}))
\end{pmatrix}
,\]
where
\[ {\rm Log}\left( R(\tfrac{\pi}{2}) \right) =
\begin{pmatrix}
0 & -\frac{\pi}{2} \\
\frac{\pi}{2} & 0
\end{pmatrix}
.\]
Then we have that $d_{SO}(I_p,B)=\frac{\pi}{2}$, which implies that $\beta_{\mathcal{G}(p)} \le \frac{\pi}{2}$.

(c) Given $X \in S_p^{\rm top}$, let $(U_X,D_X)$ and $(U'_X,D'_X)$ be two distinct eigen-decompositions of $X$. From Theorem 3.3 of \cite{Jung2015}, there is an even signed-permutation $h$ such that $(U'_X,D'_X) = (Uh^{-1},h \cdot D)$. Since $(U_X,D_X)$ and $(U'_X,D'_X)$ are distinct, $h \neq I_p$. Applying Part (a) gives the result.
\end{proof}

\subsubsection{Proof of Lemma~\ref{lem:geom}}

\begin{proof}[Proof of Lemma~\ref{lem:geom}]
It is well known that $({\rm Diag}^+(p), g_{\mathcal{D}^+})$ has non-positive sectional curvature and infinite injectivity radius, and $(SO(p), k g_{SO})$ has non-negative sectional curvature (bounded above by $\Delta(SO(p), k g_{SO}) = 1/(4k)$) and injectivity radius $r_{\rm inj}(SO(p), k g_{SO}) = \sqrt{k}\pi$ (see Section 5 of \cite{Manton2004}).
Thus the injectivity radius of $(M, g_M)$ is $r_{\rm inj}(M, g_M) = r_{\rm inj}(SO(p), k g_{SO})$, and the sectional curvature of $(M, g_M)$ is bounded by $\Delta(M, g_M) = \Delta(SO(p), k g_{SO})$.

We apply Theorem 2.1 of \cite{Afsari2011}, which shows that the minimizer $\bar{m}(P)$ of $\int_{M(p)} d_M^2(\tilde{X}, m) P(d \tilde{X})$ uniquely exists and lies in $B_{r}^{d_M}(m_0)$, provided that
\begin{equation}
r \le \min\{ r_{\rm inj}(M, g_M), \pi/\sqrt{\Delta(M, g_M)} \}/2. \label{eq:lemma-bound0}
\end{equation}
Since $$\min\{ r_{\rm inj}(M, g_M), \pi/\sqrt{\Delta(M, g_M)} \}/2 = \min\{ \sqrt{k} \frac{\pi}{2} , 2\sqrt{k}\frac{\pi}{2} \} = \frac{\sqrt{k}\pi}{2}$$ and
$r \le \sqrt{k} \beta_{\mathcal{G}(p)} \le \frac{\sqrt{k}\pi}{2}$ (by Lemma~\ref{eig_decomp_dist_bound}(b)),  we have the desired bound (\ref{eq:lemma-bound0}).
\end{proof}

\subsubsection{Proof of Theorem \ref{thm:uniqueness} and Corollary \ref{uniqueness_theorem}}

\begin{proof}[Proof of Theorem \ref{thm:uniqueness}]
By the condition (\ref{eq:condition_unique}), for any $S_1,S_2 \in {\rm supp}(P)$,
\begin{equation}\label{eq:support_bound}
  \dsr(S_1, S_2) < r'_{cx}.
\end{equation}
Since the complement of $S_p^{\rm top}$ has volume zero in $\Symp$, we have for any $m \in M(p)$,
\begin{equation}\label{eq:support_bound2}
  \fpsr(m) = \int_{\Symp} \dpsr^2(X, m) P(dX) = \int_{\Sptop} \dpsr^2(X, m) P(dX).
\end{equation}

Fix an arbitrary $S_0 \in {\rm supp}(P) \cap \Sptop$. There are exactly $v_p := 2^{p-1} p !$ distinct eigen-decompositions in $\Fc^{-1}(S_0)$, and we label them by $\ell = 1,\ldots, v_p$, that is, $\Fc^{-1}(S_0) = \{m_\ell(S_0) : \ell = 1,\ldots ,v_p\}$. We claim the following:

\begin{claim}\label{lem:unique_matching}
For each $\ell$, and for each $S' \in {\rm supp}(P) \cap \Sptop$, one can uniquely choose $m_{\ell}(S') \in \Fc^{-1}(S')$ that forms a minimal pair with $m_\ell(S_0)$. Thus, one can uniquely label all eigen-decompositions  $\{m_{\ell}(S') : \ell =1, \dotsc, v_p \} = \Fc^{-1}(S')$ for all $S'  \in {\rm supp}(P) \cap \Sptop$.
\end{claim}

\begin{proof}[Proof of Claim \ref{lem:unique_matching}]
The fact that both $S_0,S' \in {\rm supp}(P) \cap \Sptop$ and (\ref{eq:support_bound}) ensure that there exists an $m' \in \mathcal{F}^{-1}(S')$ such that
$d_M(m',m_{\ell}(S_0)) = d_{\mathcal{SR}}(S',S_0) < r'_{cx}$.
Choose such an $m'$ and label it to be $m_{\ell}(S') \in \mathcal{F}^{-1}(S')$.
Let $m_k(S') \in \mathcal{F}^{-1}(S')$ be such that $m_{k}(S') \neq m_{\ell}(S')$.
Then by the triangle inequality,
\begin{align}
 d_M(m_{k}(S') , m_\ell(S_0)) & \ge d_M(m_{k}(S'), m_{\ell}(S')) - d_M(m_{\ell}(S') , m_\ell(S_0)) \nonumber \\
 & > \sqrt{k}\beta_{\mathcal{G}(p)} - r'_{cx}
  = \frac{3}{4}\sqrt{k}\beta_{\mathcal{G}(p)}
  > r'_{cx}. \label{triangle_ineq_arg1}
\end{align}
Therefore, $m_{\ell}(S') \in \mathcal{F}^{-1}(S')$ is indeed the unique eigen-decomposition that forms a minimal pair with $m_{\ell}(S_0)$.
\end{proof}

By Claim~\ref{lem:unique_matching}, one can therefore label all eigen-decompositions $m_{\ell}(S)$ of all $S \in {\rm supp}(P) \cap \Sptop$, provided that an initial labeling of $S_0$ is given.
For
$\ell = 1,\ldots,v_p$ and for $r > 0$, define a set $H_\ell(r)$ by
\begin{align}
  H_\ell(r) & = \{ m \in M(p) : d_M(m ,m_{\ell}(S)) < r \ {\mbox{for all}} \ S \in {\rm supp}(P) \cap \Sptop \} \nonumber \\
       & = \bigcap_{S \in {\rm supp}(P) \cap \Sptop} B_{r}^{d_M} (m_{\ell}(S)). \label{eq:Hell}
\end{align}


\begin{claim}\label{lem:unique_matching2}

(a) If $r \le 2r'_{cx}$, then for $\ell \neq \ell'$, $H_\ell(r) \cap H_{\ell'}(r) = \emptyset$.

(b) If $r \ge r'_{cx}$, then for any  $S \in {\rm supp}(P) \cap \Sptop$, $m_\ell(S) \in H_\ell(r)$, for any $\ell = 1,\ldots, v_p$.

(c) If $r \le 2r'_{cx}$, then for any $S \in {\rm supp}(P) \cap \Sptop$ and for any $m \in H_\ell(r)$, the eigen-decomposition of $S$ closest to $m$ is $m_\ell(S)$.

(d) For any $S_1, S_2 \in {\rm supp}(P) \cap \Sptop$, $(m_{\ell}(S_1), m_{\ell'}(S_2))$ is a minimal pair if and only if $\ell = \ell'$.

\end{claim}

\begin{proof}[Proof of Claim \ref{lem:unique_matching2}]
Item (b) is immediate by the definition of $H_\ell$ (\ref{eq:Hell}).

Let $S \in {\rm supp}(P) \cap \Sptop$, $m \in H_\ell(r)$ and note that for any $\ell' \neq \ell$,
\begin{align}
 d_M(m,m_{\ell'}(S))  & \ge d_M(m_{\ell'}(S),m_{\ell}(S)) - d_M(m,m_{\ell}(S)) \nonumber  \\
   & > 4r_{cx}' - d_M(m,m_{\ell}(S)) > d_M(m,m_{\ell}(S)), \label{eq:uniquness_2}
\end{align}
in which we used the triangle inequality and Lemma~\ref{eig_decomp_dist_bound} (c),  and the fact that $d_M(m,m_{\ell}(S)) < 2r_{cx}'$ (given by the condition $r \le 2r'_{cx}$ and the definition of $H_\ell(r)$). This shows (c).

Take $r \in [ r'_{cx}, 2r'_{cx}]$, then parts (b) and (c) are true. To verify (d), from part(c), replace $m$ by $m_{\ell}(S_1)$, $m_{\ell'}(S')$ by $m_{\ell'}(S_2)$ for $\ell \neq \ell'$.

To verify (a) it is sufficient to assume $r = 2r'_{cx}$. Let $m \in H_\ell(r)$, then there exists an $S' \in  {\rm supp}(P) \cap \Sptop$ such that
 $d_M(m,m_{\ell}(S'))  < 2r_{cx}'$. But
  $$d_M(m,m_{\ell'}(S')) \ge d_M(m_{\ell'}(S'),m_{\ell}(S')) - d_M(m,m_{\ell}(S')) > 4r_{cx}' - 2r_{cx}' = 2r_{cx}',$$ thus yielding $m \notin H_{\ell'}(r)$.
\end{proof}

For each $\ell = 1,\ldots, v_p$ write $H_\ell^{\rm top}(r'_{cx})$ for  $H_\ell(r'_{cx}) \cap M_p^{\rm top}$ for notational simplicity.
Fix an $\ell = 1,\ldots, v_p$, and consider the eigen-composition map $\Fc$ restricted to $H_\ell^{\rm top}(r'_{cx})$, $\Fc | _{H_\ell^{\rm top}(r'_{cx})}: H_\ell^{\rm top}(r'_{cx}) \to \Symp$. Since for any $S \in  {\rm supp}(P) \cap \Sptop$, $\Fc^{-1}(S)$ intersects with $H_\ell^{\rm top}(r'_{cx})$ at a unique point, there exists the push-forward measure $P \circ \Fc | _{H_\ell^{\rm top}(r'_{cx})}$ supported on ${H_\ell^{\rm top}(r'_{cx})} \subset M(p)$. For each $\ell$, denote $P_\ell$ for this ``push-forwarded'' probability measure on $M(p)$. Since the support of $P_\ell$ lies in $H_\ell^{\rm top}(r'_{cx}) \subset B_{r'_{cx}}(m_\ell(S))$ for any $S \in {\rm supp}(P) \cap \Sptop$, by Lemma \ref{lem:geom}, there exits a unique Fr\'{e}chet mean $\bar{m}_\ell := \bar{m}(P_\ell) \in M(p)$ of $P_\ell$, and $\bar{m}_\ell \in B_{r'_{cx}}(m_\ell(S))$. Since the above holds for any $S \in {\rm supp}(P) \cap \Sptop$, we have that
\begin{equation}\label{eq:meaninHell}
  \bar{m}_\ell  \in H_\ell(r'_{cx}).
\end{equation}

We now show that the set $\{\bar{m}_\ell : \ell = 1,\ldots, v_p\}$ is exactly the PSR mean set, or, equivalently that $\bar{m}_\ell$'s are the only minimizers of (\ref{eq:support_bound2}).
Let $m \in M(p)$ be arbitrary. Choose any $S \in  {\rm supp}(P) \cap \Sptop$.

 If $m \in H_\ell(2r'_{cx})$ for some $\ell$, but $m \neq \bar{m}_\ell$ then
\begin{align*}
  \fpsr(m) & = \int_{{\rm supp}(P) \cap \Sptop} \min_{l = 1,\ldots, v_p} d_M^2(m_l(X), m) P(dX) \\
         & = \int_{{\rm supp}(P) \cap \Sptop} d_M^2(m_\ell(X), m) P(dX) \quad (\mbox{by Lemma~\ref{lem:unique_matching2}(c)})\\
         & > \int_{{\rm supp}(P) \cap \Sptop} d_M^2(m_\ell(X), \bar{m}_\ell) P(dX)\\
         & =  \fpsr( \bar{m}_\ell),
\end{align*}
in which the strict inequality is given by the fact that $\bar{m}_\ell$ is the unique Fr\'{e}chet mean of $P_\ell$.

Next, suppose that $m \notin \bigcup_{l = 1}^{v_p} H_l(2r'_{cx})$.
For any $\ell = 1,\ldots,v_p$, we have $m \notin H_\ell(2r'_{cx})$, and there exits an $S_\ell'  \in {\rm supp}(P) \cap \Sptop$ such that
$$d_M(m,m_\ell(S_\ell' )) > 2r'_{cx}.$$
Thus, for any $S \in {\rm supp}(P) \cap \Sptop$ and for any $\ell = 1,\ldots,v_p$,
$d_M(m_\ell(S_\ell' ), m_\ell(S)) = \dsr(S_\ell' ,S)$ by Lemma~\ref{lem:unique_matching2}(d), and
\begin{align*}
  d_M(m,m_\ell(S)) &\ge d_M(m,m_\ell(S')) - d_M(m_\ell(S'), m_\ell(S)) \\
   & > 2r'_{cx} - r'_{cx} = r'_{cx}.
\end{align*}
(We used the fact that for $S_\ell' ,S \in {\rm supp}(P)$, $\dsr(S_\ell' ,S) < r_{cx}'$ (\ref{eq:support_bound}.) Therefore,
\begin{align*}
  \fpsr(m) & = \int_{{\rm supp}(P) \cap \Sptop} \min_{l = 1,\ldots, v_p} d_M^2(m_l(X), m) P(dX) \\
         & >  \int_{{\rm supp}(P) \cap \Sptop} \min_{l = 1,\ldots, v_p} (r'_{cx})^2 P(dX) \\
         & =(r'_{cx})^2.
\end{align*}

However, for any $S \in {\rm supp}(P) \cap \Sptop$
\begin{align*}
\fpsr(\bar{m}_{\ell}) & = \int_{{\rm supp}(P) \cap \Sptop} d_M^2(m_\ell(X), \bar{m}_{\ell}) P(dX) \\
               & \le  \int_{{\rm supp}(P) \cap \Sptop} d_M^2(m_\ell(X), m_\ell(S)) P(dX) \\
               & <   \int_{{\rm supp}(P) \cap \Sptop} (r'_{cx})^2 P(dX) \\
         & =(r'_{cx})^2.
\end{align*}

The above two results show that the set $E_n := \{\bar{m}_{\ell}: \ell = 1,\ldots, v_p\}$ is exactly the partial scaling-rotation mean set for the sample $X_1, \dotsc, X_n$. Since there are exactly $v_p = 2^{p-1}p!$ elements in $E_n$, it follows from (\ref{sample_PSR_mean_nonunique}) that the elements of $E_n$ must belong to the same orbit under the action of $\mathcal{G}(p)$. Part (a) is now proved.

\end{proof}

\begin{proof}[Proof of Corollary \ref{uniqueness_theorem}]
For part (a), the sample $X_1,\ldots,X_n \in \Sptop$ satisfies the support condition~(\ref{eq:support_bound}). A proof of part (a) is given by following the proof of Theorem~\ref{thm:uniqueness}, with the probability measure $P$ replaced by the empirical measure given by the sample $X_1,\ldots,X_n$.

To prove (b), for any given $\ell = 1,\ldots, v_p$,
set the initial guess $\hat{m}^{(0)}$ to be the eigen-decomposition $m_{\ell}(X_1)$ of $X_1$. Then $m_{\ell}(X_i)$ forms a minimal pair with $\hat{m}^{(0)}$ for $i = 1,\ldots,n$, and, as seen earlier, is the unique element of $\mathcal{F}^{-1}(X_i)$ with this property. Thus, for each $i = 1,\ldots,n$, $m_{\ell}(X_i)$ is the unique choice of $m_i^{(0)}$ in  Step 1 of the algorithm.  Since $\bar{m}_{\ell}$ is the unique Fr\'{e}chet mean of $\{m_{1,\ell},\ldots,m_{n,\ell} \}$ by Lemma~\ref{lem:geom} ,  Step 2 of the algorithm yields
$\hat{m}^{(1)}=\bar{m}_{\ell}.$  Thus $\hat{m}^{(1)}$ is exactly the sample PSR mean $\bar{m}_{\ell}$, and
the sample PSR mean set is the orbit $\mathcal{G}(p) \cdot \hat{m}^{(1)}$.
Since $\hat{m}^{(1)} \in H_\ell(r'_{cx})$ by (\ref{eq:meaninHell}), the unique choice of $m_i^{(1)}$ in Step 2 of the procedure is $m_{\ell}(X_i)$, the same as the previous iteration. Thus, $\hat{m}^{(2)} = \hat{m}^{(1)}$ and the algorithm terminates.
\end{proof}

\subsubsection{Proof for the statements in Remark~\ref{remark:4.9}}\label{remark:4.9proof}

Proof of ``(i) yields \eqref{eq:condition_unique2}": Since $\dsr$ is a metric when restricted to $\Sptop$, and $S_0, X_i \in \Sptop$, we have $\dsr(X_i,X_j) \le \dsr(X_i, S_0) + \dsr(X_j,S_0) < r'_{cx}$ for all $i,j = 1,\ldots, n$.

Our proof of ``(ii) yields \eqref{eq:condition_unique2}" consists of two parts.

Part (1): Suppose that $m \in M^{\rm top}(p)$. Then by (\ref{SR_PSR_ineq}),
$$\dsr(X_i, \mathcal{F}(m)) \le \dpsr(X_i, m) < r'_{cx}/2.$$
(In fact, $\dsr(X_i, \mathcal{F}(m)) = \dpsr(X_i, m)$ in this case.) Since $\mathcal{F}(m)\in \Sptop$, the statement for (i) above yields \eqref{eq:condition_unique2}.

Part (2): Suppose that $m \notin  M^{\rm top}(p)$. Since condition (ii) is true, we may choose an $\epsilon \in (0, \max\{r'_{cx}/2 - \dpsr(X_i,m): i =1,\ldots,n\})$ so that
$$\dpsr(X_i, m) < r'_{cx}/2 - \epsilon$$ for all $i = 1,\ldots, n$.
Since $ M^{\rm top}(p)$ is dense in $M(p)$, one can choose $m_\epsilon \in  M^{\rm top}(p)$ such that
$d_M(m, m_\epsilon) < \epsilon$.
Recall that for each $i$ $X_i \in \Sptop$ and we write $\mathcal{F}(X_i) = \mathcal{G}(p)\cdot m_i$ for some $m_i \in \mathcal{F}(X_i) \subset M^{\rm top}(p)$. Then,
\begin{align*}
  \dpsr(X_i, m_\epsilon)  & = \inf_{h \in \mathcal{G}(p)} d_M(h\cdot m_i, m_\epsilon) \\
                          & \le \inf_{h \in \mathcal{G}(p)}d_M(h\cdot m_i, m) + d_M(m, m_\epsilon)\\
                          & = \dpsr(X_i, m) + d_M( m_\epsilon, m)\\
                          & < (r'_{cx}/2 - \epsilon) + \epsilon = r'_{cx}/2.
\end{align*}
Since for $m_\epsilon \in M^{\rm top}(p)$, $\dpsr(X_i, m_\epsilon) < r'_{cx}/2$, Part (1) gives \eqref{eq:condition_unique2}.

The following can be verified similarly: The condition \eqref{eq:condition_unique} of Theorem \ref{thm:uniqueness} is guaranteed if either (i)$^\ast$ or (ii)$^\ast$ below is satisfied. Let $X$ be a random variable following the absolutely continuous distribution $P$ on $\Symp$.
\begin{enumerate}
  \item[(i)$^\ast$] There exists an $S_0 \in \Sptop$ such that $P(\dsr(S_0,X) < r'_{cx}/2) = 1$.
  \item[(ii)$^\ast$]  There exists an $m \in M(p)$ such that $P( \dpsr(X,m) < r'_{cx}/2) = 1$.
\end{enumerate}

We also provide a toy example for the fact: ``For an $S_0 \in \Splwr$, even if a condition $\dsr(S_0,X_i) < \epsilon$ ($i = 1,\ldots,n$) is satisfied for arbitrarily small $\epsilon$, $d_{\mathcal{SR}}(X_i,X_j)$ may be larger than $r'_{cx}$."
Fix $\epsilon >0$. Let $p =2$, $S_0 = I_2$, $X_1 = R(0)\mbox{diag}(e^{\epsilon/2},e^{-{\epsilon/2}})R(0)'$ and $X_2 = R(\pi/4)\mbox{diag}(e^{\epsilon/2},e^{-{\epsilon/2}})R(\pi/4)'$. Then $\dsr(S_0,X_i) = \epsilon/\sqrt{2} < \epsilon$ for $i = 1,2$. However, $\dsr(X_1,X_2) = \sqrt{k}\pi/4 > r'_{cx}$. (For $p = 2$, $r'_{cx} = \sqrt{k}\beta_{\mathcal{G}(p)}/4 =\sqrt{k}\pi/8$.)
%
%

\subsubsection{Proof of Corollary~\ref{cor:uniqueness}}

\begin{proof}[Proof of Corollary~\ref{cor:uniqueness}]
Since $r < r'_{cx}/2$, we have $1= P(\dsr(S_0,X_i) \le r) \le P(\dsr(S_0,X_i) \le r'_{cx}/2)$.
By Remark~\ref{remark:4.9} and Theorem~\ref{thm:uniqueness}  (more precisely, Condition (i)$^\ast$ in Appendix~\ref{remark:4.9proof} is satisfied, which in turn implies that the condition of Theorem~\ref{thm:uniqueness} is satisfied), the PSR mean is unique up to the action of $\mathcal{G}(p)$.   Assertion (ii) is given by Theorem~\ref{thm:avoid_low_strata_SR}.
Theorem~\ref{thm:SRvsPSRequivalence} is applied with assertion (ii) to yield $\Esr = \mathcal{F}(\Epsr)$. By (i), $\Epsr$ is the orbit $\mathcal{G}(p)\cdot (U,D)$ for some $(U,D) \in M(p)$, and $\Esr$ only contains the SPD matrix $\bar{X}:= UDU'$.
\end{proof}

\subsection{Proofs for Section~\ref{asymptotics}} \label{app:proof_of_consistency} \label{sec:wheretoputthese}

\subsubsection{Proofs of Theorem~\ref{thm:consistency} and related results}
\begin{proof}[Proof of Theorem~\ref{thm:consistency}]
We use the following lemma.
\begin{lemma}\label{lem:fnpsr_convergence}
Under the condition of Theorem~\ref{thm:consistency}, the following holds with probability 1: For any $m \in M(p)$ and for any sequence $m_n \in M(p)$ satisfying $d_M(m_n, m) \to 0$ as $n \to \infty$, we have
\begin{equation}\label{eq:fnpsrtofpsrmn}
\lim_{n\to\infty}\fnpsr(m_n) = \fpsr(m),
\end{equation}
and in particular
\begin{equation}\label{eq:fnpsrtofpsr}
  \lim_{n\to\infty}\fnpsr(m) = \fpsr(m).
\end{equation}
\end{lemma}
\begin{proof}[Proof of Lemma~\ref{lem:fnpsr_convergence}]
For any given $m \in M(p)$, since the random variable $\dpsr^2(X,m)$ is integrable (Proposition~\ref{prop:fsr_fpsr is well defined}), we have $P( \lim_{n\to\infty} \fnpsr(m) =\fpsr(m)) = 1$ by the strong law of large numbers. We shall extend this result to
\begin{equation}\label{eq:fnpsrtofpsr_proof}
  P\left( \lim_{n\to\infty}\fnpsr(m) =\fpsr(m)\ \mbox{ for all }\ m \in M(p) \right) = 1,
\end{equation}
thus showing (\ref{eq:fnpsrtofpsr}).

Let $m_1,m_2,\ldots$ be a countable dense sequence in $M(p)$.
Since for each $k$, $\lim_{n\to\infty}\fnpsr(m_k) =\fpsr(m_k)$ almost surely, and $\{m_k\}$ is countable,
\begin{equation}\label{eq:densesequence converge}
  P\left( \lim_{n\to\infty}\fnpsr(m_k) =\fpsr(m_k)\ \mbox{ for all }\ k=1,2,\ldots \right) = 1,
\end{equation}
Moreover, an argument similar to above leads us to conclude that, for every $k$,
\begin{equation}\label{eq:convergence_sequence_dpsraverage}
\frac{1}{n}\sum_{i=1}^n \dpsr(X_i,m_k) \to \int_{\Symp} \dpsr(X,m_k) P(dX) < \infty
\end{equation}
as $n\to\infty$ almost surely.

Observe that by Lemma~\ref{lem:dpsr_continuity}(ii), for any $X\in\Symp$, $m,m' \in M(p)$,
$$|\dpsr^2(X,m) - \dpsr^2(X,m')| \le d_M(m,m')(2\dpsr(X,m) + d_M(m,m')),$$
which in turn leads to
\begin{align}
  |\fnpsr(m) -& \fnpsr(m') | \le \frac{1}{n}\sum_{i=1}^n \left| \dpsr^2(X_i,m) - \dpsr^2(X_i,m')\right|\nonumber\\
   & \le g_n(m,m') := d_M(m,m')\left( \frac{2}{n}\sum_{i=1}^n \dpsr(X_i,m') +  d_M(m,m') \right). \label{eq:inequalityforfnpsr}
\end{align}


Choose an arbitrary $m_0 \in M(p)$. Since $\{m_k\}$ is dense in $M(p)$, we can choose a subsequence $m_{k_i}$ satisfying $d_M(m_{k_i}, m_0) \to 0$ as $i \to\infty$. For each $i$, the inequality (\ref{eq:inequalityforfnpsr}) with $(m,m')$ replaced by $(m_0, m_{k_i})$ is
$$\fnpsr(m_{k_i}) - g_n(m_0,m_{k_i}) \le \fnpsr(m_0) \le \fnpsr(m_{k_i}) + g_n(m_0,m_{k_i}).$$
Taking the limit as $n\to\infty$, by (\ref{eq:densesequence converge}) and (\ref{eq:convergence_sequence_dpsraverage}),
\begin{align*}
\fpsr(m_{k_i}) - g(m_0,m_{k_i}) &\le \liminf_{n\to\infty}   \fnpsr(m_0)\\
                    & \le \limsup_{n\to\infty} \fnpsr(m_0)   \\
                    & \le \fpsr(m_{k_i}) + g(m_0,m_{k_i}),
\end{align*}
where $g(m,m') = d_M(m,m')( 2\int_{\Symp} \dpsr(X,m') P(dX) +  d_M(m,m'))$. Further taking the limit as $i \to \infty$, since $\fpsr$ is continuous (see Lemma~\ref{lem:contfpsr}), we have proven (\ref{eq:fnpsrtofpsr_proof}).

To show (\ref{eq:fnpsrtofpsrmn}), let $m_n$ be a sequence such that $\lim_{n\to\infty}d_M(m_n,m) = 0$ for some $m \in M(p)$. Again from the inequality (\ref{eq:inequalityforfnpsr}), we have
$$\fnpsr(m) + g_n(m_n,m)) \le \fnpsr(m_n) \le \fnpsr(m) + g_n(m_n,m).$$
By (\ref{eq:fnpsrtofpsr_proof}), $\fnpsr(m)$ converges to $\fpsr(m)$ almost surely, while $ g_n(m_n,m) \to 0$ almost surely as well.  This proves (\ref{eq:fnpsrtofpsrmn}).
\end{proof}

We next show that with probability 1
\begin{equation}\label{eq:Ziezold-consistency}
  \cap_{k=1}^\infty \overline{ \cup_{n=k}^\infty \Enpsr} \subset \Epsr.
\end{equation}
We assume $\cap_{k=1}^\infty \overline{ \cup_{n=k}^\infty \Enpsr}$ is non-empty; otherwise, (\ref{eq:Ziezold-consistency}) holds.

Let $\ell = \inf_{m \in M(p)} \fpsr(m)$ and $\ell_n = \inf_{m \in M(p)} \fnpsr(m)$ for $n =1,2,\ldots$.
By (\ref{eq:fnpsrtofpsr}), we have for any $m \in M(p)$ there exists $\epsilon_n \to 0$ such that
$ \fpsr(m) \ge \fnpsr(m) - \epsilon_n \ge \ell_n - \epsilon_n$. Taking the limit superior of both sides, we have
$ \fpsr(m) \ge \limsup_{n\to\infty} \ell_n$. Taking the infimum over $m \in M(p)$, we get
\begin{equation}\label{eq:limsupell}
  \ell = \inf_{m \in M(p)} \fpsr(m) \ge \limsup_{n\to\infty} \ell_n.
\end{equation}
Thus, any subsequential limit of $\ell_n$ is bounded above by $\ell$.

For any $m_0 \in \cap_{k=1}^\infty \overline{ \cup_{n=k}^\infty \Enpsr}$, there exists a subsequence $\{n_k: k=1,2,\ldots\}$ of $1,2,\ldots$ such that $m_{n_k} \in E_{n_k}^{(\mathcal{PSR})}$ and $\lim_{k\to\infty} d_M(m_{n_k}, m_0)=0$.
By (\ref{eq:fnpsrtofpsrmn}),
\begin{equation}\label{eq:applicationtosubseq}
 \ell_{n_k} = f_{n_k}^{(\mathcal{PSR})}(m_{n_k}) \to \fpsr(m_0) \ge \inf_{m \in M(p)}\fpsr(m) = \ell
\end{equation}
as $k \to \infty$ almost surely. In view of (\ref{eq:limsupell}),
$\fpsr(m_0) \le \ell$, which in turn gives $\fpsr(m_0) = \ell$, i.e., $m_0 \in \Epsr$. Since $m_0 \in  \cap_{k=1}^\infty \overline{ \cup_{n=k}^\infty \Enpsr}$ was arbitrary, (\ref{eq:Ziezold-consistency}) is verified.
%

Let $a_n := \sup_{m \in \Enpsr }d_M(m,\Epsr)$. For each $n$, choose $m_n \in \Enpsr$ such that
\begin{equation}\label{eq:choosemn}
  a_n(1-\frac{1}{n}) < d_M(m_n, \Epsr) \le a_n.
\end{equation}

Assume with probability 1 that the event (\ref{eq:Ziezold-consistency}) is occurred. Then, every accumulation point of $m_n$ lies in $\Epsr$. Thus, either $a_n \to 0$ or there is no accumulation point (equivalently, $a_n \to \infty$ as $n\to\infty$).
 We will rule out the case $a_n \to \infty$ by contradiction.

Suppose that $\lim_{n\to\infty} a_n = \infty$.
Then, for any choice $m_0 \in \Epsr$, we have by (\ref{eq:choosemn}) $d_M(m_n,m_0) \to \infty$.

For $r>0$ let $K_r = \{ S \in \Symp: \mbox{every eigenvalue $\lambda$ of $S$ satisfies $|\log\lambda| \le r$}\}$. Choose $r_0$ large enough so that $P(K_{r_0}) > 0$ and $m_0 \in \mathcal{F}^{-1}(K_{r_0})$; such $r_0$ exists since $\cup_{r > 0} K_r = \Symp$.
Then for any $X, Y \in K_{r_0}$,
\begin{equation}\label{eq:compactr1}
  \sup_{m \in \mathcal{F}^{-1}(X), m' \in \mathcal{F}^{-1}(Y)}d_M(m,m') \le (k{\rm diag}(SO(p))^2 + 4pr^2 )^{1/2} =: C(r_0).
\end{equation}

We now claim that for any $m_n$ with $d_M(m_n,m_0) \to \infty$, there exists $M_n \to \infty$ satisfying
\begin{equation}\label{eq:coercivity}
  \dpsr(X,m_n) \ge M_n,\quad \mbox{for any $X \in K_{r_0}$}
\end{equation}
and $\lim_{n\to\infty}M_n = \infty$.
To verify, for any $X \in \Symp$ and $m_n \in M(p)$, let $m_X^{(n)} \in \mathcal{F}^{-1}(X)$ satisfy
$$d_M(m_X^{(n)}, m_n) = \min_{m \in \mathcal{F}^{-1}(X)} d_M(m,m_n) =  \inf_{m \in \mathcal{F}^{-1}(X)} d_M(m,m_n)= \dpsr(X,m_n).$$
By the triangle inequality, and by (\ref{eq:compactr1}), for any $X \in K_{r_0}$ we have
\begin{align*}
  d_M(m_n,m_0) & \le d_M(m_n, m_X^{(n)}) +  d_M( m_X^{(n)}, m_X^{(0)}) + d_M(m_X^{(0)},m_0) \\
              & = \dpsr(X, m_n) + d_M( m_X^{(n)}, m_X^{(0)}) + \dpsr(X,m_0) \\
              &\le \dpsr(X, m_n) + 2C(r_0).
\end{align*}
In particular, $\dpsr(X,m_n) \ge d_M(m_n,m_0) - 2C(r_0)$. Taking $M_n = d_M(m_n,m_0) - 2C(r_0)$, (\ref{eq:coercivity}) is verified.

For each $n$, choose a subsequence $n_1,\ldots, n_{k(n)}$ of $1,2,\ldots,n$ so that $X_{n_j} \in K_{r_0}$ for $j = 1,\ldots, k(n)$. Then by the strong law of large numbers, $\lim_{n\to\infty}k(n)/n  = P(K_{r_0})>0$. This fact, together with (\ref{eq:coercivity}), gives
$$
\fnpsr(m_n) = \frac{1}{n}\sum_{i=1}^n \dpsr^2(X_i,m_n) \ge \frac{1}{n}\sum_{j=1}^{k(n)} \dpsr^2(X_{n_j},m_n) \ge \frac{k(n)}{n}M_n^2 \to \infty.$$
However, $\fnpsr(m_n) = \inf_{m \in M(p)} \fnpsr(m) \le \fnpsr(m_0)$, and $\fnpsr(m_n) \to \infty$ while $\fnpsr(m_0) \to \fpsr(m_0) <\infty$ by (\ref{eq:fnpsrtofpsr}), yielding a contradiction. Thus $a_n \to 0$ under the probability one event satisfying (\ref{eq:Ziezold-consistency}).
\end{proof}

The following lemma shows that the conclusion of Theorem~\ref{thm:consistency} is equivalent to the strong consistency of \cite{Huckemann2011b} in the sense of \cite{Bhatt2003}.
\begin{lemma}\label{lem:BPconsistency-vs-sup}
Let $(M,d)$ be a metric space and let $E, E_1,E_2,\ldots$ are non-empty sets in $M$. We have $\lim_{n \to \infty} \sup_{ m \in E_n }d( m , E) = 0$ if and only if, for any $\epsilon >0$, there exists $N(\epsilon)$ such that $\cup_{n \ge N(\epsilon)} E_n \subset \{m \in M: d(E,m) \le \epsilon\}$.
\end{lemma}
\begin{proof}[Proof of Lemma~\ref{lem:BPconsistency-vs-sup}]
By definition, $\lim_{n \to \infty} \sup_{ m \in E_n }d_M( m , E) = 0$ is equivalent to the statement that  for any $\epsilon>0$, there exists $N(\epsilon)$ such that for all $n \ge N(\epsilon)$, $\sup_{ m \in E_n }d_M( m , E) \le \epsilon$. If $\sup_{ m \in E_n }d_M( m , E) \le \epsilon$ for all $n \ge N(\epsilon)$, then for any $m \in \cup_{n \ge N(\epsilon)} E_n$, $m \in E_n$ for some $n \ge N(\epsilon)$, and $d(m,E) \le \sup_{ m \in E_n }d_M( m , E) \le \epsilon$, which gives $m \in \overline{B}_\epsilon(E) :=\{m \in M: d(M,m) \le \epsilon\}$. On the other hand, if $N(\epsilon)$ is such that
 $\cup_{n \ge N(\epsilon)} E_n \subset \overline{B}_\epsilon(E)$, then for any $n \ge N(\epsilon)$ and
 for any $m \in E_n$, $d(m,E) \le \epsilon$. Thus, $\sup_{m\in E_n} d(m,E) \le \epsilon$ as well.
\end{proof}

\begin{proof}[Proof of Corollary~\ref{cor:Hausdorff}]
By Theorem~\ref{thm:consistency} there exists a probability 1 event $A$ in which  $d_M(m_n, \Epsr) \to 0$ as $n\to\infty$ for any sequence $\{m_n \in \Enpsr\}$.
Assume the event $A$ has occurred. Choose a sequence $\{m_n \in \Enpsr\}$.
%
Since $\Epsr = \mathcal{G}(p) \cdot \mu$, there exists a sequence $\{h_n \in \mathcal{G}(p)\}$ such that
$d_M(m_n, h_n \cdot \mu) \to 0$ as $n\to \infty$.

Now fix an arbitrary element $m_0 \in \Epsr$, and for each $n$ let $h_n' \in \mathcal{G}(p)$ be such that $m_0 = h_n' \cdot h_n \cdot \mu$. Then, as $n\to \infty$
$$d_M(h_n' \cdot m_n, m_0)  = d_M(h_n' \cdot m_n, h_n' \cdot h_n \cdot \mu) = d_M(m_n, h_n \cdot \mu) \to 0.$$
Since $h_n' \cdot m_n \in \Enpsr$ as well as $m_n\in \Enpsr$ (as observed in (\ref{sample_PSR_mean_nonunique})), we have
$$
0 \le d_M(\Enpsr, m_0) = \inf_{m \in \Enpsr} d_M(m, m_0) \le d_M(h_n' \cdot m_n, m_0),$$
and in particular $d_M(\Enpsr, m_0) \to 0$ as $n\to \infty$. Since $m_0 \in \Epsr$ was arbitrary, we have shown (\ref{eq:cor:Hausdorff1}).
Assertion (\ref{eq:cor:Hausdorff2}) is verified by the conclusion of Theorem~\ref{thm:consistency} and (\ref{eq:cor:Hausdorff1}).
\end{proof}

\begin{proof}[Proof of Corollary~\ref{cor:consistency}]
(i) The hypotheses of the corollary ensure that Theorem~\ref{thm:consistency} applies, so the event that for every $\epsilon >0$, there exists $N(\epsilon)$ such that for all $n \ge N(\epsilon)$
\begin{equation}\label{eq:cor:consist_proof}
 \sup_{m \in \Enpsr} d_M(m, \Epsr) \le \epsilon,
\end{equation}
occurs with probability 1. Assume this event has occurred.
Fix an $\epsilon>0$, and let $N = N(\epsilon)$.
For each $n \ge N$, choose $S_n \in \mathcal{F}(\Enpsr)$. Then one can choose $m_n \in \mathcal{F}^{-1}(S_n)$ such that $m_n \in \Enpsr$.
Likewise for an arbitrary $S_0 \in \mathcal{F}(\Epsr)$, let $m_0 \in M(p)$ satisfy $m_0 \in \mathcal{F}^{-1}(S_0) \cap \Epsr$.
Thus,
\begin{align*}
  \dsr(S_n, \mathcal{F}(\Epsr)) & = \inf_{S_0\in \mathcal{F}(\Epsr)} \dsr(S_n,S_0) \\
   & = \inf_{m_0\in \Epsr} \dsr(S_n,\mathcal{F}(m_0))\\
   & \le \inf_{m_0\in \Epsr} \dpsr(S_n, m_0) \quad (\mbox{by (\ref{SR_PSR_ineq})})\\
   & \le \inf_{m_0\in \Epsr} d_M(m_n, m_0)
    = d_M(m_n, \Epsr)
    \le \epsilon.
\end{align*}
The last inequality holds since we have assumed the event (\ref{eq:cor:consist_proof}) has occurred.
Since $S_n \in \mathcal{F}(\Enpsr)$ was arbitrary, $\lim_{n\to\infty} \sup_{S \in \mathcal{F}(\Enpsr)} \dsr(S, \mathcal{F}(\Epsr)) = 0$. Hence the statement that this limit equals 0 is a probability-one event.

(ii) and (iii). Since the probability measure $P$ has finite PSR-variance, the condition $\Esr \subset \Sptop$ implies that $\mathcal{F}(\Epsr) = \Esr$ (by Theorem \ref{thm:SRvsPSRequivalence}). Conclusion (i) then implies conclusion (ii), which in turn implies conclusion (iii).
\end{proof}

\subsubsection{Proof of Theorem \ref{thm:PSR_mean_CLT}}\label{sec:proof_of_thm:PSR_mean_CLT}

\begin{proof}[Proof of Theorem \ref{thm:PSR_mean_CLT}]

By Assumption (A2), $\Epsr = \mathcal{G}(p) \cdot m$ for some $m' \in \Epsr$, which implies that for any $m,m' \in \Epsr$ and any $S \in \Symp$, $\dpsr(S,m) = \dpsr(S,m')$. Therefore, in the presence of Assumption (A2), Assumption (A3) implies that $P(\dpsr(X,m_0) < r'_{cx}) =1$ for any $m_0 \in \Epsr$. Let $m_0 \in \Epsr$ be given.

Let $A_1$ be the event that $\Enpsr$ is unique up to the action of $\mathcal{G}(p)$ for all $n$, and let $A_2$ be the event that $\dpsr(X_i, m_0) < r'_{cx}$ and $X_i \in \Sptop$ for all $i \in \mathbb{N}$. Assumption (A2) implies that $P(A_1) = 1$.  Assumptions (A1) and (A3) imply that $P(A_2) = 1$ as well. In the rest of this proof, we assume that the probability 1 event $A_1\cap A_2$ has occurred.

For each $n$, let $m_n' \in \Enpsr$ be arbitrary. Since the event $A_1$ has occurred, $\Enpsr = \mathcal{G}(p) \cdot m_n'$. Let $m_n$ be any minimizer of the function $d_M(\cdot, m_0)$ over $\mathcal{G}(p) \cdot m_n'$. Thus, by definition, $m_n$ is a sample PSR mean. We first show that such an $m_n$ is unique.

Let ${\rm supp}_1(P) = \{S \in \Symp: \dpsr(S, m_0) < r'_{cx}\} \cap \Sptop$, so that $X_i \in {\rm supp}_1(P)$ for all $i \in \mathbb{N}$ since the event $A_1\cap A_2$ is occurred.
For any $S \in {\rm supp}_1(P)$, one can choose a unique $\tilde{m} := \tilde{m}(m_0; S) \in \Fc^{-1}(S)$ such that $d_M(\tilde{m},m_0) = \dpsr(S, m_0) < \min_{h \in \mathcal{G}(p) \setminus \{I_p\}} d_M( h \cdot \tilde{m}, m_0)$. Note that  $d_M(\tilde{m},m_0) < r'_{cx}$ since $S \in {\rm supp}_1(P)$.
To verify that such a choice is indeed unique, let $m' = h \cdot \tilde{m}$ for some $h \in \mathcal{G}(p) \setminus \{I_p\}$. An application of Lemma~\ref{eig_decomp_dist_bound}(c), together with the triangle inequality, gives $d(m',m_0) > 3r_{cx}'$, and thus the minimizer $m$ is unique.

For each $i = 1,\ldots,n$, $X_i \in {\rm supp}_1(P)$, so we can set $m_{_{X_i}} = \tilde{m}(m_0; X_i)$ to be the unique eigen-decomposition of $X_i$, closest to $m_0$. Then, the sample PSR mean objective function can be written as
\begin{equation}\label{eq:PSRtoM}
  \fnpsr(m) = \frac{1}{n}\sum_{i=1}^n \dpsr^2(X_i,m) = \frac{1}{n}\sum_{i=1}^n d_M^2(m_{_{X_i}},m),
\end{equation}
which is exactly the Fr\'{e}chet objective function on $(M(p),d_M)$ with data-points $m_{_{X_1}}, \dots, m_{_{X_n}}$.
Since $d_M(m_{_{X_i}}, m_0) <  {r'_{cx}}$ for all $i$, Lemma \ref{lem:geom} implies that the Fr\'{e}chet mean
$\bar{m}_n$ of $\{m_{_{X_1}}, \dots, m_{_{X_n}}\}$
is unique and satisfies $d_M(\bar{m}_n, m_0) <  {r'_{cx}}$. Moreover, by (\ref{eq:PSRtoM}), $\bar{m}_n \in \Enpsr$. Since $\Fc(\bar{m}_n) = \Fc(\Enpsr) \in {\rm supp}_1(P)$ as well, $\bar{m}_n$ is the unique minimizer of $\min_{m \in \mathcal{G}(p)\cdot \bar{m}_n} d_M(m,m_0)$. Thus, $m_n = \bar{m}_n$ is determined uniquely (if the probability 1 event $A_1 \cap A_2$ has occurred).

(a) By Corollary~\ref{cor:Hausdorff}, $\lim_{n \to \infty} d_H( \Enpsr , \Epsr) = 0$  with probability 1. Therefore, with probability 1 the sequence $\{m_n \in \Enpsr\}$ chosen above satisfies,
\begin{equation}\label{eq:clt_step_1_consistency}
  \lim_{n\to\infty}d_M(m_n, m_0) = 0.
\end{equation}

(b)
For the sample $X_1,\ldots,X_n$, the PSR mean $m_n$ minimizes $\fnpsr(\cdot)$. Thus,
for any neighborhood $V \subset B_{r'_{cx}}(0) \subset \Real^d$ containing zero,
$x_n$ is a minimizer of the function
$g_n: V \to [0,\infty)$ defined by
\begin{equation}\label{eq:clt_step_2_consistency}
  g_n(x) = \sum_{i=1}^n \dpsr^2(X_i, \phi^{-1}_{m_0}(x)),
\end{equation}
 where $\phi_{m_0}(\cdot) = {\rm vec}\circ \tilde\varphi_{m_0}(\cdot)$  (see \eqref{eq:local_chart} and \eqref{eq:vectorization}).
 Note that (\ref{eq:vectorization}) implies that for any $x \in B_{r'_{cx}}(0)$, $\phi^{-1}_{m_0}(x) \in B_{r'_{cx}}^{d_M}(m_0)$.

We now establish that for each $S \in {\rm supp}_1(P)$, there exists a unique $m_S \in \mathcal{F}^{-1}(S)$ such that
\begin{equation}\label{eq:53.5}
  \dpsr^2(S, \phi_{m_0}^{-1}(x)) =  d_M^2(m_S, \phi_{m_0}^{-1}(x))
\end{equation}
 for all $x \in  B_{r'_{cx}}(0)$.
To verify this, suppose that $S \in {\rm supp}_1(P)$ satisfies \eqref{eq:53.5} and let $m_S \in \Fc^{-1}(S)$ be the unique point at which $\min_{m \in \mathcal{F}^{-1}(S)} d_M(m, m_0)$ is achieved.
%
%
By the triangle inequality,
$$
d_M(m_S, \phi_{m_0}^{-1}(x)) \leq d_M(m_S, m_0) + d_M(m_0, \phi_{m_0}^{-1}(x)) < 2r'_{cx}.$$
For $m' \in \mathcal{F}^{-1}(S)$ such that $m' \neq m_S$, we have
$d_M(m',m_S) \ge 4r'_{cx}$ by Lemma \ref{eig_decomp_dist_bound}(c). Again by the triangle inequality,
$$d_M(m', \phi_{m_0}^{-1}(x)) \geq d_M(m',m_S) - d_M(m_S,\phi_{m_0}^{-1}(x)) > 2r'_{cx} > d_M(m_S, \phi_{m_0}^{-1}(x)).$$
Thus, $m_S$ is the unique element of $\mathcal{F}^{-1}(S)$ satisfying
\begin{equation}\label{eq:Xtom_X}
  \dpsr(S, \phi_{m_0}^{-1}(x)) =
     \inf_{m \in \mathcal{F}^{-1}(S)} d_M(m, \phi_{m_0}^{-1}(x)) = d_M(m_S, \phi_{m_0}^{-1}(x)),
\end{equation}
as asserted.

 Next, recall that for each $i$, $m_{_{X_i}} = \tilde{m}(m_0; X_i)$ is the eigen-decomposition of $X_i$ closest to $m_0$, and let $x \in B_{r'_{cx}}(0)$ be arbitrary.
Using (\ref{eq:Xtom_X}), we rewrite (\ref{eq:clt_step_2_consistency}) as
$$ g_n(x) = \sum_{i=1}^n d_M^2(m_{_{X_i}}, \phi^{-1}_{m_0}(x)),$$
where for every $i$,
$m_{_{X_i}} \in B_{r'_{cx}}^{d_M}(m_0)$, a ball that also contains $\phi^{-1}_{m_0}(x)$.
We shall now discuss the consequence of the bounded support $B_{r'_{cx}}^{d_M}(m_0)$.
It is well known that $({\rm Diag}^+(p), g_{\mathcal{D}^+})$ has non-positive sectional curvature and infinite injectivity radius, and $(SO(p), k g_{SO})$ has non-negative sectional curvature (bounded above by $\Delta(SO(p), k g_{SO}) = 1/(4k)$ and injectivity radius $r_{\rm inj}(SO(p), k g_{SO}) = \sqrt{k}\pi$. Thus, for the product Riemannian manifold $(M(p), g_M)$, it follows that  $r_{\rm inj} := r_{\rm inj}(M, g_M) = r_{\rm inj}(SO(p), k g_{SO})$,
$\Delta := \Delta(M, g_M) = \Delta(SO(p), k g_{SO})$,
and that
the radius $r'_{cx}$ of the ball $B_{r'_{cx}}^{d_M}(m_0)$ satisfies
\begin{equation}\label{eq:radius}
   r'_{cx} = \frac{\sqrt{k}\beta_{\mathcal{G}(p)}}{4} \le
      \frac{\sqrt{k}\pi}{8} = \frac{1}{2}\min\{ r_{\rm inj}, \frac{\pi}{2\sqrt{\Delta}}\} = \frac{\sqrt{k}\pi}{2},
\end{equation}
where the first inequality follows from Lemma~\ref{eig_decomp_dist_bound}(b). (The right-hand side of (\ref{eq:radius}) equals the convexity radius of $(M, g_M)$.)

By \cite{Afsari2011} and \cite{afsari2013convergence}, the inequality (\ref{eq:radius}) ensures that
(i) the open ball $B_{r'_{cx}}^{d_M}(m_0)$ is \emph{strongly convex}\footnote{A set $B$ in $(M,g)$ is strongly convex if any two points in $B$ can be connected by a unique minimal-length geodesic in $M$ and the geodesic segment entirely lies in $B$.} in $M(p)$;
(ii) for any $m \in B_{r'_{cx}}^{d_M}(m_0)$, the function $d_M^2(m,\cdot)$ is a $C^\infty$ function in $B_{r'_{cx}}^{d_M}(m_0)$ (since the cut locus of $m$ does not intersect $B_{r'_{cx}}^{d_M}(m_0$)), which in turn implies that
$g_n$ is $C^\infty$
; and
(iii) for any $m_1,\ldots,m_n \in B_{r'_{cx}}^{d_M}(m_0)$, the function $\sum_{i=1}^n d_M^2(m_i, \cdot)$ (restricted to $B_{r'_{cx}}^{d_M}(m_0)$) is convex (strictly convex if at least two $m_i$'s are distinct). In particular, when $X_i$'s are sampled from an absolutely continuous distribution $P$ (as assumed in (A1)),  with probability 1 the Hessian matrix of $d_M^2(m_{_{X_i}}, \phi^{-1}_{m'}(x))$ at $x =0$, for arbitrary $m' \in  B_{r'_{cx}}^{d_M}(m_0)$, is well-defined and positive definite.
Furthermore, thanks to the identification \eqref{eq:53.5}, we are assured that with probability 1, for any $i = 1,\ldots,n$, the function $h_{X_i}(\cdot):=\dpsr^2(X_i, \phi_{m_0}^{-1}(\cdot)) = d_M^2(m_{X_i}, \phi_{m_0}^{-1}(\cdot))$ is $C^\infty$.

For all $n\geq 0$, now let $x_n = \phi_{m_0}(m_n)  \in \Real^d$ ; note that $x_0=0$.  Observe that (\ref{eq:clt_step_1_consistency}) implies that as $n\to\infty$, $\phi_{m_0}(m_n)\to \phi_{m_0}(m_0)$, i.e. that $x_n\to 0$.
By Theorem 2.1 of \cite{Afsari2011} (or, equivalently by Theorem 2.6 of \cite{afsari2013convergence}), the gradient vector field ${\rm grad}_x\, g_n(x)$ has a unique zero in $B_{r'_{cx}}(0)$, and the location of this zero is $x_n$. 

By the Mean Value Theorem applied to each component of ${\rm grad}_x\, g_n$,
$$0 = n^{-1/2}{\rm grad}_x\, g_n(x_n) = n^{-1/2}{\rm grad}_x\, g_n(0) + n^{-1}\mathbf{H} g_n( t_n) \cdot (\sqrt{n} x_n),$$
where the $j$th coordinate of $t_n$ is $t_j   x_{n,j}$ for suitable $t_j \in [0,1]$, where $x_{n,j}$ is the $j$th coordinate of $x_n \in \Real^d$.

Let $x \in B_{r'_{cx}}(0)$ be arbitrary. Since $X_1,X_2,\ldots$ are i.i.d. with bounded support and $\dpsr^2(X_i, \phi^{-1}_{m_0}(\cdot))$ is $C^\infty$, the random vectors ${\rm grad}_x\, \dpsr^2(X_i, \phi^{-1}_{m_0}(x))$  ($i = 1,\ldots, n$) are i.i.d. and bounded. This fact leads to
\begin{align*}
   \int_{\Symp}\frac{\partial}{\partial x_i}\dpsr^2(X, \phi^{-1}_{m_0}(x)) P(dX) & = \frac{\partial}{\partial x_i} \int_{\Symp} \dpsr^2(X, \phi^{-1}_{m_0}(x))P(dX)
 \end{align*}
 (which equals 0 at $x = 0$ as $m_0 \in \Epsr$), and thus $E({\rm grad}_x\, \dpsr^2(X, \phi^{-1}_{m_0}(0))) = 0$. Moreover, since the product of any two entries of ${\rm grad}_x\, \dpsr^2(X_i, \phi^{-1}_{m_0}(x))$ is bounded as well, $\Sigma_P := {\rm Cov}({\rm grad}_x\, \dpsr^2(X, \phi^{-1}_{m_0}(0)))$ exists.

 Since the first two moments of ${\rm grad}_x\, \dpsr^2(X_i, \phi^{-1}_{m_0}(0))$ exist,  the multivariate classical  central limit theorem  \citep[\emph{cf}. ][]{anderson1958introduction} implies that
 as $n\to\infty$,
$$n^{-1/2}{\rm grad}_x\, g_n(0) = \frac{\sqrt{n}}{n}\sum_{i=1}^n {\rm grad}_x\, \dpsr^2(X_i, \phi^{-1}_{m_0}(0))$$
 weakly converges to $N_d(0, \Sigma_P)$. 


 Likewise,  the continuity of $\mathbf{H} g_n(\cdot)$   ensures that each entry of the matrix
$$H_P(x) := E \left( \mathbf{H}\dpsr^2(X, \phi^{-1}_{m_0}( x ))\right)$$ exists.
Since, with probability 1, $t_n \to 0$ (because $x_n \to 0$) and since  $\mathbf{H} g_n(\cdot) = \sum_{i=1}^n \mathbf{H} h_{X_i}(\cdot) $ is continuous, the law of large numbers implies that  $n^{-1}\mathbf{H} g_n(t_n)$ converges in probability to $H_P := H_P(0)$ as $n\to\infty$. Recall that, with probability 1, the function $h_{X}(\cdot)$ is $C^\infty$ and strictly convex on $B_{r'_{cx}}(0)$, and thus for any $x\in  B_{r'_{cx}}(0)$, both $\mathbf{H}\dpsr^2(X, \phi^{-1}_{m_0}( x ))$ and $\mathbf{H} g_n(x) = \sum_{i=1}^n  \mathbf{H} h_{X_i}(x)$ are positive definite almost surely. Therefore, $H_P = H_P(0)$ is invertible, and so is $\mathbf{H} g_n(t_n)$ almost surely.
Thus, by Slutsky's theorem, $\sqrt{n} x_n = (n^{-1} \mathbf{H} g_n (t_n))^{-1} \times n^{-1/2} {\rm grad}_x\, g_n(0)$ converges in distribution to $N_d(0, H_P^{-1}\Sigma_{P} H_P^{-1})$.
\end{proof}

\section{Additional numerical results}\label{sec:app_QQplot}
  As referenced in Section~\ref{sec:tbm}, we plot the quantiles of the log-eigenvalues and rotation angles of the linearized PSR means against the quantiles of the standard normal distribution for each group in \autoref{fig:Bootstrap_PSR_Mean_QQ_plots}, as a visual check of normality of the linearized PSR sampling distributions. With the exception of the tails, the normal QQ plots remain within the 95\% confidence envelope, despite the small sample sizes $(n_i = 19, 17)$.

\begin{figure}[pt]
\centering
\includegraphics[width = 1\textwidth]{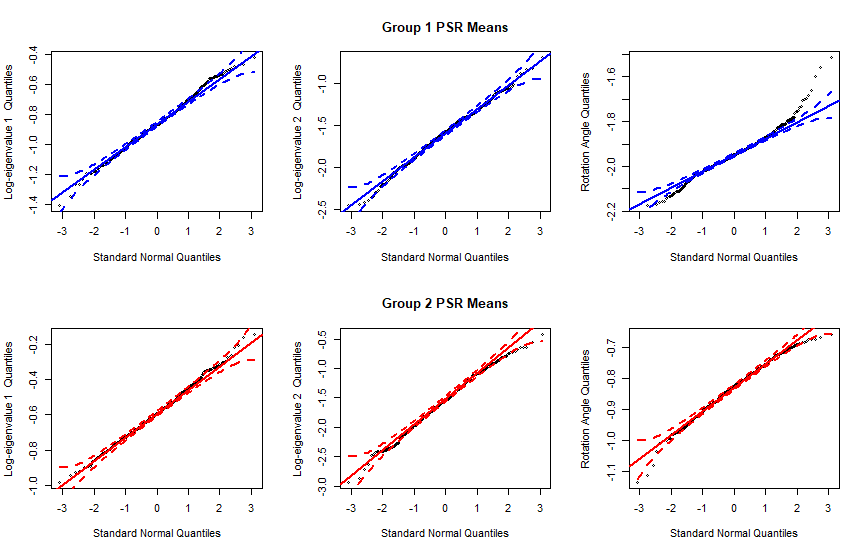}
\caption{Quantile-quantile plots of the log-eigenvalues and angles of the bootstrap PSR means versus standard normal quantiles with 95\% pointwise confidence envelopes represented by dashed lines. }
\label{fig:Bootstrap_PSR_Mean_QQ_plots}
\end{figure}

\end{appendix}
%
%

 \section*{Acknowledgements}
%
The first author was supported by the National Research Foundation of Korea (NRF) grant funded by the Korea government (MSIT) (No. 2019R1A2C2002256).
%



\bibliographystyle{imsart-nameyear} 
\bibliography{scarotmean_2023-06-07_arXiv}       

\end{document}